\documentclass{article}

\usepackage[final, nonatbib]{neurips_2023}

\usepackage[utf8]{inputenc} 
\usepackage[T1]{fontenc}    
\usepackage{hyperref}       
\usepackage{url}            
\usepackage{booktabs}       
\usepackage{amsfonts}       
\usepackage{nicefrac}       
\usepackage{microtype}      
\usepackage{graphicx}
\usepackage{todonotes}

\usepackage{amsmath}
\usepackage{amssymb}
\usepackage{mathtools}
\usepackage{amsthm}

\usepackage[capitalize,noabbrev]{cleveref}

\usepackage{neurips_paper}
\usepackage{subcaption}
\newcounter{num}
\usepackage{multirow}

\usepackage{stackengine}

\newcommand{\tildecircume}{\stackon[2pt]{\^{e}}{\~{}}}

\makeatletter

\newcommand{\TODO}[1]{\textcolor{red}{[TODO\@ifnotempty{#1}{: #1}]}}
\newcommand{\sandeep}[1]{\textcolor{orange}{[sandeep\@ifnotempty{#1}{: #1}]}}
\newcommand{\ali}[1]{\textcolor{purple}{[Ali\@ifnotempty{#1}{: #1}]}}
\makeatother

\usepackage{algorithm}
\usepackage{algpseudocode}
\usepackage{wrapfig}

\DeclareMathOperator*{\median}{\mathsf{median}}

\title{Improved Frequency Estimation Algorithms with and without Predictions}

\author{Anders Aamand \\ MIT \\ \texttt{aamand@mit.edu}
\And Justin Y.\ Chen \\ MIT \\ \texttt{justc@mit.edu}
\And Huy L\^{e} Nguy\tildecircume n\\ Northeastern University \\ \texttt{hu.nguyen@northeastern.edu}
\And Sandeep Silwal\\ MIT  \\ \texttt{silwal@mit.edu} \And Ali Vakilian\\ TTIC\\ \texttt{vakilian@ttic.edu}  
}

\begin{document}

\maketitle

\begin{abstract}
Estimating frequencies of elements appearing in a data stream is a key task in large-scale data analysis.
Popular sketching approaches to this problem (e.g., CountMin and CountSketch) come with worst-case guarantees that probabilistically bound the error of the estimated frequencies for any possible input. The work of Hsu et al.~(2019) introduced the idea of using machine learning to tailor sketching algorithms to the specific data distribution they are being run on. In particular, their learning-augmented frequency estimation algorithm uses a learned heavy-hitter oracle which predicts which elements will appear many times in the stream.
We give a novel algorithm, which in some parameter regimes, already theoretically outperforms the learning based algorithm of Hsu et al.~\emph{without} the use of any predictions. Augmenting our algorithm with heavy-hitter predictions further reduces the error and improves upon the state of the art. Empirically, our algorithms achieve superior performance in all experiments compared to prior approaches.

\end{abstract}
\section{Introduction}
In frequency estimation, we  stream a sequence of elements from $[n] := \{1, \ldots, n\}$, and the goal is to estimate $f_i$, the frequency of the $i$th element, at the end of the stream using low-space. Frequency estimation is one of the central problems in data streaming with a wide range of applications from networking (gathering important monitoring statistics \cite{estan2003new, yu2013software, liu2016one}) to machine learning (NLP \cite{goyal2012sketch}, feature selection \cite{aghazadeh2018mission}, semi supervised learning \cite{talukdar2014scaling}). 
CountMin (CM) \cite{cormode2005improved} and CountSketch (CS) \cite{charikar2002finding} are arguably the most popular and versatile of the algorithms for frequency estimation, and are implemented in many popular packages such as Spark \cite{zaharia2016apache}, Twitter Algebird \cite{algebird}, and Redis.

Standard approaches to frequency estimation are designed to perform well in the worst-case due to the multitudinous benefits of worst-case guarantees. However, algorithms designed to handle any possible input do not exploit special structure of the particular distribution of inputs they are used for. In practice, these patterns can be described by domain experts or learned from historical data. Following the burgeoning trend of combining machine learning and classical algorithm design, \cite{HsuIKV19} initiated the study of \emph{learning-augmented} frequency estimation by extending the classical CM and CS algorithms in a simple but effective manner via a heavy-hitters oracle. During a training phase, they construct a classifier (e.g. a neural network) to detect whether an element $i$ is “heavy”
(e.g., whether $f_i$ is among the most frequent items). After such a classifier is trained, they scan the input stream, and apply the classifier to each element $i$. If the element is predicted to be heavy, it is allocated a unique bucket, so that an exact value of $f_i$ is computed. Otherwise, the stream element is inputted into the standard sketching algorithms.

The advantage of their algorithm was analyzed under the assumption that the true frequencies follow a heavy-tailed Zipfian distribution. This is a common and natural reoccurring pattern in real world data where there are a few very frequent elements and many infrequent elements. Experimentally, \cite{HsuIKV19} showed several real datasets where the Zipfian assumption (approximately) held and useful heavy-hitter oracles could be trained in practice. Our paper is motivated by the following natural questions and goals in light of prior works:

\begin{quote}
     \textit{Can we design better frequency estimation algorithms (with and without predictions) for heavy-tailed distributions?}
\end{quote}
In particular, we consider the setting of \cite{HsuIKV19} where the underlying data follow a heavy-tailed distribution and investigate whether sketching algorithms can be further tailored for such distributions. Before tackling this question, we must tightly characterize the benefits--and limitations--of these existing methods, which is another goal of our paper:
\begin{quote}
     \textit{Give tight error guarantees for CountMin and CountSketch, as well as their learning-augmented variants, on Zipfian data.}
\end{quote}

Lastly, any algorithms we design must possess worst case bounds in the case that either the data does not match our Zipfian (or more generally, heavy-tailed) assumption or the learned predictions have high error, leading to the following `best of both worlds' goal:
\begin{quote}
     \textit{Design algorithms which exploit heavy tailed distributions and ML predictions but also maintain worst-case guarantees.}
\end{quote}

We addresses these challenges and goals and our contributions can be summarized as follows:
\begin{itemize}[leftmargin=*]
    \item  We give tight upper and lower bounds for CM and CS, with and without predictions, for heavy tailed distributions. A surprising conclusion from our analysis is that (for a natural error metric) a constant number of rows is optimal for both CM and CS. 
    In addition, our theoretical analysis shows that CS outperforms CM, both with and without predictions, validating the experimental results of \cite{HsuIKV19}.
    \item We go beyond CM and CS based algorithms to give a better frequency estimation algorithm for heavy tailed distributions, with and without the use of predictions. We show that our algorithms can deliver up to a logarithmic factor improvement in the error bound over CS and its learned variant. In addition, our algorithm has worst case guarantees.
    \item Prior learned approaches require querying an oracle for every element in the stream. In contrast, we obtain a \emph{parsimonious} version of our algorithm which only requires a limited  number of queries to the oracle. The number of queries we use is approximately equal to the given space budget.
    \item Lastly, we evaluate our algorithms on two real-world datasets with and without ML based predictions and show superior empirical performance compared to prior work in all cases.
\end{itemize}

\subsection{Preliminaries}
\paragraph{Notation and Estimation Error}
The stream updates an $n$ dimensional frequency vector and every stream element is of the form $(i, \Delta)$ where $i \in [n]$ and $\Delta \in \R$ denotes the update on the coordinate. The final frequency vector is denoted as $f \in \R^n$. Let $N = \sum_{i \in [n]} f_i$ denote the sum of all frequencies. To simplify notation, we assume that $f_{1}\ge f_{2}\ge\ldots\ge f_{n}$. $\tilde{f}_i$ denotes the estimate of the frequency $f_i$. Given estimates $\{ \tilde{f}_i\}_{i \in [n]}$, the error of a particular frequency is $|\tilde{f}_i - f_i|$. We also consider the following notion of overall weighted error as done in \cite{HsuIKV19}:
\begin{equation}\label{eq:weighted_error}
    \text{Weighted Error:} = \frac{1}{N} \sum_{i \in [n]} f_i \cdot | \tilde{f}_i - f_i|.
\end{equation}
The weighted error can be interpreted as measuring the error with respect to a query distribution which is the same as the actual frequency distribution. As stated in \cite{HsuIKV19}, theoretical guarantees of frequency estimation algorithms are typically phrased in
the traditional $(\varepsilon, \delta)$-error formulations. However as argued in there, the simple weighted objective \eqref{eq:weighted_error} is a more holistic measure and does not require tuning of two different parameters, and is thus more natural from an ML perspective.

\paragraph{Zipfian Stream} We also work under the common assumption that the frequencies follow the Zipfian law, i.e., the $i$th largest frequency $f_i$ is equal to $A/i$ for some parameter $A$. Note we know $A$ at the end of the stream since the stream length is $A\cdot H_n$. By rescaling, we may assume that $A=1$ without loss of generality. We will make this assumption throughout the paper.

\paragraph{CountMin (CM)}
For parameters $k$ and $B$, which determine the total space used, CM uses $k$ independent and uniformly random hash functions $h_1,\dots,h_k:[n] \to [B]$. Letting $C$ be an array of size $[k] \times [B]$ we let $C[\ell,b]=\sum_{j\in [n]}[h_\ell(j)=b]f_j$. When querying $i\in [n]$ the algorithm returns $\tilde f_i=\min_{\ell \in [k]} C[\ell,h_\ell(i)]$. Note that we always have that $\tilde f_i \geq f_i$. 

\paragraph{CountSketch (CS)}
In CS, we again have the hash functions $h_i$ as above as well as sign functions $s_1,\dots,s_k:[n] \to  \{-1,1\}$. The array $C$ of size $[k] \times [B]$ is now tracks $C[\ell,b]=\sum_{j\in [n]}[h_\ell(j)=b]s_\ell(j)f_j$. When querying $i\in [n]$ the algorithm returns the estimate $\tilde f_i=\median_{\ell \in [k]} s_{\ell}(i) \cdot C[\ell,h_{\ell}(i)]$.

\paragraph{Learning-Augmented Sketches~\cite{HsuIKV19}} Given a base sketching algorithm (either CM or CS) and a space budget $B$, the corresponding learning-augmented algorithm (learned CM or learned CS) allocates a constant fraction of the space $B$ to the base sketching algorithm and the rest of the space to store items identified as heavy by a learned predictor. These items predicted to be heavy-hitters are stored in a separate table which maintains their counts exactly, and their updates are not sent to the sketching algorithm.

\subsection{Summary of Main Results and Paper Outline}

Our analysis, both of CM and CS, our algorithm, and prior work, is summarized in Table \ref{tbl:all_results}. 

\begin{table}[!ht]
\begin{center}
{\renewcommand{\arraystretch}{1.25}
\begin{tabular}{c|c|c|c}
\toprule
Algorithm                      & Weighted Error                                     & Uses Predictions? & Reference        \\ \hline
CountMin (CM)                     & $\Theta\left(\frac{\log n}{B }\right)$           & No                & Theorem \ref{thm:simplecm}        \\ 
CountSketch (CS)                 & $\Theta\left(\frac{1}B \right)$                   & No                & Theorem \ref{unlearnedmultiCS}\\ 
Learned CountMin              & $\Theta\left(\frac{\log(n/B)^2}{B \log n}\right)$ & Yes               & \cite{HsuIKV19}        \\ 
Learned CountSketch           & $\Theta\left(\frac{\log(n/B)}{B \log n}\right)$   & Yes               & Theorem \ref{thm:lcs1}  \\ 
\midrule
Our (Without predictions)        & $O\left(\frac{\log B + \text{poly}(\log \log n)}{B \log n}\right)$  & No                & Theorem \ref{thm:main_unlearned} \\ 
Our (Learned version) & $O\left(\frac{1}{B \log n}\right)$        & Yes               & Theorem \ref{thm:main_learned}
\\
\bottomrule
\end{tabular}}
\end{center}
\caption{Bounds are stated assuming that the total space is $B$ words of memory. Weighted error means that element $i$ is queried with probability proportional to $1/i$. Moreover, the table considers normalized frequencies, so that $f_i=1/i$.
}\label{tbl:all_results}
\end{table}

\paragraph{Summary of Theoretical Results}
We interpret Table \ref{tbl:all_results}. $B$ denotes the space bound, which is the total number of entries used in the CM or CS tables. First note that CS achieves lower weighted error compared to CM, proving the empirical advantage observed in \cite{HsuIKV19}. However, the learned version of CS only improves upon standard CS in the regime $B=n^{1-o(1)}$. While this setting does appear sometimes in practice \cite{goyal2012sketch, HsuIKV19} (referred to as high-accuracy regime), for CS, learning gives no asymptotic advantage in the low space regime. 

On the other hand, in the low space regime of $B = \text{poly}(\log n)$, our algorithm, without predictions, already archives close to a logarithmic factor improvement over even \emph{learned} CS. Furthermore, our learning-augmented algorithm achieves a logarithmic factor improvement over classical CS across all space regimes, whereas the learned CS only achieves a logarithmic factor improvement in the regime $B=n^{1-o(1)}.$ Furthermore, our learned version outperforms or matches learned CS in all space regimes.

Our learning-augmented algorithm can also be made \emph{parsimonious} in the sense that we only query the heavy-hitter oracle $\tilde O(B)$ times. This is desirable in large-scale streaming applications where evaluating even a small neural network on every single element would be prohibitive.

\begin{remark}
We remark that all bounds in this paper are proved by bounding the expected error when estimating the frequency of a single item, $\E[|\tilde f_i-f_i|]$, then using linearity of expectation. While we specialized our bounds to a query distribution which is proportional to the actual frequencies in \eqref{eq:weighted_error}, our bounds can be easily generalized to \emph{any} query distribution by simply weighing the expected errors of different items according to the given query distribution.
\end{remark}

\paragraph{Summary of Empirical Results}

We compare our algorithm without prediction to CS and our algorithm with predictions to that of \cite{HsuIKV19} on synthetic Zipfian data and on two real datasets corresponding to network traffic and internet search queries. In all cases, our algorithms outperform the baselines and often by a significant margin (up to \textbf{17x} in one setting). The improvement is especially pronounced when the space budget is small.

\paragraph{Outline of the Paper}
Our paper is divided into roughly two parts. One part covers novel and tight analysis of the classical algorithms CountMin (CM) and CountSketch (CS). The second part covers our novel algorithmic contributions which go beyond CM and CS. The main body of our paper focuses on our novel algorithmic components, i.e. the second part, and we defer our analysis of the performance of CountMin (CM) and CountSketch (CS), with and without predictions, to the appendix: in Section \ref{sec:countmin} we give tight analysis of CM for a Zipfian frequency distribution. In Section \ref{sec:countsketch} we give the analogous bounds for CS. Lastly, Section \ref{sec:learnedcountsketch} gives tight bounds for CS with predictions. Section \ref{sec:better_alg_without_predictions} covers our better frequency estimation without predictions while Section \ref{sec:better_alg_with_predictions} covers the learning-augmented version of the algorithm, as well as its extentions.

\subsection{Related Works}\label{sec:related_works}
\paragraph{Frequency Estimation}
While there exist other frequency estimation algorithms beyond CM and CS (such as \cite{misra1982finding, manku2002approximate, demaine2002frequency, karp2003simple, metwally2005efficient,braverman2017bptree}
) we study hashing based methods such as CM  \cite{cormode2005improved} and CS  \cite{charikar2002finding} as they are widely employed in practice and have additional benefits, such as supporting insertions \emph{and deletions}, and have applications beyond frequency estimation, such as in machine learning (feature selection \cite{aghazadeh2018mission}, compressed sending \cite{candes2006robust, donoho2006compressed}, and dimensionality reduction \cite{woodruff2014sketching, clarkson2017low} etc.).

\paragraph{Learning-augmented algorithms}
The last few  years have  witnessed  a rapid  growth in using machine learning methods to improve ``classical''  algorithmic problems. For example, they have been used to improve the performance of
data structures~\cite{kraska2017case, mitz2018model}, 
online algorithms~\cite{lykouris2018competitive,purohit2018improving,pmlr-v97-gollapudi19a,angelopoulos2020online,wei2020optimal,lattanzi2020online,aamand2022matchingdegrees, antoniadis2023online,anand2020customizing,diakonikolas2021learning,gupta2022augmenting}, 
combinatorial optimization~\cite{khalil2017learning,balcan2018learning,lattanzi2020online,mitzenmacher2020scheduling,dinitz2021faster,ChenSVZ22}, 
similarity search and clustering~\cite{wang2016learning,dong2019learning,ErgunFSWZ22,nguyen22,silwal2023kwikbucks}.
Similar to our work, sublinear constraints, such as memory or sample complexity, have also been studied under this framework \cite{HsuIKV19,indyk2019learning,jiang2020learningaugmented,cohen2020composable,du2021puttingthelearning,EdenINRSW21,ChenEILNRSWWZ22, li2023learning,silwal2023kwikbucks}.


\section{Improved Algorithm without Predictions}\label{sec:better_alg_without_predictions}
We first present our frequency estimation algorithm which does not use any predictions. Later, we build on top of it for our final learning-augmented frequency estimation algorithm.
\begin{algorithm}[!ht]
\caption{\label{alg:update_unlearned}(Not augmented) Frequency update algorithm}
\begin{algorithmic}[1]
\State \textbf{Input:} Stream of updates to an $n$ dimensional vector, space budget $B$
\Procedure{Update}{}
\State $T \gets \Theta(\log \log n)$
\For{$j = 1$ to $T-1$}
\State $S_j \gets $ CountSketch table with $3$ rows and $\frac{B}{6T}$ columns
\EndFor
\State $S_T \gets $ CountSketch table with $3$ rows and $\frac{B}{6}$ columns
\For{stream element $(i, \Delta)$}
\State Input $(i,\Delta)$ in each of the $T$ CountSketch tables $S_j$
\EndFor
\EndProcedure
\end{algorithmic}
\end{algorithm}

\begin{algorithm}[!ht]
\caption{\label{alg:main_unlearned}(Not augmented) Frequency estimation algorithm}
\begin{algorithmic}[1]
\State \textbf{Input:} Index $i \in [n]$ for which we want to estimate $f_i$
\Procedure{Query}{}
\For{$j = 1$ to $T-1$}
\State $\hat{f}_i^j \gets $ estimate of the $i$th frequency given by table $S_j$
\EndFor
\State $\tilde{f}_i \gets \text{Median}(\hat{f}_i^1, \ldots, \hat{f}_i^{T-1})$
\If{$\tilde{f}_i < O((\log \log n))/B$}
\State \textbf{Return} 0
\Else
\State \textbf{Return} $\hat{f}_i^T$, the estimate given by table $S_T$
\EndIf
\EndProcedure
\end{algorithmic}
\end{algorithm}

The main guarantees of of the algorithm is the following:
\begin{theorem}\label{thm:main_unlearned}
Consider Algorithm \ref{alg:update_unlearned} with space parameter $B \ge \log n$ updated over a Zipfian stream. Let $\{\hat{f}_i\}_{i=1}^n$ denote the estimates computed by Algorithm \ref{alg:main_unlearned}. The expected weighted error \eqref{eq:weighted_error} is $\mathbb{E}\left[ \frac{1}{N} \cdot \sum_{i=1}^n f_i \cdot |f_i - \hat{f}_i| \right] = O\left( \frac{\log B + \text{poly}(\log \log n)}{B \log n} \right)$.
\end{theorem}

\paragraph{Algorithm and Proof intuition:}
Let $B' = B/\log \log n$. At a high level, we show that for every $i \le  B'$, we execute line 10 of Algorithm \ref{alg:main_unlearned} and the error satisfies $|1/i - \hat{f}_i| \approx 1/B'$ (recall in the Zipfian case, the $i$th largest frequency is $f_i = 1/i$). On the other hand, for $i \ge B'$, we show that (with sufficiently high probability) line 8 of Algorithm \ref{alg:main_unlearned} will be executed, resulting in $|1/i - \hat{f}_i| = |1/i - 0| = 1/i$. 

It might be perplexing at first sight why we wish to set the estimate to be $0$, but this idea has solid intuition: it turns out the \emph{additive} error of standard CountSketch with $B'$ columns is actually of the order $1/B'$. Thus, it does not make sense to estimate elements whose true frequencies are much smaller than $1/B'$ using CountSketch. A challenge is that we do not know a priori which elements these are. We circumvent this via the following reasoning: if CountSketch itself outputs $\approx 1/B'$ as the estimate, then either one of the following must hold: 
\begin{itemize}[leftmargin=*]
    \item The element has frequency $1/i \ll 1/B'$, in which case we should set the estimate to $0$ to obtain error $1/i$, as opposed to error $1/B' - 1/i \approx 1/B'$.
    \item The true element has frequency $\approx 1/B'$ in which case either using the output of the CountSketch table or setting the estimate to $0$ both obtain error approximately $O(1/B')$, so our choice is inconsequential. 
\end{itemize}
In summary, the output of CountSketch itself suggests whether we should output an estimated frequency as $0$. We slightly modify the above approach with $O(\log \log n)$ repetitions to obtain sufficiently strong concentration, leading to a \emph{robust} method to identify small frequencies. The proof formalizes the above plan and is given in full detail in Section \ref{sec:proof_main_unlearned}.

By combining our algorithm with predictions, we obtain improved guarantees.

\section{Improved Learning-Augmented Algorithm}\label{sec:better_alg_with_predictions}
\begin{algorithm}[!ht]
\caption{\label{alg:update_learned}(Learning-augmented) Frequency update algorithm}
\begin{algorithmic}[1]
\State \textbf{Input:} Stream of updates to an $n$ dimensional vector, space budget $B$, access to a heavy-hitter oracle which correctly identifies the top $B/2$ heavy-hitters
\Procedure{Update}{}
\State $T \gets O(\log \log n)$
\For{$j = 1$ to $T-1$}
\State $S_j \gets $ CountSketch table with $3$ rows and $\frac{B}{12T}$ columns
\EndFor
\State  $S_T \gets $ CountSketch table with $3$ rows and $\frac{B}{12}$ columns
\For{stream element $(i, \Delta)$}
\If{$i$ is a top $B/2$ heavy-hitter}
\State Maintain the frequency of $i$ exactly
\Else
\State Input $(i,\Delta)$ in each of the $T$ CountSketch tables $S_j$
\EndIf
\EndFor
\EndProcedure
\end{algorithmic}
\end{algorithm}

\begin{algorithm}[!ht]
\caption{\label{alg:main_learned}(Learning-augmented) Frequency estimation algorithm}
\begin{algorithmic}[1]
\State \textbf{Input:} Index $i \in [n]$ for which we want to estimate $f_i$
\Procedure{Query}{}
\If{$i$ is a top $B/2$ heavy-hitter}
\State Output the exact maintained frequency of $i$
\Else{}
\State \textbf{Return}
$\hat{f}_i \gets$ output of Alg. \ref{alg:main_unlearned} using the CountSkech tables created in Alg.\ref{alg:update_learned}
\EndIf
\EndProcedure
\end{algorithmic}
\end{algorithm}

\begin{theorem}\label{thm:main_learned}
Consider Algorithm \ref{alg:update_learned} with space parameter $B \ge \log n$ updated over a Zipfian stream. Suppose we have access to a heavy-hitter oracle which correctly identifies the top $B/2$ heavy-hitters in the stream. Let $\{\hat{f}_i\}_{i=1}^n$ denote the estimates computed by Algorithm \ref{alg:main_learned}. The expected weighted error \eqref{eq:weighted_error} is
$\mathbb{E}\left[ \frac{1}{N} \cdot \sum_{i=1}^n f_i \cdot |f_i - \hat{f}_i| \right] = O\left( \frac{1}{B \log n} \right). $
\end{theorem}

\paragraph{Algorithm and Proof Intuition:}
Our final algorithm follows a similar high-level design pattern used in the learned CM algorithm of \cite{HsuIKV19}: given an oracle prediction, we either store the frequency of heavy element directly, or input the element into our algorithm from the prior section which does not use any predictions. 

The workhorse of our analysis is the proof of Theorem \ref{thm:main_unlearned}. First note that we obtain $0$ error for $i < B/2$. Thus, all error comes from indices $i \ge B/2$. Recall the intuition for this case from Theorem \ref{thm:main_unlearned}: we want to output $0$ as our estimates as this results in lower error than the additive error from CS. The same analysis as in the proof of Theorem \ref{thm:main_unlearned} shows that we are able to detect small frequencies and appropriately output an estimate from either the $T$th CS table or output $0$.

\subsection{Parsimonious Learning}
In~\cref{thm:main_learned}, we assumed access to a heavy-hitter oracle which we can use on every single stream element to predict if it is heavy. In practical streaming applications, this will likely be infeasible. Indeed, even if the oracle is a small neural network, it is unlikely that we can query it for every single element in a large-scale streaming application. We therefore consider the so called \emph{parsimonious} setting with the goal of obtaining the same error bounds on the expected error but with an algorithm that makes \emph{limited queries} to the heavy-hitter oracle. This setting has recently been explored for other problems in the learning-augmented literature \cite{im2022parsimonious,Bhaskara2021logregret,Drygala2023costly}.

Our algorithm works similarly to~\cref{alg:update_learned} except that when an element $(i,\Delta)$ arrives, we only query the heavy-hitter oracle with some probability $p$ (proportional to $\Delta$). We will choose $p$ so that we in expectation only query $\tilde O(B)$ elements, rather than querying the entire stream. To be precise, whenever an item arrives, we first check if it is already classified as one of the top $B/2$ heavy-hitters in which case, we update its exact count (from the point in time where was classified as heavy). Otherwise, we query the heavy-hitter oracle with probability $p$. In case the item is queried and is indeed one of the top $B/2$ heavy-hitters, we start an exact count of that item. An arriving item which is not used as a query for the heavy-hitter oracle and was not earlier classified as a heavy-hitter is processed as in~\cref{alg:update_learned}.

Querying for an element, we first check if it is classified as a heavy-hitter and if so, we use the estimate from the separate lookup table. If not, we estimate its frequency using~\cref{alg:main_learned}. With this algorithm, the count of a heavy-hitter will be underestimated since it may appear several times in the stream before it is used as a query for the oracle and we start counting it exactly. However, with our choice of sampling probability, with high probability it will be sampled sufficiently early to not affect its final count too much. We present the pseudocode of the algorithm as well as the precise result and its proof in Appendix~\ref{app:parsimonious}.

\subsection{Algorithm variant with worst case guarantees}\label{sec:worst_case}
In this section we discuss a variant of our algorithm with worst case guarantees. To be more precise, we consider the case where the actual frequency distribution is not Zipfian. The algorithm we discuss is actually a more general case of Algorithm \ref{alg:main_unlearned} and in fact, it completely recovers the asymptotic error guarantees of Theorem \ref{thm:main_unlearned} (as well as Theorem \ref{alg:main_learned} if we use predictions).

Recall that Algorithm \ref{alg:main_unlearned}
outputs $0$ when the estimated frequency is below $T/B$ for $T = O(\log \log n)$. This parameter has been tuned to the Zipfian case. As stated in Section \ref{sec:better_alg_without_predictions}, the main intuition for this parameter is that it is of the same order as the additive error inherent in CountSketch, which we discuss now. Denote by $f_{\overline{P}}$ the frequency vector where we zero out the largest $P$ coordinates. For every frequency, the expected additive error incurred by a CountSketch table with $B'$ columns is $O(\|f_{\overline{B'}}\|_2/\sqrt{B'})$. In the Zipfian case, this is equal to
$ O\left(\frac{\|f_{\overline{B'}}\|_2}{\sqrt{B'}}\right) = O\left( \frac{1}{B'} \right)$,
which is exactly the threshold we set\footnote{Recall $B' = B/T$ in Algorithm \ref{alg:main_unlearned}.}. Thus, our robust variant simply replaces this tuned parameter $O(T/B)$ with an estimate of $O(\|f_{\overline{B'}}\|_2/\sqrt{B'})$ where $B' = B/T$. We given an algorithm which efficiently estimates this quantity in a stream. Note this quantity is only needed for the query phase.

\begin{lemma}\label{lem:tail_estimator}
With probability at least $1-\exp\left(\Omega\left(B\right)\right)$,
Algorithm \ref{alg:tail_estimator} outputs an estimate $V$ satisfying
$\Omega\left(\left\Vert f_{\overline{3B'}}\right\Vert _{2}^{2}/B'\right)\le V\le O\left(\left\Vert f_{\overline{B'/10}}\right\Vert _{2}^{2}/B'\right)$.
\end{lemma}

The algorithm and analysis are given in Section \ref{sec:tail_estimator_algorithm}. Replacing the threshold in Line $7$ of Algorithm \ref{alg:main_unlearned} with the output of Algorithm \ref{alg:tail_estimator} (more precisely the square root of the value) readily gives us the following worst case guarantees. Lemma \ref{lem:worst_case} states that the expected error of the estimates outputted by Algorithm \ref{alg:main_unlearned} using $B$, regardless of the true frequency distribution, is no worse than that of a standard CountSketch table using slightly smaller $O(B/\log \log n)$ space.

\begin{lemma}\label{lem:worst_case}
Suppose $B \ge \log n$. Let $\{\hat{f}_i\}_{i=1}^n$ denote the estimates of Algorithm \ref{alg:main_unlearned} using $B/2$ space and with Line $7$ replaced by the square root of the estimate of Algorithm \ref{alg:tail_estimator}, also using $B/2$ space. Suppose the condition of Lemma \ref{lem:tail_estimator} holds. Let  $\{\hat{f}'_i\}_{i=1}^n$ denote the estates computed by a CountSketch table with $\frac{cB}{\log \log n}$ columns for a sufficiently small constant $c$. Then, $\mathbb{E}[ |\hat{f}_i - f_i|] \le \mathbb{E}[ |\hat{f}_i' - f_i|]$.
\end{lemma}

\begin{remark}
    The learned version of the algorithm automatically inherits any worst case guarantees from the unlearned (without predictions) version. This is because we only set aside half the space to explicitly track the frequency of some elements, which has worst case guarantees, while the other half is used for the unlearned version, also with worst case guarantees.
\end{remark}

\section{Experiments}

We experimentally evaluate our algorithms with and without predictions on real and synthetic datasets and demonstrate that the improvements predicted by theory hold in practice.
 Comprehensive additional figures are given in \cref{appendix-experiment}.

\begin{figure}[t]
    \centering
    \includegraphics[width=0.35\textwidth]{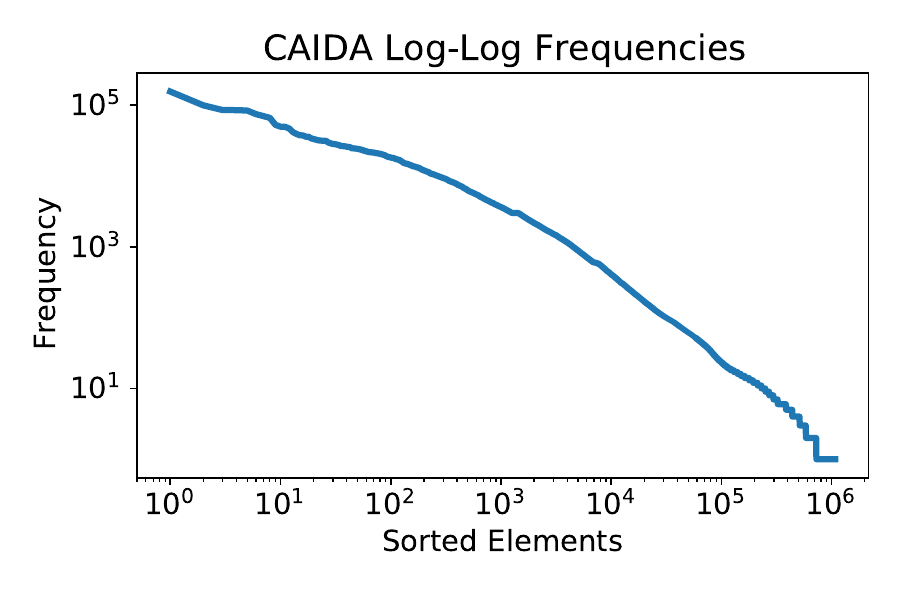}
    \includegraphics[width=0.35\textwidth]{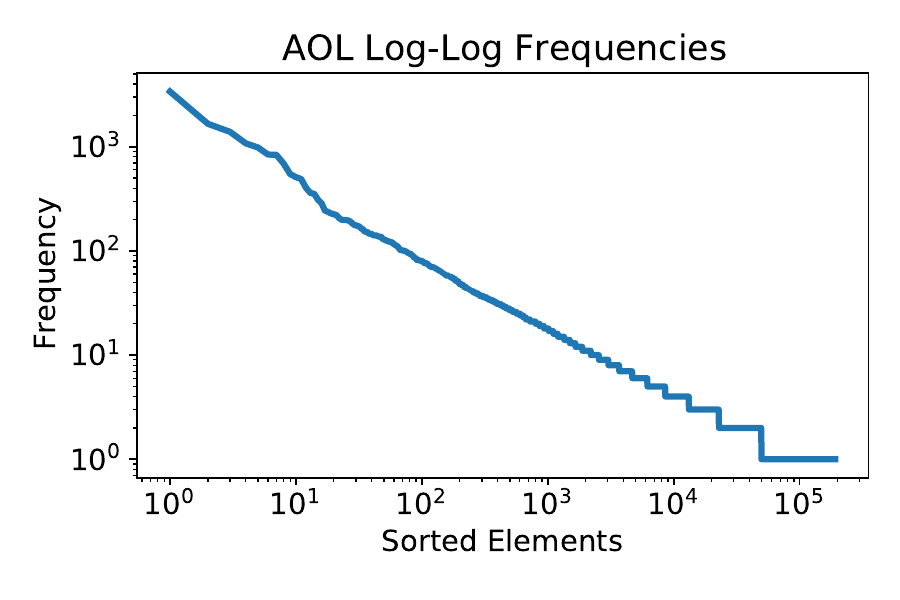}
    \caption{Log-log plots of the sorted frequencies of the first day/minute of the CAIDA/AOL datasets. Both data distributions are heavy-tailed with few items accounting for much of the total stream.}
    \label{fig:loglog}
\end{figure}

\begin{figure}[t]
    \centering
    \includegraphics[width=0.45\textwidth]{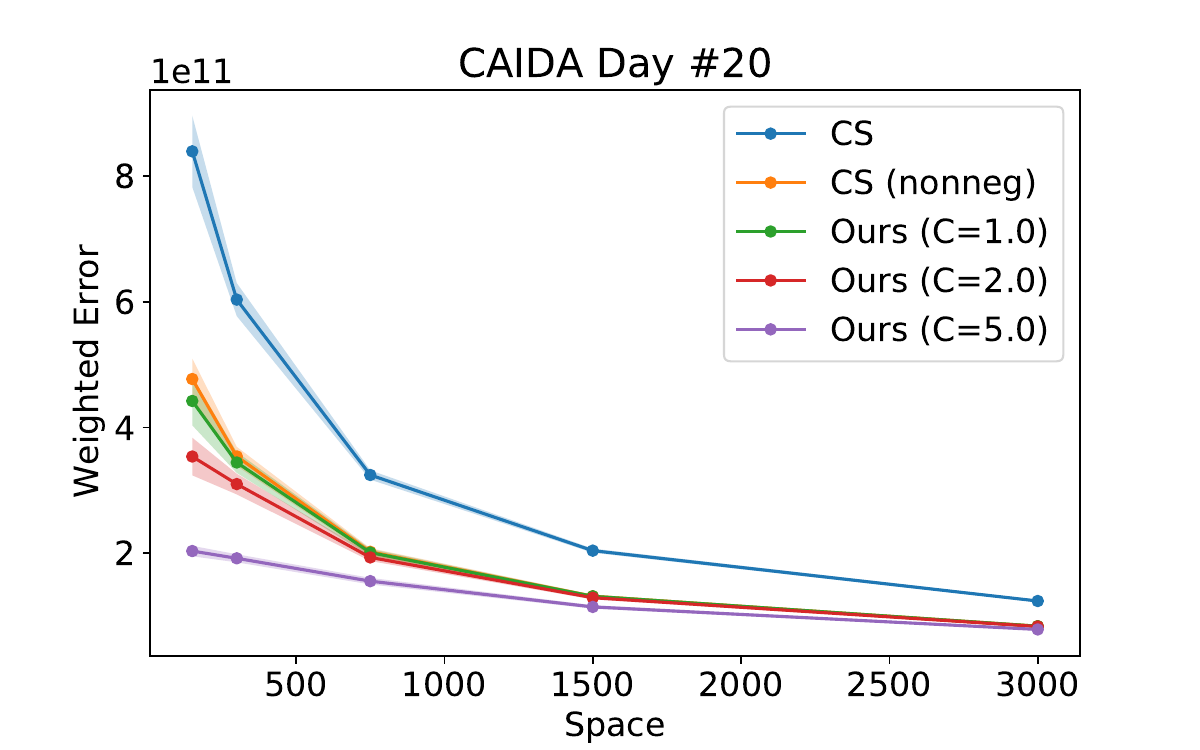}
    \includegraphics[width=0.45\textwidth]{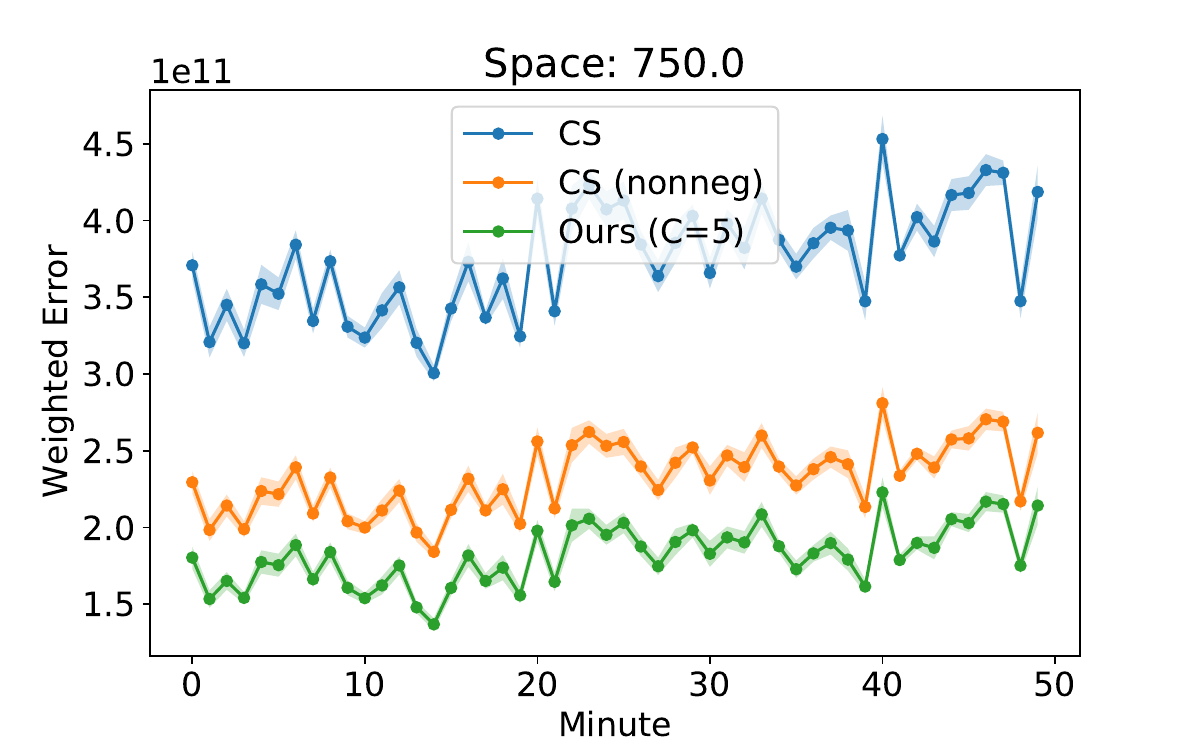}
    \caption{Comparison of weighted error without predictions on the CAIDA dataset. The left plot compares the performance of various algorithms (including our algorithm with different choices of $C$) for a fixed dataset and varying space. The right plot compares algorithms over time across separate streams for each minute of data for a specific choice of space being $750$.}
    \label{fig:ip-nopreds}
\end{figure}

\begin{figure}[t]
    \centering
    \includegraphics[width=0.45\textwidth]{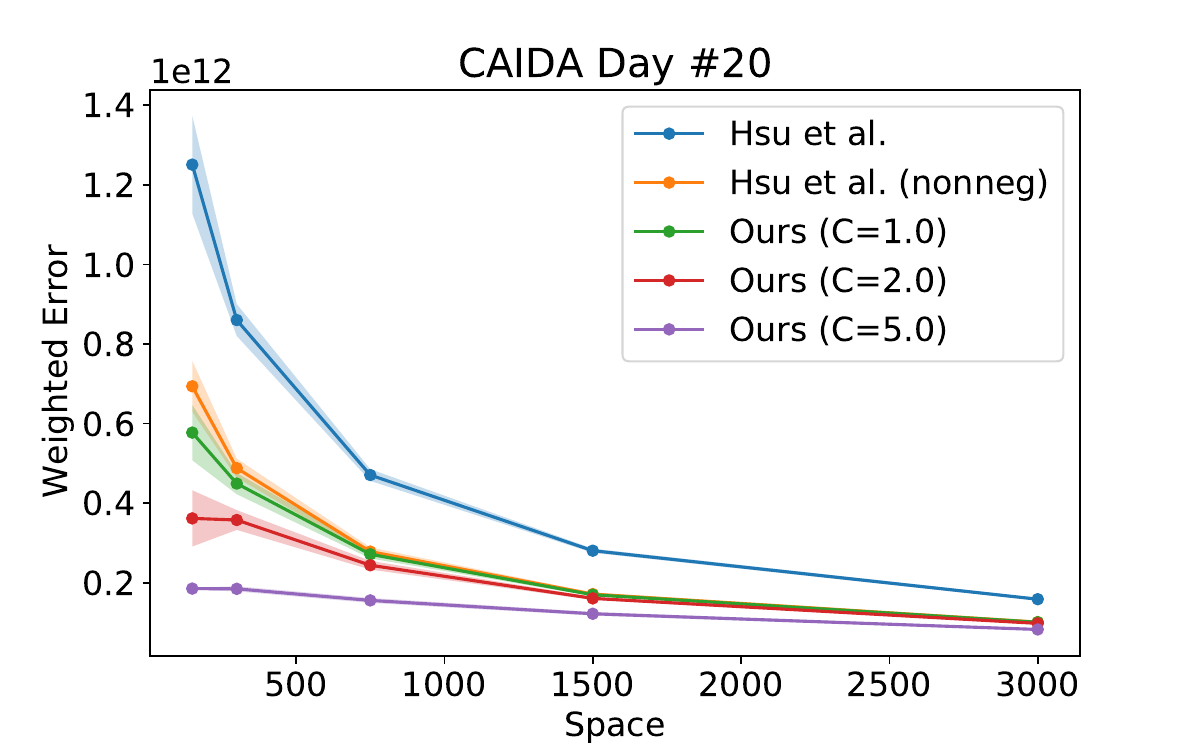}
    \includegraphics[width=0.45\textwidth]{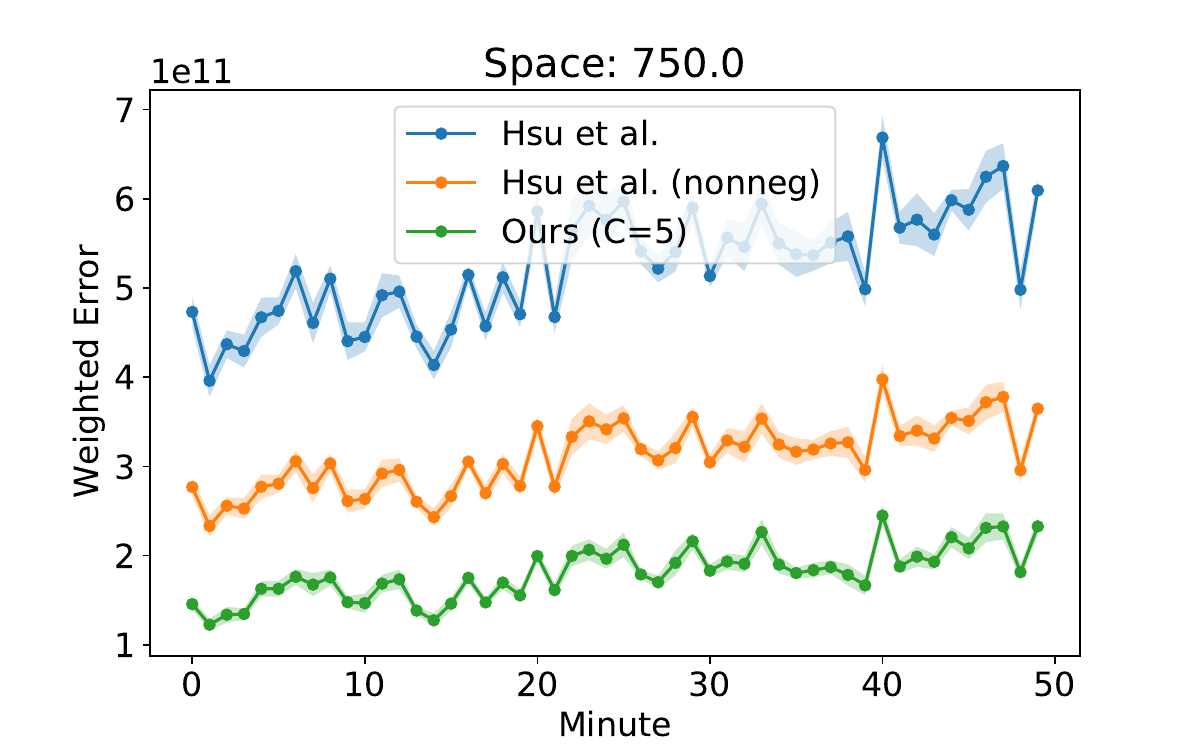}
    \caption{Comparison of weighted error with predictions on the CAIDA dataset.}
    \label{fig:ip-preds}
\end{figure}

\begin{figure}[t]
    \centering
    \includegraphics[width=0.45\textwidth]{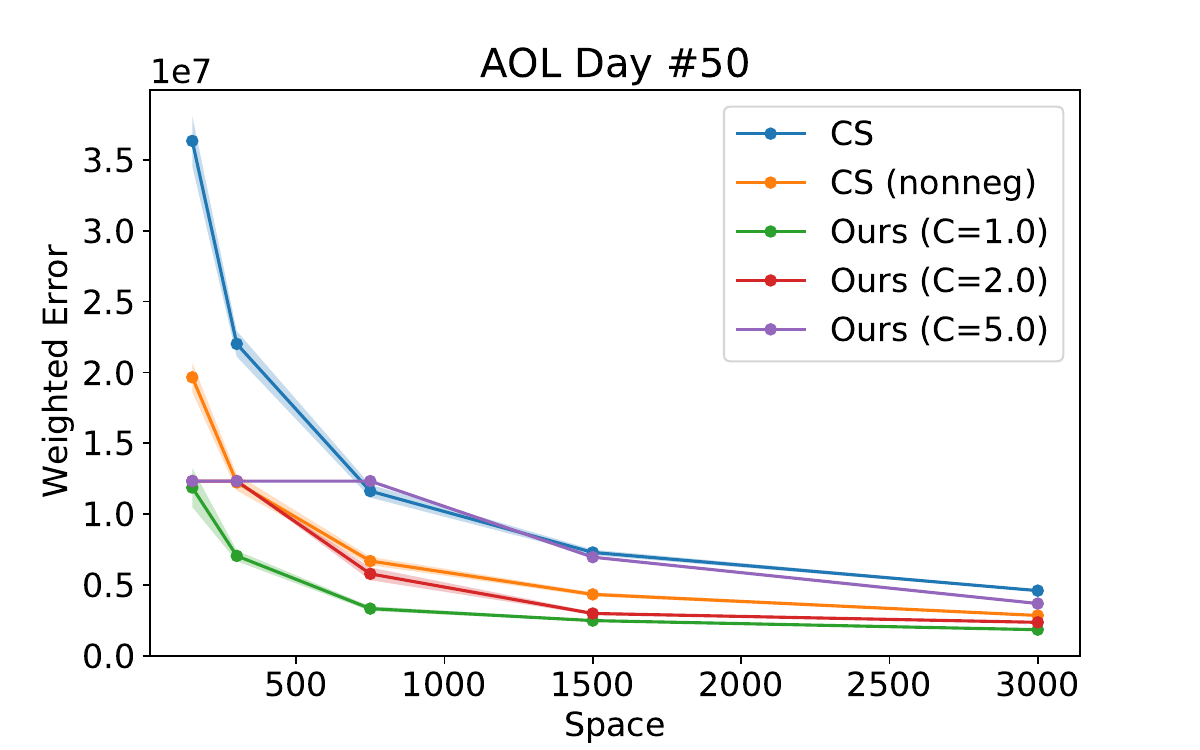}
    \includegraphics[width=0.45\textwidth]{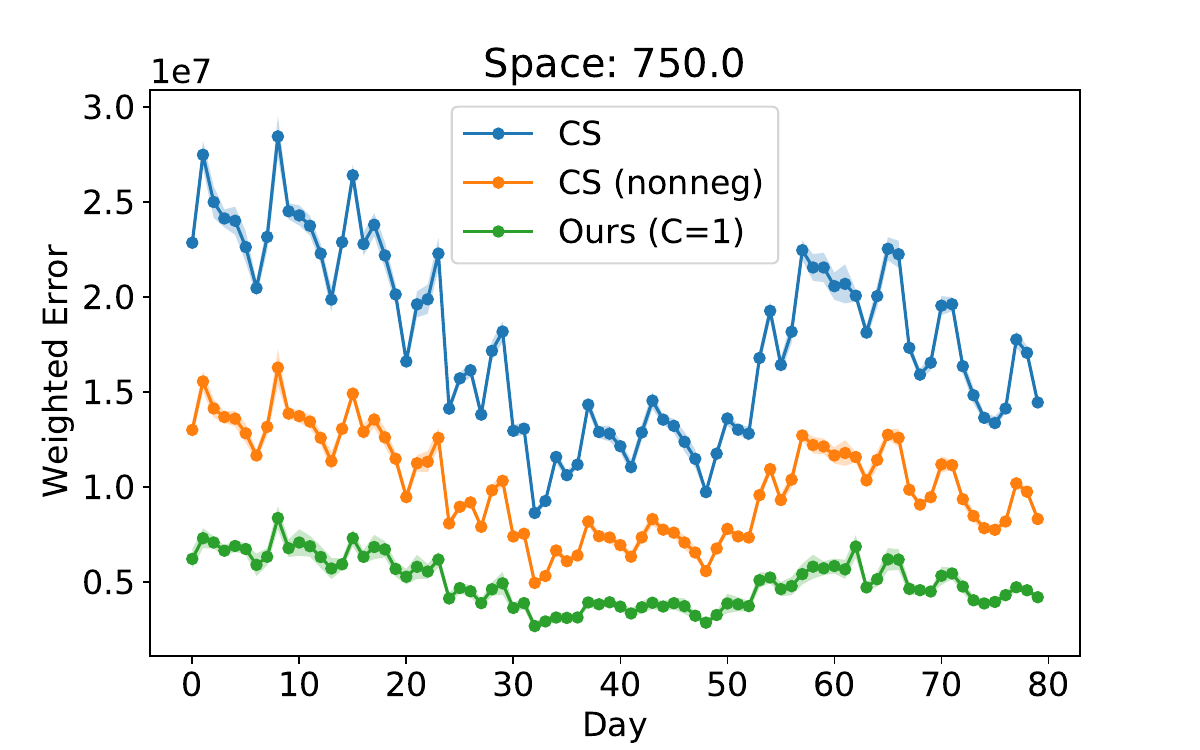}
    \caption{Comparison of weighted error without predictions on the AOL dataset.}
    \label{fig:aol-nopreds}
\end{figure}

\begin{figure}[t]
    \centering
    \includegraphics[width=0.45\textwidth]{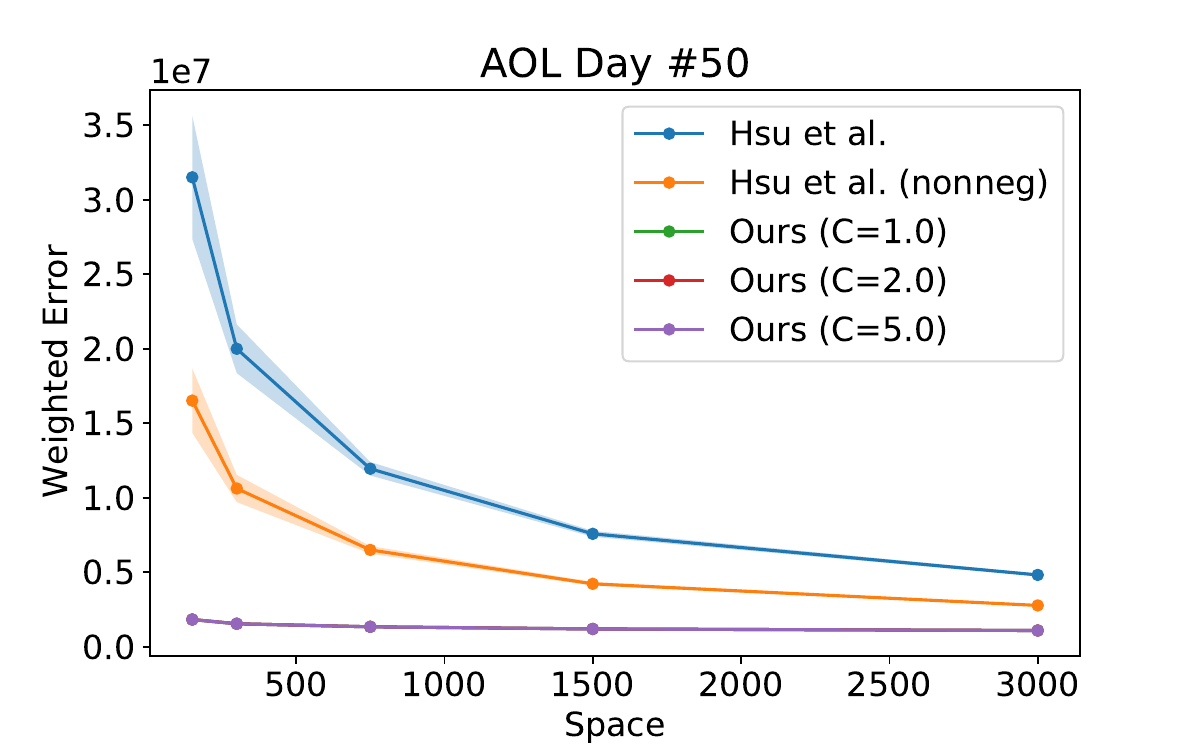}
    \includegraphics[width=0.45\textwidth]{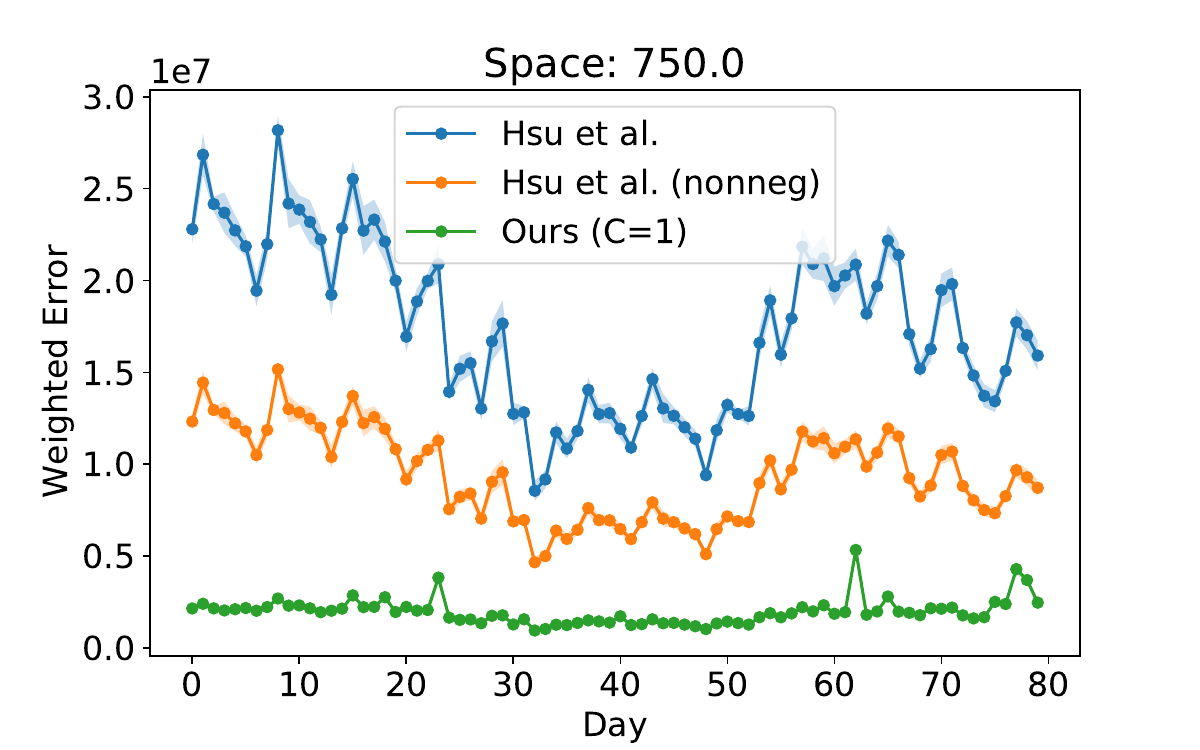}
    \caption{Comparison of weighted error with predictions on the AOL dataset.}
    \label{fig:aol-preds}
\end{figure}

\paragraph{Algorithm Implementations}
In the setting without predictions, we compare our algorithm to CountSketch (CS) (which was shown to have favorable empirical performance compared to CountMin (CM) in~\cite{HsuIKV19} and better theoretical performance due to our work). In the setting with predictions, we compare the algorithm of~\cite{HsuIKV19}, using CS as the base sketch and dedicated half of the space for items which are predicted to be heavy by the learned oracle. For all implementations, we use three rows in the CS table and vary the number of columns. We additionally augment each of these baselines with a version that truncates all negative estimated frequencies to zero as none of our datasets include stream deletions. This simple change does not change the asymptotic $(\eps, \delta)$ classic sketching guarantees but does make a big difference when measuring empirical weighted error.

We implement a simplified and practical version of our algorithm which uses a single CS table. If the median estimate of an element is below a threshold of $C n / w$ for domain size $n$, sketch width $w$ (a third of the total space), and a tunable constant $C$, the estimate is instead set to $0$. As all algorithms use a single CS table as the basic building block with different estimation functions, for each trial we randomly sample hash functions for a single CS table and only vary the estimation procedure used.

We evaluate algorithms according the weighted error as in \cref{eq:weighted_error} but also according to unweighted error which is simply the sum over all elements of the absolute estimation error, given by $\sum_i |f_i - \tilde{f}_i|$.
Space is measured by the size of the sketch table, and all errors are averaged over 10 independent trials with standard deviations shown shaded in.

\paragraph{Datasets}
We compare our algorithm with prior work on three datasets.
We use the same two real-world datasets and predictions from~\cite{HsuIKV19}: the CAIDA and AOL datasets. The CAIDA dataset~\cite{caida} contains 50 minutes of internet traffic data. For each minute of data, the stream is formed of the IP addresses associated with packets going through a Tier1 ISP. A typical minute of data contains 30 million packets accounted for by 1 million IPs. The AOL dataset~\cite{pass2006picture} contains 80 days of internet search queries with a typical day containing $\approx 3\cdot 10^5$ total queries and $\approx 10^5$ unique queries. As shown in \cref{fig:loglog}, both datasets approximately follow a power law distribution.
For both datasets, we use the predictions from prior work~\cite{HsuIKV19} formed using recurrent neural networks.
We also generate synthetic data following a Zipfian distribution with $n=10^7$ elements and where the $i$th element has frequency $n/i$.

\paragraph{Results}
Across the board, our algorithm outperforms the baselines. On the CAIDA and AOL datasets without predictions, our algorithm consistently outperforms the standard CS with up to \textbf{4x} smaller error with space $300$. This gap widens when we compare our algorithm with predictions to that of \cite{HsuIKV19} with a gap of up to \textbf{17x} with space $300$. In all cases, the performance of CS and \cite{HsuIKV19} is significantly improved by the simple trick of truncating negative estimates to zero. However, our algorithm still outperforms these ``nonneg'' baselines.
The longitudinal plots which compare algorithms over time show that our algorithm consistently outperforms the state-of-the-art with and without predictions.

In the case of the CAIDA dataset, predictions do not generally improve the performance of any of the algorithms. This is consistent with the findings of~\cite{HsuIKV19} where the prediction quality for the CAIDA dataset was relatively poor. However, for the AOL which has a more accurate learned oracle, our algorithm in particular is significantly improved when augmented with predictions. Intuitively, the benefit of our algorithm comes from removing error due to noise for low frequency elements. Conversely, good predictions help to obtain very good estimates of high frequency elements. In combination, this yields very small total weighted error.

In \cref{appendix-experiment}, we display comprehensive experiments of the performance of the algorithms across the CAIDA and AOL datasets with varying space and for both weighted and unweighted error as well as results for synthetic Zipfian data. In all cases, our algorithm outperforms the baselines. On synthetic Zipfian, the gap between our algorithm and the non-negative CS for weighted error is relatively small compared to that for the real datasets. While we mainly focus on weighted error in this work, the benefits of our algorithm are even more significant for unweighted error as setting estimates below the noise floor to zero is especially impactful for this error measure. In general, we see the trend, matching our theoretical results, that as space increases, the gap between the different algorithms shrinks as the estimates of the base CS become more accurate.

\section*{Acknowledgements}
We are grateful to Piotr Indyk for insightful discussions. Anders Aamand is supported by DFF-International Postdoc Grant 0164-00022B from the Independent Research Fund Denmark and a Simons Investigator Award. Justin Chen is supported by an NSF Graduate Research Fellowship under Grant No.\ 174530. Huy Nguyen is supported by NSF Grants 2311649 and 1750716.

\bibliographystyle{plain}
\bibliography{main}

\appendix

\newpage
\section{Organization of the Appendix}

In Section \ref{sec:countmin}, we give tight bounds for CM with Zipfians, as well as tight bounds for CS (in Section \ref{sec:countsketch}) and its learning augmented variants (in Section \ref{sec:learnedcountsketch}). Our results for CM and CS, with and without predictions, can be summarized in Table \ref{tbl:results}. We highlight the results of these sections are presented assuming that we use a \emph{total} of $B$ buckets. With $k$ hash functions, the range of each hash functions is therefore $[B/k]$. We make this assumption since we wish to compare the expected error incurred by the different sketches when the total sketch size is fixed. 

Sections \ref{sec:proof_of_thm_main_unlearned} and \ref{sec:proof_main_learned} contain the proofs of Theorems \ref{thm:main_unlearned} and \ref{thm:main_learned}, respectively. Section \ref{sec:tail_estimator_algorithm} contains omitted proofs of Section \ref{sec:worst_case}.

In Section \ref{appendix-experiment}, we include additional experimental results.


\paragraph{Notation} We use the bracket $[ \cdot ]$ notation for the indicator function. $a \lesssim b$ denotes $a \le Cb$ for some fixed positive constant $C$.

\begin{table}[h!]
\centering	
\resizebox{.9\textwidth}{!}{%
\renewcommand{\arraystretch}{1.5}
\begin{tabular}{l|l l} 
\toprule
& $k=1$ & $k > 1$\\
\midrule
Count-Min (CM) & $\Theta \left( \frac{\log n}{B} \right)$~\cite{HsuIKV19}	&	$\Theta \left(\frac{k \cdot \log({{kn}\over B})}{B} \right)$ \\
Learned Count-Min (L-CM) & $\Theta \left(\frac{\log^2({n\over B})}{B\log n} \right)$~\cite{HsuIKV19}	&	$\Omega \left( \frac{\log^2 (\frac{n}{B})}{B\log n} \right)$~\cite{HsuIKV19}\\
Count-Sketch (CS) & $\Theta \left( \frac{  \log B}{B} \right)$ & $\Omega \left(\frac{k^{1/2}  }{B\log k } \right)$ and $O \left(\frac{k^{1/2}  }{B} \right)$\\ 
Learned Count-Sketch (L-CS) & $\Theta \left( \frac{\log \frac{n}{B}}{B \log n} \right)$ & $\Omega \left( \frac{\log \frac{n}{B}}{B \log n} \right)$ \\ 
\bottomrule
\end{tabular}
}
\caption{This table summarizes our and previously known results on the expected frequency estimation error of Count-Min (CM), Count-Sketch (CS) and their learned variants (i.e., L-CM and L-CS) that use $k$ functions and overall space $k\times {B\over k}$ under Zipfian distribution.  For CS, we assume that $k$ is odd (so that the median of $k$ values is well defined). 
}\label{tbl:results}
\end{table}

\section{Tight Bounds for Count-Min with Zipfians}\label{sec:countmin}
For both Count-Min and Count-Sketch we aim at analyzing the expected value of the variable $\sum_{i\in [n]}f_i\cdot|\tilde f_i- f_i|$ where $f_i=1/i$ and $\tilde{f_i}$ is the estimate of $f_i$ output by the relevant sketching algorithm. 
Throughout this paper we use the following notation: For an event $E$ we denote by $[E]$ the random variable in $\{0,1\}$ which is $1$ if and only if $E$ occurs.
We begin by presenting our improved analysis of Count-Min with Zipfians. The main theorem is the following.
\begin{theorem}\label{thm:simplecm}
Let $n,B,k \in \N$ with $k\geq 2$ and $B\leq n/k$. Let further $h_1,\dots,h_k: [n] \to [B]$ be independent  and truly random hash functions. For $i \in [n]$ define the random variable $\tilde{f_i}=\min_{\ell \in [k]} \left( \sum_{j\in [n]} [h_{\ell}(j)=h_\ell(i)]f_j \right)$. For any $i\in [n]$ it holds that $ \E[|\tilde f_i- f_i|]=\Theta \left( \frac{\log \left( \frac{n}{B} \right)}{B} \right)$.
\end{theorem}
Replacing $B$ by $B/k$ in~\Cref{thm:simplecm} and using linearity of expectation we obtain the desired bound for Count-Min in the upper right hand side of~\Cref{tbl:results}. The natural assumption that $B\leq n/k$ simply says that the total number of buckets is upper bounded by the number of items. 

To prove~\Cref{thm:simplecm} we start with the following lemma which is a special case of the theorem.
\begin{lemma}\label{simpleanalysis}
Suppose that we are in the setting of~\Cref{thm:simplecm} and further that\footnote{In particular we dispose with the assumption that $B \leq n/k$.} $n=B$. Then 
\begin{align*}
\E[|\tilde{f_i}-f_i|]=O \left( \frac{1}{n} \right).
\end{align*}
\end{lemma}
\begin{proof}
It suffices to show the result when $k=2$ since adding more hash functions and corresponding tables only decreases the value of $|\tilde f_i- f_i|$. Define $Z_\ell=\sum_{j\in [n]\setminus\{i\}} [h_\ell(j)=h_\ell(i)]f_j$ for $\ell \in [2]$ and note that these variables are independent. For a given $t\geq 3/n$ we wish to upper bound $\Pr[Z_{\ell}\geq t]$. 
Let $s<t$ be such that $t/s$ is an integer, and note that if $Z_{\ell}\geq t$ then either of the following two events must hold:
\begin{enumerate}
\item[$E_1$:] There exists a $j\in [n]\setminus \{i\}$ with $f_j>s$ and $h_{\ell}(j)=h_{\ell}(i)$.
\item[$E_2$:] The set $\{j\in [n]\setminus \{i\}:h_{\ell}(j)=h_{\ell}(i)\}$ contains at least $t/s$ elements.
\end{enumerate}
To see this, suppose that $Z_\ell\geq t$ and that $E_1$ does not hold. Then 
$$
t\leq Z_\ell =\sum_{j\in [n]\setminus\{i\}} [h_\ell(j)=h_\ell(i)]f_j\leq s |\{j\in [n]\setminus \{i\}:h_{\ell}(j)=h_{\ell}(i)\}|,
$$
so it follows that $E_2$ holds.
By a union bound,
\begin{align*}
\Pr[Z_{\ell}\geq t]\leq \Pr[E_1]+\Pr[E_2]\leq  \frac{1}{ns}+\binom{n}{t/s}n^{-t/s}\leq  \frac{1}{ns}+\left( \frac{es}{t}\right)^{t/s}.
\end{align*}
Choosing $s=\Theta(\frac{t}{\log (tn)})$ such that $t/s$ is an integer, and using $t\geq {3\over n}$, a simple calculation yields that $\Pr[Z_{\ell}\geq t]=O\left( \frac{\log (tn)}{tn} \right)$. Note that  $|\tilde{f_i}-f_i|=\min (Z_1,Z_2)$. As $Z_1$ and $Z_2$ are independent, $\Pr[|\tilde{f_i}-f_i|\geq t]=O\left( \left(\frac{\log (tn)}{tn}\right)^2 \right)$, so 
\begin{align*}
\E[|\tilde{f_i}-f_i|]=\int_{0}^\infty \Pr[Z \geq t ] \, dt\leq \frac{3}{n}+O \left(\int_{3/n}^\infty \left(\frac{\log (tn)}{tn}\right)^2 \, dt \right)=O\left( \frac{1}{n} \right).
\end{align*}
\end{proof}
We can now prove the full statement of~\Cref{thm:simplecm}.
\begin{proof}[Proof of~\Cref{thm:simplecm}]
We start out by proving the upper bound. Let $N_1=[B]\setminus \{i\}$ and $N_2=[n]\setminus ([B] \cup \{i\})$. 
Let $b\in [k]$ be such that $\sum_{j\in N_1} f_j \cdot [h_b(j)=h_b(i)]$ is minimal. Note that $b$ is itself a random variable. We also define
\begin{align*}
Y_1&=\sum_{j\in N_1} f_j\cdot  [h_b(j)=h_b(i)], \text{ and } Y_2=\sum_{j\in N_2} f_j\cdot  [h_b(j)=h_b(i)].
\end{align*}
Then, $|\tilde f_i- f_i|\leq Y_1 +Y_2$. Using~\Cref{simpleanalysis}, we obtain that $\E[Y_1]=O(\frac{1}{B})$. For $Y_2$ we observe that
\begin{align*}
\E[Y_2\mid b]=\sum_{j\in N_2}  \frac{f_j}{B}= O\left( \frac{ \log \left( \frac{n}{B}\right)}{B}\right).
\end{align*}
We conclude that 
\begin{align*}
\E[|\tilde f_i- f_i|]\leq \E[Y_1]+\E[Y_2]=\E[Y_1]+\E[\E[Y_2 \mid b]] =O\left( \frac{ \log \left( \frac{n}{B}\right)}{B}\right).
\end{align*}
Next we prove the lower bound. We have already seen that the main contribution to the error comes from the tail of the distribution. As the tail of the distribution is relatively ``flat'' we can simply apply a concentration inequality to argue that with probability $\Omega(1)$, we have this asymptotic contribution for each of the $k$ hash functions. To be precise, for $j\in [n]$ and $\ell \in [k]$ we define $X_\ell^{(j)}=f_j \cdot \left([h_\ell(j)=h_\ell(i)]-\frac{1}{B} \right)$. Note that the variables $(X_\ell^{(j)})_{j\in [n]}$ are independent. We also define $S_\ell=\sum_{j\in N_2} X_\ell^{(j)}$ for $\ell\in [k]$. Observe that $|X_\ell^{(j)}|\leq f_j \leq \frac{1}{B}$ for $j \geq B$, $\E[X_\ell^{(j)}]=0$, and that
\begin{align*}
\Var[S_\ell]=\sum_{j\in N_2} f_j^2 \left(\frac{1}{B}-\frac{1}{B^2} \right)\leq  \frac{1}{B^2}.
\end{align*}
Applying Bennett's inequality(\Cref{thm:Bennett} of~\Cref{concentration}), with $\sigma^2= \frac{1}{B^2}$ and $M=1/B$ thus gives that 
\begin{align*}
\Pr[S_\ell\leq -t]\leq \exp\left(-h\left(tB \right)\right).
\end{align*}
Defining $W_\ell=\sum_{j\in N_2} f_j \cdot [h_\ell(j)=h_\ell(i)]$ it holds that $\E[W_\ell]=\Theta \left( \frac{ \log \left( \frac{n}{B}\right)}{B}\right)$ and $S_\ell =W_\ell-\E[W_\ell]$, so putting $t=\E[W_\ell]/2$ in the inequality above we obtain that 
\begin{align*}
\Pr[W_\ell\leq \E[W_\ell]/2]=\Pr[S_\ell\leq -\E[W_\ell]/2]\leq \exp\left(-h\left(\Omega\left( \log \frac{n}{B} \right) \right) \right).
\end{align*}
Appealing to~\Cref{asymptotics} and using that $B\leq n/k$ the above bound becomes
\begin{align}\label{bound}
\Pr[W_\ell\leq \E[W_\ell]/2]
&\leq \exp \left(- \Omega \left( \log \frac{n}{B} \cdot \log \left(\log \frac{n}{B}+1 \right) \right) \right) \nonumber \\
&= \exp(-\Omega (\log k \cdot \log ( \log k +1) ))=k^{-\Omega (\log (\log k+1))}.
\end{align}
By the independence of the events $(W_\ell> E[W_\ell]/2)_{\ell\in [k]}$, we have that
\begin{align*}
\Pr\left[|\tilde f_i- f_i|\geq \frac{\E[W_\ell]}{2} \right]\geq (1-k^{-\Omega (\log (\log k+1))})^k=\Omega(1),
\end{align*}
and so $\E[|\tilde f_i- f_i|]=\Omega(\E[W_\ell])=\Omega\left(\frac{\log\left( \frac{n}{B}\right)}{B}\right)$, as desired.
\end{proof}
\begin{remark}\label{remark:lowindependence} We have stated~\Cref{thm:simplecm} for truly random hash functions but it suffices with $O(\log B)$-independent hashing to prove the upper bound. Indeed, the only step in which we require high independence is in the union bound in~\Cref{simpleanalysis} over the $\binom{n}{t/s}$ subsets of $[n]$ of size $t/s$. To optimize the bound we had to choose $s=t/\log (tn)$, so that $t/s=\log(tn)$. As we only need to consider values of $t$ with $t\leq \sum_{i=1}^n f_i=O(\log n)$, in fact $t/s=O(\log n)$ in our estimates. Finally, we applied~\Cref{simpleanalysis} with $n=B$ so it follows that $O(\log B)$-independence is enough to obtain our upper bound.
\end{remark}

\section{(Nearly) Tight Bounds for Count-Sketch with Zipfians}\label{sec:countsketch}
In this section we proceed to analyze Count-Sketch for Zipfians either using a single or more hash functions. We start with two simple lemmas which for certain frequencies $(f_i)_{i \in [n]}$ of the items in the stream can be used to obtain respectively good upper and lower bounds on $\E[|\tilde f_i-f_i|]$ in Count-Sketch with a single hash function. We will use these two lemmas both in our analysis of standard and learned Count-Sketch for Zipfians.

\begin{lemma}\label{explemma}
Let $w=(w_1,\dots,w_n)\in \R^n$, $\eta_1,\dots,\eta_n$ Bernoulli variables taking value $1$ with probability $p$, and $\sigma_1,\dots ,\sigma_n\in \{-1,1\}$ independent Rademachers, i.e., $\Pr[\sigma_i=1]=\Pr[\sigma_i=-1]=1/2$. Let $S=\sum_{i=1}^n w_i \eta_i \sigma_i$. Then, $\E[|S|]= O \left(\sqrt{p} \|w\|_2 \right)$.
\end{lemma}
\begin{proof}
Using that $\E[\sigma_i\sigma_j]=0$ for $i\neq j$ and Jensen's inequality $\E[|S|]^{2}\leq \E[S^2]=\E \left[\sum_{i=1}^n w_i^2\eta_i \right]=p \|w \|_2^2$, from which the result follows.
\end{proof}

\begin{lemma}\label{lowerexplemma}
Suppose that we are in the setting of~\Cref{explemma}. Let $I\subset [n]$ and let $w_I\in \R^n$ be defined by $(w_I)_i=[i\in I]\cdot w_i$. Then
\begin{align*}
\E[|S|]\geq \frac{1}{2}p\left(1-p \right)^{|I|-1} \|w_I\|_1.
\end{align*}
\end{lemma}
\begin{proof}
Let $J=[n]\setminus I$, $S_1=\sum_{i\in I} w_i \eta_i \sigma_i$, and $S_2=\sum_{i\in J} w_i \eta_i \sigma_i$. Let $E$ denote the event that $S_1$ and $S_2$ have the same sign or $S_2=0$. Then $\Pr[E]\geq 1/2$ by symmetry. For $i\in I$ we denote by $A_i$ the event that $\{j\in I: \eta_j\neq 0\}=\{i\}$. Then $\Pr[A_i]=p(1-p)^{|I|-1}$ and furthermore $A_i$ and $E$ are independent.  If $A_i\cap E$ occurs, then $|S|\geq |w_i|$ and as the events $(A_i \cap E)_{i\in I}$ are disjoint it thus follows that $\E[|S|]\geq \sum_{i\in I} \Pr[A_i\cap E] \cdot|w_i|\geq  \frac{1}{2}p\left(1-p \right)^{|I|-1} \|w_I\|_1$.
\end{proof}
With these tools in hand, we proceed to analyse Count-Sketch for Zipfians with one and more hash functions in the next two sections.
\subsection{One hash function}
By the same argument as in the discussion succeeding~\Cref{thm:simplecm}, the following theorem yields the desired result for a single hash function as presented in~\Cref{tbl:results}.
\begin{theorem}\label{unlearnedcs1}
Suppose that $B\leq n$ and let $h:[n] \to [B]$ and $s:[n] \to \{-1,1\}$ be truly random hash functions. Define the random variable $\tilde{f_i}=\sum_{j\in [n]} [h(j)=h(i)]s(j)f_j$ for $i \in [n]$. 
Then 
\begin{align*}
\E[|\tilde f_i-s(i)f_i|]=\Theta \left( \frac{ \log B}{B} \right).
\end{align*}
\end{theorem}

\begin{proof}
Let $i\in [n]$ be fixed. 
We start by defining $N_1=[B] \setminus\{i\}$ and $N_2=[n] \setminus ([B]\cup \{i\})$ and note that
\begin{align*}
|\tilde {f_i}-s(i)f_i|
 \leq \left|\sum_{j\in N_1}[h(j)=h(i)] s(j)f_j\right|+\left|\sum_{j\in N_2}[h(j)=h(i)] s(j)f_j\right|:=X_1+X_2.
\end{align*}
Using the triangle inequality $\E[X_1]\leq \frac{1}{B}\sum_{j \in N_1} f_j = O(\frac{\log B}{ B})$.
Also, by~\Cref{explemma}, $\E[X_2]=O\left( \frac{1}{B}\right)$ and combining the two bounds we obtain the desired upper bound.
For the lower bound we apply~\Cref{lowerexplemma} with $I=N_1$ concluding that
\begin{align*}
\E[|\tilde {f_i}-s(i)f_i|]\geq \frac{1}{2B} \left(1-\frac{1}{B} \right)^{|N_1|-1}\sum_{i\in N_1} f_i=\Omega \left(\frac{\log B}{B} \right).
\end{align*}

\end{proof}
\subsection{Multiple hash functions}
Let $k\in \N$ be odd. For a tuple $x=(x_1,\dots,x_k)\in \R^k$ we denote by $\median x$ the median of the entries of $x$. The following theorem immediately leads to the result on CS with $k\geq 3$ hash functions claimed in~\Cref{tbl:results}.

\begin{theorem}\label{unlearnedmultiCS}
Let $k\geq 3$ be odd, $n\geq kB$, and $h_1,\dots,h_k:[n] \to [B]$ and $s_1,\dots,s_k:[n] \to \{-1,1\}$ be truly random hash functions. Define~$\tilde{f_i}=\median_{\ell\in[k]}\left(\sum_{j\in [n]} [h_\ell(j)=h_\ell(i)]s_\ell(j)f_j\right)$ for $i \in [n]$. 
Assume that\footnote{This very mild assumption can probably be removed at the cost of a more technical proof. In our proof it can even be replaced by $k\leq B^{2-\eps}$ for any $\eps=\Omega(1)$.} $k\leq B$. 
Then 
\begin{align*}
\E[|\tilde {f_i}-s(i)f_i|]=\Omega\left(\frac{1}{B \sqrt k \log k} \right), \quad \text{and} \quad \E[|\tilde {f_i}-s(i)f_i|]=O\left(\frac{1}{B\sqrt{k}}\right)
\end{align*}
\end{theorem}
The assumption $n\geq kB$ simply says that the total number of buckets is upper bounded by the number of items. Again using linearity of expectation for the summation over $i\in [n]$ and replacing $B$ by $B/k$ we obtain the claimed upper and lower bounds of $\frac{\sqrt k}{B \log k}$ and $\frac{\sqrt k}{B}$ respectively. 
We note that even if the bounds above are only tight up to a factor of $\log k$ they still imply that it is asymptotically optimal to choose $k=O(1)$, e.g. $k=3$. To settle the correct asymptotic growth is thus of merely theoretical interest.

In proving the upper bound in~\Cref{unlearnedmultiCS}, we will use the following result by Minton and Price (Corollary 3.2 of~\cite{minton2014improved}) proved via an elegant application of the Fourier transform.
\begin{lemma}[Minton and Price~\cite{minton2014improved}]\label{lemma:minton}
Let $\{X_i: i \in [n]\}$ be independent symmetric random variables such that $\Pr[X_i=0]\geq 1/2$ for each $i$. Let $X=\sum_{i=1}^n X_i$ and $\sigma^2=\E[X^2]=\Var[X]$. For $\eps<1$ it holds that $\Pr[|X|<\eps \sigma]=\Omega(\eps)$
\end{lemma}

\begin{proof}[Proof of~\Cref{unlearnedmultiCS}]
If $B$ (and hence $k$) is a constant, then the results follow easily from~\Cref{explemma}, so in what follows we may assume that $B$ is larger than a sufficiently large constant. We subdivide the exposition into the proofs of the upper and lower bounds.

\paragraph{Upper bound} 
Define $N_1=[B] \setminus \{i\}$ and $N_2=[n]\setminus ([B]\cup \{i\})$. Let for $\ell \in [k]$, $X_1^{(\ell)}=\sum_{j\in N_1}[h_\ell(j)=h_\ell(i)]s_\ell(j)f_j$ and $X_2^{(\ell)}=\sum_{j\in N_2} [h_\ell(j)=h_\ell(i)]s_\ell(j)f_j$ and let $X^{(\ell)}=X_1^{(\ell)}+X_2^{(\ell)}$.

As the absolute error in Count-Sketch with one pair of hash functions $(h,s)$ is always upper bounded by the corresponding error in Count-Min with the single hash function $h$, we can use the bound in the proof of~\Cref{simpleanalysis} to conclude that $\Pr[|X_1^{(\ell)}|\geq t] = O(\frac{\log (tB)}{tB})$, 
when $t\geq 3/B$. Also $\Var[X_2^{(\ell)}] = (\frac{1}{B}-\frac{1}{B^2}) \sum_{j\in N_2}f_j^2\leq \frac{1}{B^2}$,
so by Bennett's inequality (\Cref{thm:Bennett}) with $M=1/B$ and $\sigma^2=1/B^2$ and~\Cref{asymptotics},
\begin{align*}
\Pr[|X_2^{(\ell)}|\geq t]\leq 2\exp \left(-h(tB)\right)\leq 2\exp \left(-\frac{1}{2} tB \log \left(tB+1\right)\right)=O \left( \frac{\log (tB)}{tB} \right),
\end{align*}
for $t\geq \frac{3}{B}$. It follows that for $t\geq 3/B$, 
\begin{align*}
\Pr[|X^{(\ell)}|\geq 2t]\leq \Pr[(|X_1^{(\ell)}|\geq t)] +\Pr(|X_2^{(\ell)}|\geq t)]=O \left( \frac{\log (tB)}{tB} \right).
\end{align*}
Let $C$ be the implicit constant in the $O$-notation above. If $|\tilde{f_i}-s(i)f_i|\geq 2t$, at least half of the values $(|X^{(\ell)}|)_{\ell \in [k]}$ are at least $2t$. For $t\geq 3/B$ it thus follows by a union bound that 
\begin{align}\label{eq:tbig}
\Pr[|\tilde{f_i}-s(i)f_i|\geq 2t ]\leq 2 \binom{k}{\lceil k/2 \rceil}\left( C\frac{\log (tB)}{tB} \right)^{\lceil k/2\rceil}\leq 2 \left( 4C\frac{\log (tB)}{tB} \right)^{\lceil k/2\rceil}.
\end{align}
If $\alpha =O(1)$ is chosen sufficiently large it thus holds that
\begin{align*}
\int_{\alpha/B}^\infty \Pr[|\tilde{f_i}-s(i)f_i|\geq t ] \, dt
&= 2\int_{\alpha/(2B)}^\infty \Pr[|\tilde{f_i}-s(i)f_i|\geq 2t ] \, dt \\
&\leq \frac{4}{B}\int_{\alpha/2}^\infty \left( 4C\frac{\log (t)}{t} \right)^{\lceil k/2\rceil} \, dt \\
&\leq \frac{1}{B2^{k}}\leq  \frac{1}{B\sqrt{k}}.
\end{align*}
Here the first inequality uses~\cref{eq:tbig} and a change of variable. The second inequality uses that $\left(4C \frac{\log t}{t}\right)^{\lceil k/2 \rceil}\leq (C'/t)^{2k/5}$ for some constant $C'$ followed by a calculation of the integral.
Now,
$$
\E[|\tilde{f_i}-s(i)f_i|]=\int_{0}^\infty \Pr[|\tilde{f_i}-s(i)f_i|\geq t ] \, dt,
$$
 so for our upper bound it therefore suffices to show that $\int_{0}^{\alpha/B} \Pr[|\tilde{f_i}-s(i)f_i|\geq t ] \, dt=O\left( \frac{1}{B \sqrt{k}}\right)$. For this we need the following claim:
\begin{claim}\label{intervalclaim}
Let $I\subset \R$ be the closed interval centered at the origin of length $2t$, i.e., $I=[-t,t]$. Suppose that $0< t\leq \frac{1}{2B}$. For $\ell \in [k]$, $\Pr[X^{(\ell)}\in I]= \Omega (tB)$.
\end{claim}
\begin{proof}
Note that $\Pr[X_1^{(\ell)}=0]\geq \Pr[\bigwedge_{j\in N_1} (h_{\ell}(j)\neq h_{\ell}(i))] = (1-{1\over B})^{N_1}=\Omega(1)$. 
Secondly $\Var[X_2^{(\ell)}]= (\frac{1}{B}-\frac{1}{B^2})\sum_{j \in N_2}f_j^2\leq \frac{1}{B^2}$.
Using that $X_1^{(\ell)}$ and $X_2^{(\ell)}$ are independent and~\Cref{lemma:minton} with $\sigma^2=\Var[X_2^{(\ell)}]$, it follows that $\Pr[X^{(\ell)}\in I]=\Omega \left(\Pr[X_2^{(\ell)}\in I] \right)=\Omega(tB)$.
\end{proof}
Let us now show how to use the claim to establish the desired upper bound. For this let $0<t \leq \frac{1}{2B}$ be fixed. If $|\tilde{f_i}-s(i)f_i|\geq t$, at least half of the values $(X^{(\ell)})_{\ell \in [k]}$ are at least $t$ or at most $-t$. Let us focus on bounding the probability that at least half are at least $t$, the other bound being symmetric giving an extra factor of $2$ in the probability bound. By symmetry and~\Cref{intervalclaim}, $\Pr[X^{(\ell)}\geq t]=\frac{1}{2}-\Omega(tB)$. For $\ell\in [k]$ we define $Y_{\ell}=[X^{(\ell)}\geq t]$, and we put $S=\sum_{\ell \in [k]}Y_\ell$. Then $\E[S]=k \left(\frac{1}{2}-\Omega(tB) \right)$. If at least half of the values $(X^{(\ell)})_{\ell \in [k]}$ are at least $t$ then $S\geq k/2$. By Hoeffding's inequality (\Cref{thm:Hoeffding}) we can bound the probability of this event by
\begin{align*}
\Pr[S\geq k/2]=\Pr[S-\E[S]=\Omega(ktB)]=\exp(-\Omega(kt^2B^2)).
\end{align*}
It follows that $\Pr[|\tilde{f_i}-s(i)f_i|\geq t]\leq 2\exp(-\Omega(kt^2B^2))$. Thus 
\begin{align*}
\int_{0}^{\alpha/B} \Pr[|\tilde{f_i}-s(i)f_i|\geq t ] \, dt 
&\leq \int_{0}^{\frac{1}{2B}} 2\exp(-\Omega(kt^2B^2)) \, dt+ \int_{\frac{1}{2B}}^{\alpha/B}2\exp(-\Omega(k)) \,dt \\
&\leq\frac{1}{B \sqrt{k}}\int_0^{\sqrt{k}/2} \exp(-t^2)\, dt+\frac{2\alpha \exp(-\Omega(k))}{B}=O\left(\frac{1}{B\sqrt{k}}\right).
\end{align*}
Here the second inequality used a change of variable. The proof of the upper bound is complete.

\paragraph{Lower Bound} 
Fix $\ell \in [k]$ and let $M_1=[B \log k]\setminus \{i\}$ and $M_2=[n]\setminus ([B \log k] \cup \{i\})$. Write
\begin{align*}
S:=\sum_{j\in M_1} [h_\ell(j)=h_\ell(i)]s_\ell(j)f_j
+\sum_{j\in M_2} [h_\ell(j)=h_\ell(i)]s_\ell(j)f_j 
:= S_1+S_2.
\end{align*}
We also define $J:=\{j\in M_1: h_\ell(j)=h_{\ell}(i) \}$. Let $I\subseteq \R$ be the closed interval around $s_\ell(i)f_i$ of length $\frac{1}{ B \sqrt{k}\log k}$. We now upper bound the probability that $S\in I$ conditioned on the value of $S_2$. To ease the notation, the conditioning on $S_2$ has been left out in the notation to follow. Note first that
\begin{align}\label{eq:somethingtobound}
\Pr[S\in I ]=\sum_{r=0}^{|M_1|} \Pr[S\in I \mid |J|=r] \cdot \Pr[|J|=r].
\end{align}
For a given $r\geq 1$ we now proceed to bound $\Pr[S\in I \mid |J|=r]$. This probability is the same as the probability that $S_2+\sum_{j\in R} \sigma_jf_j\in I$, where $R\subseteq M_1$ is a uniformly random $r$-subset and the $\sigma_j$'s are independent Rademachers. Suppose that we sample the elements from $R$ as well as the corresponding signs $(\sigma_i)_{i \in R}$ sequentially, and let us condition on the values and signs of the first $r-1$ sampled elements. At this point at most $\frac{B \log k}{\sqrt{k}}+1 $ possible samples for the last element in $R$ can cause that $S \in I$. Indeed, the minimum distance between distinct elements of $\{f_j: j \in M_1\}$ is at least $1/(B\log k)^2$ and furthermore $I$ has length $\frac{1}{B\sqrt{k} \log k}$. Thus, at most
\begin{align*}
 \frac{1}{ B \sqrt{k}\log k} \cdot (B\log k)^2+1 = \frac{B \log k}{\sqrt{k}}+1
\end{align*}
 choices for the last element of $R$ ensure that $S\in I$.  For $1\leq r \leq (B \log k) /2$ we can thus upper bound
\begin{align*}
\Pr[S\in I \mid |J|=r]\leq \frac{\frac{B \log k}{\sqrt{k}}+1}{|M_1| -r+1}\leq \frac{2}{\sqrt{k}}+\frac{2}{B \log k}\leq \frac{3}{\sqrt{k}}.
\end{align*}
Note that $\mu:=\E[|J|]\leq \log k$ so for $B\geq 6$, it holds that 
\begin{align*}
\Pr[|J| \geq  (B \log k) /2]\leq \Pr\left[|J|\geq \mu\frac{B}{2}\right]\leq \Pr\left[|J|\geq \mu \left(1+\frac{B}{3}\right)\right]\leq  \exp \left( -\mu h(B/3) \right)=k^{-\Omega (h(B/3))},
\end{align*}
where the last inequality follows from the Chernoff bound of~\Cref{thm:Chernoff}. Thus, if we assume that $B$ is larger than a sufficiently large constant, then $\Pr[|J|\geq B \log k /2]\leq k^{-1}$. Finally, $\Pr[|J|=0]=(1-1/B)^{B\log k}\leq k^{-1}$. Combining the above, we can continue the bound in~\eqref{eq:somethingtobound} as follows. 
\begin{align}
\Pr[S\in I ] 
\leq &  \Pr[|J|=0]+\sum_{r=1}^{(B \log k)/2}   \Pr[S\in I \mid |J|=r] \cdot \Pr[|J|=r] \nonumber\\
+&\sum_{r=(B \log k) /2+1}^{|M_1|} \Pr[|J|=r]= O \left(\frac{1}{\sqrt{k}} \right), \label{eq:SinI}
\end{align}
which holds even after removing the conditioning on $S_2$. We now show that with probability $\Omega (1)$ at least half the values $(X^{(\ell)})_{\ell \in [k]}$ are at least $\frac{1}{2B \sqrt k \log k}$. Let $p_0$ be the probability that $X^{(\ell)}\geq \frac{1}{2B \sqrt k \log k}$. This probability does not depend on $\ell \in [k]$ and by symmetry and~\eqref{eq:SinI}, $p_0=1/2-O(1/\sqrt{k})$. Define the function $f:\{0,\dots,k\} \to \R$ by
\begin{align*}
f(t)=\binom{k}{t} p_0^t (1-p_0)^{k-t}.
\end{align*}
Then $f(t)$ is the probability that exactly $t$ of the values $(X^{(\ell)})_{\ell \in [k]}$ are at least $\frac{1}{B \sqrt k \log k}$. Using that $p_0=1/2-O(1/\sqrt{k})$, a simple application of Stirling's formula gives that $f(t)=\Theta \left( \frac{1}{\sqrt{k}} \right)$ for $t=\lceil k/2\rceil,\dots,\lceil k/2+\sqrt{k}\rceil$ when $k$ is larger than some constant $C$. 
It follows that with probability $\Omega(1)$ at least half of the $(X^{(\ell)})_{\ell \in [k]}$ are at least $\frac{1}{B \sqrt k \log k}$ and in particular 
\begin{align*}
\E[|\tilde {f_i}-f_i|]=\Omega\left(\frac{1}{B \sqrt k \log k} \right).
\end{align*}
Finally we handle the case where $k\leq C$. It follows from simple calculations (e.g., using~\Cref{lowerexplemma}) that $X^{(\ell)}=\Omega(1/B)$ with probability $\Omega(1)$. Thus this happens for all $\ell \in [k]$ with probability $\Omega(1)$ and in particular $\E[|\tilde {f_i}-f_i|]= \Omega(1/B)$, which is the desired for constant $k$.
\end{proof}

\section{Learned Count-Sketch for Zipfians}\label{sec:learnedcountsketch}
We now proceed to analyze the learned Count-Sketch algorithm. In~\Cref{sec:learned1} we estimate the expected error when using a single hash function and in~\Cref{sec:learnedk} we show that the expected error only increases when using more hash functions. Recall that we assume that the number of buckets $B_h$ used to store the heavy hitters that $B_h=\Theta(B-B_h)=\Theta(B)$. 
\subsection{One hash function}\label{sec:learned1}
By taking $B_1=B_h=\Theta(B)$ and $B_2=B-B_h=\Theta(B)$ in the theorem below, the result on L-CS for $k=1$ claimed in~\Cref{tbl:results} follows immediately. 

\begin{theorem}\label{thm:lcs1}
Let $h:[n]\setminus [B_1] \to [B_2]$ and $s:[n] \to \{-1,1\}$ be truly random hash functions where $n,B_1,B_2 \in \N$ and\footnote{The first inequality is the standard assumption that we have at least as many items as buckets. The second inequality says that we use at least as many buckets for non-heavy items as for heavy items (which doesn't change the asymptotic space usage).} $n-B_1\geq B_2\geq B_1$. Define the random variable $\tilde{f_i}=\sum_{j=B_1+1}^n [h(j)=h(i)]s(j)f_j$ for $i \in [n]\setminus [B_1]$. Then 
\begin{align*}
\E[|\tilde f_i-s(i)f_i|]=\Theta  \left(\frac{\log \frac{B_2+B_1}{B_1}}{B_2} \right)
\end{align*}
\end{theorem}
\begin{proof}
Let $N_1=[B_1+B_2]\setminus ([B_1]\cup \{i\})$ and $N_2=[n]\setminus ([B_1+B_2]\cup \{i\})$.
Let $X_1=\sum_{j\in N_1} [h(j)=h(i)]s(j)f_j$ and $X_2=\sum_{j\in N_2} [h(j)=h(i)]s(j)f_j$.
By the triangle inequality and linearity of expectation, 
\begin{align*}
\E[|X_1|]= O \left(\frac{\log \frac{B_2+B_1}{B_1}}{B_2} \right).
\end{align*}
Moreover, it follows directly from~\Cref{explemma} that  $\E\left[|X_2| \right]=O\left( \frac{1}{B_2} \right)$. Thus 
\begin{align*}
\E[|\tilde f_i-s(i)f_i|]\leq \E[|X_1|]+\E[|X_2|]=O \left(\frac{\log \frac{B_2+B_1}{B_1}}{B_2} \right),
\end{align*}
as desired. For the lower bound on $\E\left[\left|\tilde {f_i}-s(i)f_i \right| \right]$ we apply~\Cref{lowerexplemma} with $I=N_1$ to obtain that,
\begin{align*}
\E\left[\left|\tilde {f_i}-s(i)f_i \right| \right]\geq \frac{1}{2B_2} \left(1-\frac{1}{B_2} \right)^{|N_1|-1}\sum_{i\in N_1}f_i=\Omega \left(\frac{\log \frac{B_2+B_1}{B_1}}{B_2} \right).
\end{align*}
\end{proof}

\begin{corollary}\label{cro:lcs1}
Let $h:[n]\setminus [B_h] \to [B-B_h]$ and $s:[n] \to \{-1,1\}$ be truly random hash functions where $n,B,B_h \in \N$ and $B_h = \Theta(B) \leq B/2$. Define the random variable $\tilde{f_i}=\sum_{j=B_h+1}^n [h(j)=h(i)]s(j)f_j$ for $i \in [n]\setminus [B_h]$. Then $\E[|\tilde f_i-s(i)f_i|]=\Theta(1 / B)$.
\end{corollary}
\begin{remark}
The upper bounds of~\Cref{thm:lcs1} and~\Cref{cro:lcs1} hold even without the assumption of fully random hashing. In fact, we only require that $h$ and $s$ are $2$-independent. Indeed~\Cref{explemma} holds even when the Rademachers are $2$-independent (the proof is the same). Moreover, we need $h$ to be $2$-independent as we condition on $h(i)$ in our application of~\Cref{explemma}. With $2$-independence the variables $[h(j)=h(i)]$ for $j\neq i$ are then Bernoulli variables taking value $1$ with probability $1/B_2$.
\end{remark}


\subsection{More hash functions}\label{sec:learnedk}
We now show that, like for Count-Sketch, using more hash functions does not decrease the expected error. We first state the Littlewood-Offord lemma as strengthened by Erd\H{o}s.
\begin{theorem}[Littlewood-Offord~\cite{littlewood1939number}, Erd\H{o}s~\cite{erdos1945lemma}]\label{LittlewoodOfford} 
Let $a_1,\dots,a_n\in \R$ with $|a_i|\geq1$ for $i\in [n]$. Let further  $\sigma_1,\dots,\sigma_n\in \{-1,1\}$ be random variables with $\Pr[\sigma_i=1]=\Pr[\sigma_i=-1]=1/2$ and define $S=\sum_{i=1}^n \sigma_i a_i$. For any $v\in \R$ it holds that $\Pr[|S-v|\leq 1] = O(1/\sqrt{n})$.
\end{theorem}

Setting $B_1=B_h=\Theta(B)$ and $B_2=B-B_2=\Theta(B)$ in the theorem below gives the final bound from~\Cref{tbl:results} on L-CS with $k\geq 3$.
\begin{theorem}
Let $n\geq B_1+B_2\geq 2B_1$, $k\geq 3$ odd, and $h_1,\dots,h_k:[n] \setminus [B_1] \to [B_2/k]$ and $s_1,\dots,s_k:[n] \setminus [B_1]\to \{-1,1\}$ be independent and truly random. Define the random variable $\tilde{f_i}=\median_{\ell\in[k]}\left(\sum_{j\in [n]\setminus [B_1]} [h_\ell(j)=h_\ell(i)]s_\ell(j)f_j\right)$ for $i \in [n]\setminus [B_1]$. 
Then 
\begin{align*}
\E[|\tilde f_i-s(i)f_i|]=\Omega\left(\frac{1}{B_2}\right).
\end{align*}
\end{theorem}
\begin{proof}
Like in the proof of the lower bound of~\Cref{unlearnedmultiCS} it suffices to show that for each $i$ the probability that the sum $S_\ell:=\sum_{j\in [n]\setminus ([B_1]\cup \{i\})} [h_\ell(j)=h_\ell(i)]s_\ell(j)f_j$ lies in the interval $I=\left[-1/(2B_2),1/(2B_2) \right]$ is $O(1/\sqrt{k})$. Then at least half the $(S_\ell)_{\ell \in [k]}$ are at least $1/(2B_2)$ with probability $\Omega(1)$ by an application of Stirling's formula, and it follows that $\E[|\tilde f_i-s(i)f_i|]=\Omega (1/B_2)$. 

Let $\ell \in [k]$ be fixed, $N_1=[2B_2]\setminus ([B_2]\cup \{i\})$, and $N_2=[n] \setminus (N_1\cup \{i\})$, and write 
\begin{align*}
S_{\ell}=\sum_{j\in N_1}[h_\ell(j)=h_\ell(i)]s_\ell(j)f_j+\sum_{j\in N_2}[h_\ell(j)=h_\ell(i)]s_\ell(j)f_j:=X_1+X_2.
\end{align*}
Now condition on the value of $X_2$. Letting $J=\{j\in N_1:h_\ell(j)=h_\ell(i)\}$ it follows by~\Cref{LittlewoodOfford} that 
\begin{align*}
\Pr[S_\ell \in I \mid X_2]=O\left( \sum_{J'\subseteq N_1}\frac{\Pr[J=J']}{\sqrt{|J'|+1}} \right)=O\left( \Pr[|J|< k/2]+1/\sqrt{k}\right).
\end{align*}
An application of Chebyshev's inequality gives that $\Pr[|J|< k/2]=O(1/k)$, so $\Pr[S_\ell\in I]=O(1/\sqrt{k})$. Since this bound holds for any possible value of $X_2$ we may remove the conditioning and the desired result follows.
\end{proof}
\begin{remark} 
The bound above is probably only tight for $B_1=\Theta(B_2)$. Indeed, we know that it cannot be tight for all $B_1\leq B_2$ since when $B_1$ becomes very small, the bound from the standard Count-Sketch with $k\geq 3$ takes over --- and this is certainly worse than the bound in the theorem. It is an interesting open problem (that requires a better anti-concentration inequality than the Littlewood-Offord lemma) to settle the correct bound when $B_1 \ll B_2$. 
\end{remark}

\section{Proof of Theorem \ref{thm:main_unlearned}}\label{sec:proof_of_thm_main_unlearned}

In this section we give the complete proof of Theorem \ref{thm:main_unlearned}. We need the following special case of a result about the behaviour of CountSketch, proved in the prior sections. 

\begin{theorem}[Theorem \ref{unlearnedmultiCS}]\label{thm:cs_tail}
Let $\hat{f}_i$ be the estimate of the $i$th frequency given by a $3 \times B/3$ CountSketch table. There exists a universal constant $C$ such that the following two inequalities hold:
\begin{align}
    \Pr\left( | f_i - \hat{f}_i^j | \ge \frac{C}B \right) &\le \frac{1}2, \\
     \forall t \ge 3/B, \quad \Pr\left( | f_i - \hat{f}_i^j | \ge t \right) & \le C\left(\frac{\log(tB)}{tB} \right)^2.
\end{align}
\end{theorem}


\begin{proof}[Proof of Theorem \ref{thm:main_unlearned}]\label{sec:proof_main_unlearned}
\noindent \textbf{Case 1:} $i > B/\log \log n$. 
Recall that $\hat{f}_i^j$ denotes the estimate of the $i$th frequency given by table $S_j$ in Algorithm \ref{alg:main_unlearned}. Furthermore, $\tilde{f}_i \gets \text{Median}(\hat{f}_i^1, \ldots, \hat{f}_i^{T-1})$ denotes the median of the estimates of the first $T-1$ tables in Algorithm \ref{alg:main_unlearned}. From Theorem \ref{thm:cs_tail}, we have that for every fixed $j$,
\[ \Pr\left( |f_i - \hat{f}_i^j | \ge \frac{2C \log \log n}{B} \right) \le \frac{1}4  \]
and so it follows that 
\begin{equation}\label{eq:main_thm_tail_bound}
    \Pr\left( |f_i - \tilde{f}_i  | \ge \frac{2C \log \log n}{B} \right) \le \exp(- \Omega(T)) \le \frac{1}{(\log n)^{100}}  
\end{equation}

by adjusting the constant in front of $T$. We let $2C$ be the constant for the $O$ notation in line $7$ of Algorithm \ref{alg:main_unlearned}. Now consider the expected value of $|\hat{f}_i - 1/i|$, where the expectation is taken over the randomness used by the CountSketch tables of Algorithm \ref{alg:main_unlearned}. By conditioning on the event that we either output $0$ or output the estimate of the $T$th table, we have
\[
    \mathbb{E}\left[ \frac{1}i \cdot \left | \hat{f}_i - \frac{1}i \right| \right] \lesssim \Pr(\text{We output $0$}) \cdot \frac{1}{i^2} + \Pr(\text{We output estimate of table $T$}) \cdot \frac{1}{iB}
\]
where we have used the first inequality in Theorem \ref{thm:cs_tail} in the above inequality and $\lesssim$ denotes inequality up to a constant factor. We have bounded the second probability in Equation \eqref{eq:main_thm_tail_bound} which gives
\begin{equation}\label{eq:estimate_1}
     \mathbb{E}\left[ \frac{1}i \cdot \left | \hat{f}_i - \frac{1}i \right| \right] \lesssim \frac{1}{i^2} + \frac{1}{i B \cdot (\log n)^{99}}.
\end{equation}
\\

\noindent \textbf{Case 2:} $i \le B/(\log \log n)^4$. 
We employ the more refined tail bound for Count Sketch stated in the second inequality of Theorem \ref{thm:cs_tail}.

For any $i$ smaller than $B/(\log \log n)^4$, we have that for any fixed $j$,
\[ \Pr\left(  \hat{f}_i^j  \le \frac{2C \log \log n}{B}\right) \le  \Pr\left( |f_i - \hat{f}_i^j  |  \ge \frac{1}{2i}\right) \lesssim \left( \frac{\log(B'/i) \cdot i}{B'} \right)^2\]
where $B' = B/(4T) = \Theta(B/\log \log n)$.
It follows that
\[ \Pr\left( \tilde{f}_i  \le \frac{2C \log \log n}{B} \right) \le T \cdot \Pr\left( |f_i - \hat{f}_i^1  |  \le \frac{2C \log \log n}{B}\right) \lesssim T \cdot  \left( \frac{\log(B'/i) \cdot i}{B'} \right)^2. \]
Therefore, for $i \le B/(\log \log n)^4$, we again have
\begin{align*}
    \mathbb{E}\left[ \frac{1}i \cdot \left | \hat{f}_i - \frac{1}i \right| \right] &\lesssim \Pr(\text{We output $0$}) \cdot \frac{1}{i^2} + \Pr(\text{We output estimate of table $T$}) \cdot \frac{1}{iB} \\
    &\lesssim (\log \log n)^3 \cdot \left( \frac{\log(B/i)}{B} \right)^2 + \frac{1}{iB}.
\end{align*}
We can summarize this case as:

\begin{equation}\label{eq:estimate_2}
     \mathbb{E}\left[ \frac{1}i \cdot \left | \hat{f}_i - \frac{1}i \right| \right] \lesssim (\log \log n)^3 \cdot \left( \frac{\log(B/i)}{B} \right)^2 + \frac{1}{iB}.
\end{equation}
\\

\noindent \textbf{Putting everything together.}
Equation \eqref{eq:estimate_1} gives us 
\begin{align*}
    \frac{1}{\log n} \cdot \sum_{i > B/\log \log n}\mathbb{E}\left[ \frac{1}i \cdot \left | \hat{f}_i - \frac{1}i \right| \right] &\lesssim \frac{1}{\log n} \sum_{i > B/\log \log n} \frac{1}{i^2} + \frac{1}{B \cdot (\log n)^{100}} \, \sum_{i = 1}^n \frac{1}i \\
    &\lesssim \frac{\log \log n}{B \log n}.
\end{align*}
Equation \eqref{eq:estimate_2} gives us 
\begin{align*}
    \frac{1}{\log n} \cdot \sum_{i \le B/(\log \log n)^4}\mathbb{E}\left[ \frac{1}i \cdot \left | \hat{f}_i - \frac{1}i \right| \right] &\lesssim \frac{1}{\log n}\sum_{i \le B/(\log \log n)^4 } \left( (\log \log n)^3 \cdot \left( \frac{\log (B/i)}{B} \right)^2 + \frac{1}{iB} \right) \\
    &\lesssim \frac{(\log \log n)^3}{B^2 \log n}\int_1^{B/(\log \log n)^4} \log(B/x)^2 \ dx + \frac{ \log B}{B \log n} \\
    &\lesssim \frac{(\log \log n)^3}{B^2 \log n} \cdot \frac{B  (\log^2(\log \log n))}{(\log \log n)^4}  + \frac{ \log B}{B \log n} \\
    &\lesssim \frac{\log B}{B \log n}
\end{align*}
where the second to last inequality follows from the indefinte integral $\int \log^2(c/x) \ dx = x \log^2(c/x) + 2x \log(c/x) + 2x + \text{Constant}$. 

Finally, we deal with the remaining case: $i$ between $B/\log \log n$ and $B/(\log \log n)^4$. For these $i$'s, the worst case error happens when we set their estimates to $0$, incurring error $1/i$, as opposed to incurring error $O(1/B)$ if we used the estimate of table $T$:
\begin{align*}
   &\frac{1}{\log n} \cdot \sum_{B/(\log \log n)^4 \le i \le B/\log \log n} \mathbb{E}\left[ \frac{1}i \cdot \left | \hat{f}_i - \frac{1}i \right| \right]\\
   &\lesssim \frac{1}{\log n} \sum_{B/(\log \log n)^4 \le i\le  B/\log \log n} \frac{1}{i^2} \\
   &\lesssim \frac{(\log \log n)^4}{B \log n} .
\end{align*}

Combining our three estimates completes the proof.
\end{proof}

\section{Proof of Theorem \ref{thm:main_learned}}\label{sec:proof_main_learned}

\begin{proof}[Proof of Theorem \ref{thm:main_learned}]
We summarize the intuition and give the full proof. Recall the workhorse of our analysis is the proof of Theorem \ref{thm:main_unlearned}. First note that we obtain $0$ error for $i < B/2$. Thus, all our error comes from indices $i \ge B/2$. Recall the intuition for this case from the proof of Theorem \ref{thm:main_unlearned}: we want to output $0$ as our estimates. Now the same analysis as in Case $1$ of Theorem \ref{thm:main_unlearned} gives us that the probability we use the estimate of table $T$ can be bounded by say $\frac{1}{(\log n)^{100}}$. Thus, similar to Equation \eqref{eq:estimate_1}, we have
\begin{align*}
     \mathbb{E}\left[ \frac{1}i \cdot \left | \hat{f}_i - \frac{1}i \right| \right] &\lesssim 
     \Pr(\text{We output $0$}) \cdot \frac{1}{i^2} + \Pr(\text{We output estimate of table $T$}) \cdot \frac{1}{iB}\\
     &\lesssim \frac{1}{i^2} + \frac{1}{i B \cdot (\log n)^{99}}.
\end{align*}
Thus, our total error consists of only one part of the total error calculation of Theorem \ref{thm:main_unlearned}:
\begin{align*}
    \frac{1}{\log n} \cdot \sum_{i > B}\mathbb{E}\left[ \frac{1}i \cdot \left | \hat{f}_i - \frac{1}i \right| \right] &\lesssim \frac{1}{\log n} \sum_{i > B} \frac{1}{i^2} + \frac{1}{B \cdot (\log n)^{100}} \, \sum_{i = 1}^n \frac{1}i \\
    &\lesssim \frac{1}{B \log n},
\end{align*}
as desired.
\end{proof}
\begin{algorithm}[!ht]
\caption{\label{alg:update_parsimnious} Parsimonious frequency update algorithm}
\begin{algorithmic}[1]
\State \textbf{Input:} Stream of updates to an $n$ dimensional vector, space budget $B$, access to a heavy hitter oracle which correctly identifies the top $B/2$ heavy hitters.
\Procedure{Update}{}
\State $T \gets O(\log \log n)$
\For{$j = 1$ to $T-1$}
\State $S_j \gets $ CountSketch table with $4$ rows and $\frac{B}{16T}$ columns
\EndFor
\State  $S_T \gets $ CountSketch table with $4$ rows and $\frac{B}{16}$ columns
\For{stream element $(i, \Delta)$}
\If{$i$ is already classified as a top $B/2$ heavy hitter}
\State Maintain the count of $i$ exactly (from the point of time it was detected as heavy).
\Else
\State Query the heavy hitter oracle with probability $p=\min\left(1, CB(\log n)^2\Delta\right)$
\If{$i$ gets queried and is classified as a top $B/2$ heavy hitter}
\State Maintain the count of $i$ exactly (from this point of time).
\Else
\State Input $(i,\Delta)$ in each of the $T$ CountSketch tables $S_j$
\EndIf
\EndIf
\EndFor
\EndProcedure
\end{algorithmic}
\end{algorithm}

\section{Parsimonious learning}\label{app:parsimonious}
In this appendix, we state our result on parsimonious learning precisely. We consider the modification to~\cref{alg:update_learned} where whenever an element $(i,\Delta)$ arrives, we only query the heavy hitter oracle with probability $p=\min\left(1,\gamma B(\log n)^2\Delta\right)$ for $\gamma$ a sufficiently large constant\footnote{This sampling probability depends on the length of the stream which is likely unknown to us. We will discuss how this assumption can be removed shortly.}. To be precise, when an item $i$ arrives, we first check if it is already classified as a top $B/2$ heavy hitter. If so, we update its exact count (from the first point of time where it was classified as heavy). If not, we query the heavy hitter oracle with probability $p$. In case $i$ gets queried and classified as one of the top $B/2$ heavy hitters, we store its count exactly (from this point of time). Otherwise, we input it to the CountSketch tables $S_j$ similarly to~\Cref{alg:update_unlearned} and~\Cref{alg:update_learned}. \Cref{alg:update_parsimnious} shows the pseudocode for the update procedure of our parsimonious learning algorithm. The query procedure is similar to~\Cref{alg:main_learned}. We now state our main result on our parsimonious learning algorithm, namely that it achieves the same expected weighted error bound as in~\Cref{thm:main_learned}.

\begin{theorem}\label{thm:parsimonious}
Consider Algorithm~\ref{alg:update_parsimnious} with space parameter $B \ge \log n$ updated over a Zipfian stream. Suppose the heavy-hitter oracle correctly identifies the top $B/2$ heavy hitters in the stream. Let $\{\hat{f}_i\}_{i=1}^n$ denote the estimates computed by Algorithm~\ref{alg:main_learned}. The expected weighted error \eqref{eq:weighted_error} is
$\mathbb{E}\left[ \frac{1}{N} \cdot \sum_{i=1}^n f_i \cdot |f_i - \hat{f}_i| \right] = O\left( \frac{1}{B \log n} \right). $ The algorithm makes $O(B(\log n)^3)$ queries to the heavy hitter oracle in expectation.
\end{theorem}

\begin{proof}[Proof of~\cref{thm:parsimonious}]
Introducing some notation, we denote the stream $((x_1,\Delta_1),\dots, (x_m,\Delta_m))$. Letting $S_i=\{j\in [m]\mid x_j=i\}$, we then have that $\sum_{j\in S_i}\Delta_j=f_j=1/j$. Then, whenever an element $(x_j,\Delta_j)$ arrives, the algorithm queries the heavy hitter oracle with probability $p_j=\min\left(1,C\gamma B(\log n)^2 \Delta_j\right)$.

Let us first consider the expected error when estimating the frequency of a heavy hitter $i\leq B/2$. Let $j_0\in S_i$ be minimal such that $\sum_{j\in S_i,j\leq j_0}\Delta_j\geq \frac{1}{B\log n}$. Since $i$ is a heavy hitter with total frequency $f_i\geq 2/B$, such a $j_0$ exists. If there exists $j\in S_i$ with $j\leq j_0$ such that $p_j=1$, then $i$ will be classified as a heavy hitter by time $j_0$ with probability $1$. Otherwise, the probability that $i$ is not classified as a heavy hitter by time $j_0$ is upper bounded by
\begin{align*}
\prod_{j\in S_i,j\leq j_0}(1-p_j)\leq & \exp\left(-\sum_{j\in S_i,j\leq j_0} p_j \right)=\exp\left(-\gamma B(\log n)^2\sum_{j\in S_i,j\leq j_0} \Delta_j \right) \\ 
\leq & \exp(-\gamma \log n)=n^{-\gamma}.
\end{align*}
Union bounding over the $B/2$ top heavy hitters we find that with high probability in $n$ they are indeed classified as heavy at the first point of time where they have appeared with weight at least $\frac{1}{B\log n}$. In particular, with the same high probability the error when estimating each of the top $B/2$ heavy hitters is at most $\frac{1}{B\log n}$ and so,
\[
\mathbb{E}\left[ \frac{1}{N} \cdot \sum_{i=1}^{B/2} f_i \cdot |f_i - \hat{f}_i| \right]=O\left(\frac{1}{B\log n}\right).
\]
Let us now consider the light elements $i>B/2$. Such an element is never classified as heavy and consequently is estimated using the CountSketch tables $S_j$ as in~\cref{alg:main_unlearned}. Denoting by $E$ the event that we output $0$ (that is, the median of the first $T-1$ CountSketch tables is small enough) and by $E^c$ the event that we output the estimate from table $T$, as in~\cref{sec:proof_main_learned}, we again have
\[
\mathbb{E}\left[ \frac{1}i \cdot \left | \hat{f}_i - \frac{1}i \right| \right] \lesssim \Pr(E) \cdot \frac{1}{i^2} + \Pr(E^c) \cdot \frac{1}{iB}\leq \frac{1}{i^2}+\Pr(E^c) \cdot \frac{1}{iB}.
\]
Here, the bound of $O(1/B)$ on the expected error of table $T$ holds even though the $B/2$ heavy hitters might appear in table $T$. The reason is with high probability, these heavy hitters appear with weight at most $\frac{1}{B\log n}$ and conditioned on this event, we can plug into~\cref{explemma} to get that the expected error is still $O(1/B)$. It remains to bound $\Pr(E^c)$. Again, from~\cref{explemma}, it follows that the expected error of each of the first $T-1$ tables is at most $C\frac{2\log \log n}{B}$ for a sufficiently large constant $C$ (even including the contribution from the heavy hitters), and so by Markov's inequality,
\[
\Pr\left( |f_i - \hat{f}_i^j | \ge \frac{2C \log \log n}{B} \right) \le \frac{1}{4}
\]
and again,
\[
\Pr\left( |f_i - \tilde{f}_i  | \ge \frac{2C \log \log n}{B} \right) \le \exp(- \Omega(T)) \le \frac{1}{(\log n)^{100}}. 
\]
Thus, we can bound,
\[
\mathbb{E}\left[ \frac{1}i \cdot \left | \hat{f}_i - \frac{1}i \right| \right] \lesssim \frac{1}{i^2}+ \cdot \frac{1}{(\log n)^{100}iB}.
\]
Recalling that $N=H_n$ and summing over $i\geq B/2$ we get that
\[
\mathbb{E}\left[ \frac{1}{N} \cdot \sum_{i=B/2+1}^{n} f_i \cdot |f_i - \hat{f}_i| \right]=O\left(\frac{1}{B\log n}+\frac{\log (n/B)}{B(\log n)^{100}}\right)=O\left(\frac{1}{B\log n}\right),
\]
as desired. The expected number of queries to the heavy hitter oracle is 
\[
\sum_{j=1}^m p_j\leq\sum_{i=1}^n\sum_{j\in S_i} \gamma B(\log n)^2 \Delta_j=\sum_{i=1}^n \gamma B(\log n)^2 f_i=O(B(\log n)^3).
 \]
\begin{remark}
We note that~\cref{alg:update_parsimnious} makes use of the length of the stream to set $p$. Usually we would not know the length of the stream but at the cost of an extra log-factor in the number of queries made to the oracle, we can remedy this. Indeed, the query probability is $p=\min\left(1,\frac{\gamma B(\log n)^3}{m}\right)$ where $m$ is the length of the stream. If we instead increase the query probability after we have seen $j$ stream elements to $p_{j}=\min\left(1,\frac{\gamma B(\log n)^3}{j}\right)$, we obtain the same bound on the expected weighted error. Indeed, we will only detect the heavy hitters earlier. Moreover, the expected number of queries to the oracle is at most
\[
\sum_{j=1}^m \frac{\gamma B(\log n)^3}{j}=O\left( B(\log n)^3 \log m\right).
\]
\end{remark}
\end{proof}

\section{Omitted Proofs of Section \ref{sec:worst_case}}\label{sec:tail_estimator_algorithm}

In this section, we discuss a version of our algorithm using a worst case estimate of the tail of the distribution, generalizing the value $O(AT/B)$ designed for Zipfian distributions. The algorithm \texttt{Basic-Tail-Sketch} is essentially the classic AMS sketch \cite{alon1996space} with $c = O(1)$ counters for the elements whose hash value is $1$. It is easy to see that the final algorithm, Algorithm \ref{alg:tail_estimator} uses $O(B)$ words of space.

\begin{algorithm}[H]
\caption{\label{alg:tail_estimator}Estimating the tail of the frequency vector $f$}
\begin{algorithmic}[1]
\State \textbf{Input:} Stream of updates to an $n$ dimensional vector $f$, space budget $O(B)$
\Procedure{Tail-Estimator}{}
\State Initialize $B$ independent copies of $\texttt{Basic-Tail-Sketch}$
\State Update each copy of $\texttt{Basic-Tail-Sketch}$ with updates from the stream
\For{$1 \le i \le B$}
\State $V_i \gets$ value outputted by $i$th copy of $\texttt{Basic-Tail-Sketch}$ after stream ends
\EndFor
\State \textbf{Return} $V \gets$ the $B/3$-th largest value among $\{V_i\}_{i= 1}^B$
\EndProcedure
\end{algorithmic}
\end{algorithm}

\begin{algorithm}[H]
\caption{\label{alg:basic_tail_sketch} Auxilliary algorithm for Algorithm \ref{alg:tail_estimator}}
\begin{algorithmic}[1]
\State \textbf{Input:} Stream of updates to an $n$ dimensional vector $f$
\Procedure{Basic-Tail-Sketch}{}
\State $T \gets \Theta(\log \log n)$
\State $B' \gets \Theta(B/T)$
\State $h: [n] \rightarrow [B']$ (4-wise independent hash function)
\State $c \gets 32$
\For{$1 \le j \le c$}
\State $s_j: [n] \rightarrow \{\pm 1\}$ (4-wise independent hash function)
\EndFor
\State Keep track of the sum $\frac{1}{c}\sum_{j=1}^{c}\left(\sum_{i:h(i)=1}f_{i}s_{j}(i)\right)^{2}$
\EndProcedure
\end{algorithmic}
\end{algorithm}





We now show that $V$, the output of Algorithm \ref{alg:tail_estimator}, satisfies $V \approx \| f_{\overline{\Theta(B')}}\|_2^{2}/B'$, which is of the same order as the threshold value used in line $7$ of Algorithm \ref{alg:main_learned}, generalizing the Zipfian case.

\begin{proof}[Proof of Lemma \ref{lem:tail_estimator}]
We analyze one copy of the sketch $V_{1}$, starting with the upper
bound.

Let $a$ be the number of elements $i\in[B'/10]$ such that $h(i)=1$.
Because $\E[a]=1/10$, by Markov's inequality, we have $a\le9/10$
with probability at least $8/9$. Next, let $W_{j}=\sum_{i>B'/10:h(i)=1}f_{i}s_{j}(i)$.
We have

\[
\E[W_{j}^{2}]=\sum_{i\ge B'/10}f_{i}^{2}\cdot[h(i)=1]=\left\Vert f_{\overline{B'/10}}\right\Vert _{2}^{2}/B'
\]

By Markov's inequality, we have $\frac{1}{4}\sum_{j}W_{j}^{2}\le9\left\Vert f_{\overline{B'/10}}\right\Vert _{2}^{2}/B'$
with probability $8/9$. By the union bound, $V_{1}^{2}\le9\left\Vert f_{\overline{B'/10}}\right\Vert _{2}^{2}/B$
with probability at least $7/9$.

Next, we show the lower bound. Let $X_{1}=\sum_{i:h(i)=1}\min\left(f_{i}^{2},f_{3B'}^{2}\right)$
and $Y_{1}=\sum_{i:h(i)=1}f_{i}^{2}$. Observe that $X_{1}\le Y_{1}$.
We have

\begin{align*}
\E_{h}\left[X_{1}\right] & =\left\Vert f_{\overline{3B'}}\right\Vert _{2}^{2}/B'+3f_{3B'}^{2}\\
Var\left(X_{1}\right) & =\frac{B'-1}{B'^{2}}\left(\sum_{i>3B'}f_{i}^{4}+3Bf_{3B'}^{4}\right)
\end{align*}

By Chebyshev's inequality,

\begin{align*}
\Pr\left[X_{1}\le\left\Vert f_{\overline{3B'}}\right\Vert _{2}^{2}/(3B')\right] & \le\frac{\frac{B'-1}{B'^{2}}\left(\sum_{i>3B'}f_{i}^{4}+3B'f_{3B'}^{4}\right)}{\left(2\left\Vert f_{\overline{3B'}}\right\Vert _{2}^{2}/(3B')+3f_{3B'}^{2}\right)^{2}}\\
 & \le\frac{\left(B'-1\right)\left(\sum_{i>3B}f_{i}^{4}+3B'f_{3B'}^{4}\right)}{4\left\Vert f_{\overline{3B'}}\right\Vert _{2}^{4}/9+9B'^{2}f_{3B'}^{2}+4B'\left\Vert f_{\overline{3B'}}\right\Vert _{2}^{2}f_{3B'}^{2}}\\
 & \le\frac{1}{3}.
\end{align*}

Thus, with probability at least $2/3$, we have $Y_{1}\ge\left\Vert f_{\overline{3B'}}\right\Vert _{2}^{2}/(3B')$.
Next we can bound $V_{1}$ in terms of $Y_{1}$ using the standard
analysis of the AMS sketch. Let $Z_{j}=\sum_{i:h(i)=1}f_{i}s_{j}(i)$.

\begin{align*}
\E_{s}\left[Z_{j}^{2}|h\right] & =Y_{1}\\
\E_{s}\left[Z_{j}^{4}|h\right] & =\sum_{i:h(i)=1}f_{i}^{4}+6\sum_{i<j:h(i)=h(j)=1}f_{i}^{2}f_{j}^{2}=3Y_{1}^{2}-2\sum_{i:h(i)=1}f_{i}^{4}.
\end{align*}

By the Chebyshev's inequality, $\Pr\left[V_{1}\le Y_{1}/2\right]\le\frac{2Y_{1}^{2}/c}{Y_{1}^{2}/4}\le\frac{8}{c}\le\frac{1}{4}$.
By the union bound, we have $V_{1}\ge\left\Vert f_{\overline{3B}}\right\Vert _{2}^{2}/(6B')$
with probability at least $5/12$.

The lemma follows from applying the Chernoff bound to the independent
copies $V_{1},\ldots,V_{B}$.
\end{proof}
Given the estimator $V$, we can output $0$ for elements whose squared
estimated frequency is below $V$.

\begin{lemma}
Let $E$ be the event that $V$ is accurate, which holds with probability $1-\exp\left(\Omega\left(B\right)\right)$. 
If $f_i^2 \le \left\Vert f_{\overline{3B'}}\right\Vert _{2}^{2}/(12B')$ then with
probability $1-\exp\left(\Omega\left(B\right)\right)-1/polylog(n)$, the algorithm
outputs $0$.
If $f_i^2 \ge \left\Vert f_{\overline{3B'}}\right\Vert _{2}^{2}/(12B')$ then 
with constant probability, 
\[  \left\Vert f_{i}^{2}-\hat{f}_{i}^{2} \right\Vert \le O\left(\left\Vert f_{\overline{\Omega(B')}}\right\Vert _{2}^{2}/B'\right)
\]
\end{lemma}

\begin{proof}
Observe that the error in the comparison between the threshold $V$ and $\tilde{f}_i$ is bounded by $V$ plus the estimation
error of $\tilde{f}_{i}$. By the standard analysis of the CountSketch, 
with probability $1-\exp\left(\Omega\left(T\right)\right)$,
\[
\left|f_{i}^{2}-\tilde{f}_{i}^{2}\right|\le O\left(\left\Vert f_{\overline{\Omega(B')}}\right\Vert _{2}^{2}/B'\right)
\]

Thus, if $f_i^2 \le \left\Vert f_{\overline{3B'}}\right\Vert _{2}^{2}/(12B')$ then with
probability $1-\exp\left(\Omega\left(B\right)\right)-1/polylog(n)$, we have $V\ge \tilde{f}_i$ and the algorithm outputs $0$.

On the other hand, consider $f_i^2 \ge \left\Vert f_{\overline{3B'}}\right\Vert _{2}^{2}/(12B')$. First, consider the case when the algorithm outputs $0$. Except for a failure probability of $\exp\left(\Omega\left(B\right)\right)+1/polylog(n)$, it must be the case that $f_i^2 = O\left(\left\Vert f_{\overline{3B'}}\right\Vert _{2}^{2}/(12B')\right)$ so we have $|f_i^2-0| =O\left(\left\Vert f_{\overline{3B'}}\right\Vert _{2}^{2}/(12B')\right)$. Next, consider the case when the algorithm outputs the answer from $S_T$. The correctness guarantee of this case follows from the standard analysis of CountSketch, which guarantees that for a single row of CountSketch with $B$ columns, with constant probability,$\left|f_{i}^{2}-\tilde{f}_{i}^{2}\right|\le O\left(\left\Vert f_{\overline{\Omega(B)}}\right\Vert _{2}^{2}/B\right)$.
\end{proof}

\begin{proof}[Proof of Lemma \ref{lem:worst_case}]
Note that we are assuming Lemma \ref{lem:tail_estimator} is satisfied, which happens with probability $1-1/\text{poly}(n)$. For elements with true frequencies less than  $O(\|f_{\overline{B'}}\|_2/\sqrt{B'})$ for $B' = O(B/\log \log n)$, we either we either use the last CS table in Algorithm \ref{alg:main_unlearned} or we set the estimate to be $0$. In either case, the inequality holds as $O(\|f_{\overline{B'}}\|_2/\sqrt{B'})$ is the expected error of a standard $1 \times B'$ CS table.

For elements with frequency larger than $O(\|f_{\overline{B'}}\|_2/\sqrt{B'})$, we ideally want to use the last CS table in Algorithm \ref{alg:main_unlearned}. In such a case, we easily satisfy the desired inequality since we are using a CS table with even more columns. But there is a small probability we output $0$. We can easily handle this as follows. Let $f_i = \ell \|f_{\overline{B'}}\|_2/\sqrt{B'}$ be the frequency of element $i$ for $\ell \ge C$ for a large constant $C$. Any fixed CS table with $B'$ columns gives us expected error $\|f_{\overline{B'}}\|_2/\sqrt{B'}$, so the probability that it estimates the frequency of $f_i$ to be smaller than $\|f_{\overline{B'}}\|_2/\sqrt{B'}$ is at most $1/\Omega(\ell)$ by a straightforward application of Markov's inequality. Since we take the median across $\Theta(\log \log n)$ different CS tables in Algorithm \ref{alg:main_unlearned}, a standard Chernoff bound implies that the probability the median estimate is smaller than $O(|f_{\overline{B'}}\|_2/\sqrt{B'})$ is at most $(1/\ell)^{\Omega(\log \log n)}$. In particular, the expected error of our estimate is at most $\ll \left(\ell \|f_{\overline{B'}}\|_2/\sqrt{B'}\right) \cdot 1/\ell = O(\|f_{\overline{B'}}\|_2/\sqrt{B'})$, which can be upper bounded by the expected error of CS table with $cB/\log\log n$ columns for a sufficiently small $c$, completing the proof.
\end{proof}

\section{Concentration bounds}\label{concentration}
In this appendix we collect some concentration inequalities for reference in the main body of the paper. The inequality we will use the most is Bennett's inequality. However, we remark that for our applications, several other variance based concentration result would suffice, e.g., Bernstein's inequality. 
 \begin{theorem}[Bennett's inequality~\cite{bennett1962}]\label{thm:Bennett}
Let $X_1,\dots,X_n$ be independent, mean zero random variables. Let $S=\sum_{i=1}^n X_i$, and $\sigma^2,M>0$ be such that $\Var[S]\leq \sigma^2$ and $|X_i|\leq M$ for all $i\in [n]$. For any $t\geq 0$,
\begin{align*}
\Pr[S\geq t]\leq \exp \left(-\frac{\sigma^2}{M^2} h \left( \frac{tM}{\sigma^2} \right)\right),
\end{align*}
where $h:\R_{\geq 0} \to \R_{\geq 0 }$ is defined by $h(x)=(x+1) \log (x+1)-x$. The same tail bound holds on the probability $\Pr[S\leq -t]$.
\end{theorem}
\begin{remark}\label{asymptotics}
    For $x\geq 0$, $\frac{1}{2}x \log (x+1) \leq h(x) \leq x \log (x+1)$. We will use these asymptotic bounds repeatedly in this paper. 
\end{remark}

A corollary of Bennett's inequality is the classic Chernoff bounds.

\begin{theorem}[Chernoff~\cite{Che52:chernoff}]\label{thm:Chernoff}
Let $X_1,\dots,X_n\in [0,1]$ be independent random variables and $S=\sum_{i=1}^n X_i$. Let $\mu=\E[S]$. Then
\begin{align*}
\Pr[S\geq (1+\delta)\mu ] \leq \exp(-\mu h(\delta)).
\end{align*}
\end{theorem}

Even weaker than Chernoff's inequality is Hoeffding's inequality.

\begin{theorem}[Hoeffding~\cite{hoeffding1963inequality}]\label{thm:Hoeffding}
Let $X_1,\dots,X_n \in [0,1]$ be independent random variables. Let $S=\sum_{i=1}^n X_i$. Then 
\begin{align*}
\Pr[S-\E[S]\geq t]\leq e^{-\frac{2t^2}{n}}.
\end{align*}
\end{theorem}

\newpage

\section{Additional Experiments}\label{appendix-experiment}

In this section, we display figures for synthetic Zipfian data and additional figures for the CAIDA and AOL datasets.

\begin{figure}[h]
    \centering
    \includegraphics[width=0.49\textwidth]{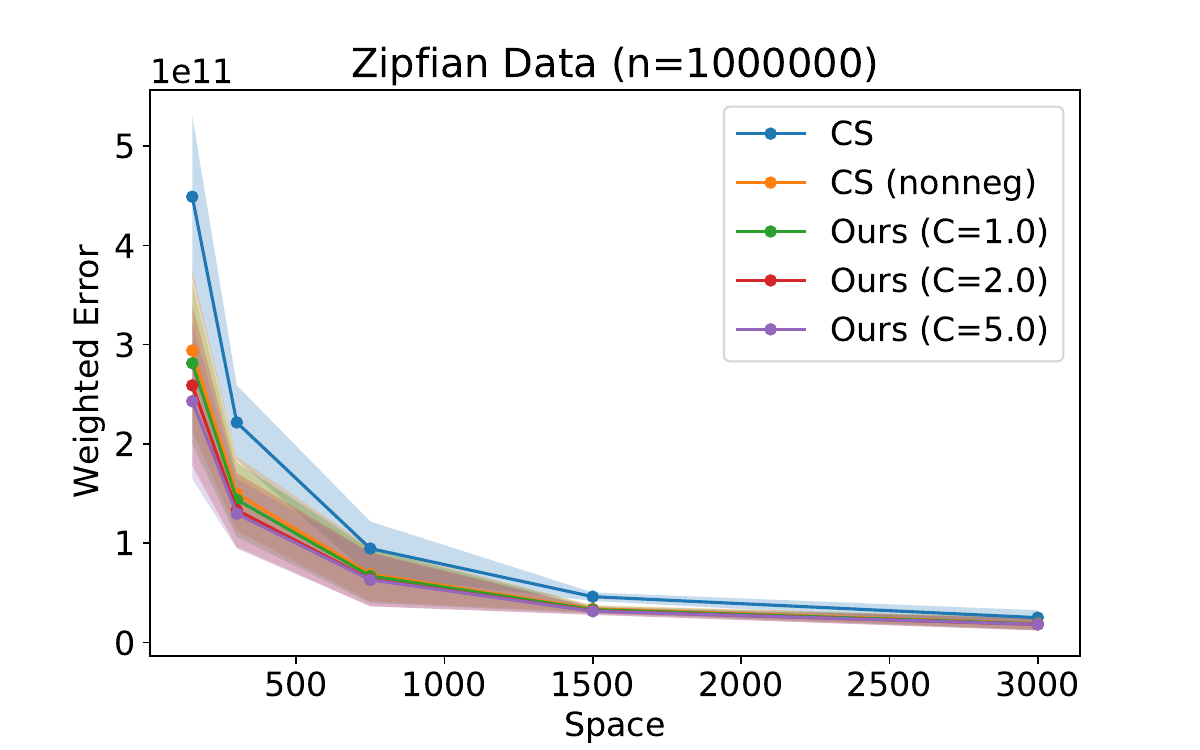}
    \includegraphics[width=0.49\textwidth]{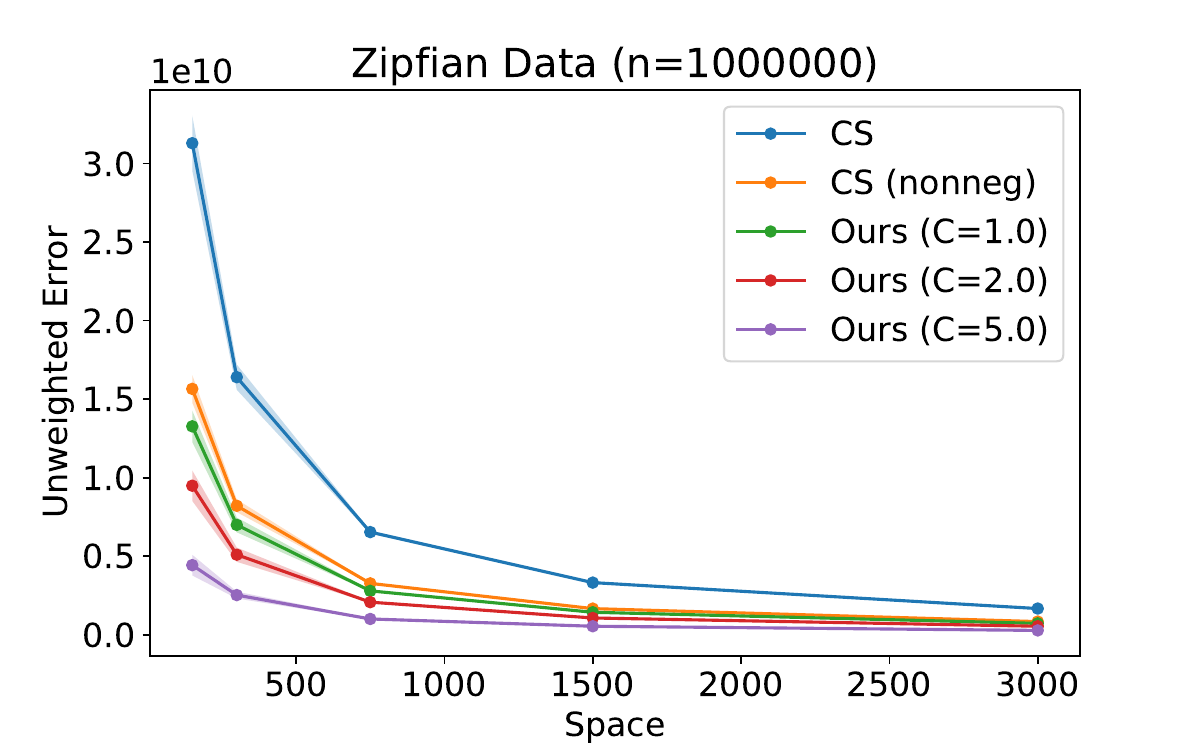}
    \caption{Comparison of weighted and unweighted error without predictions on Zipfian data.}
    \label{fig:zipf}
\end{figure}

\begin{figure}[h]
    \centering
    \includegraphics[width=0.49\textwidth]{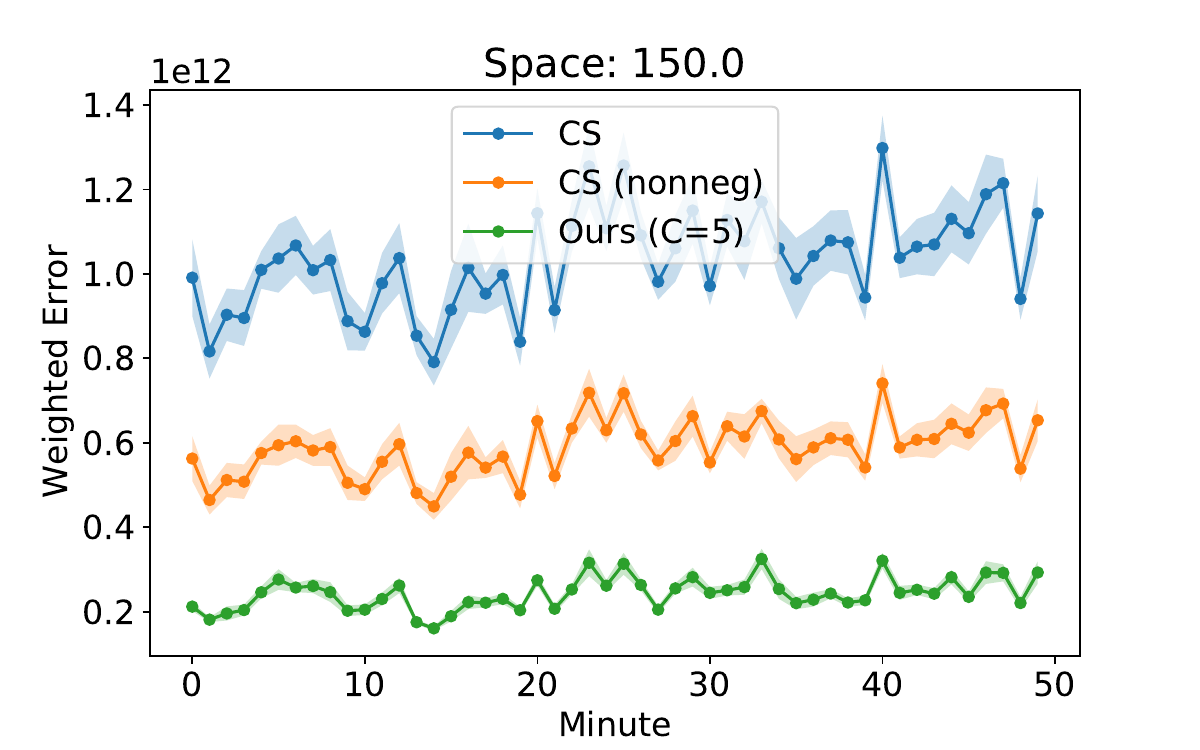}
    \includegraphics[width=0.49\textwidth]{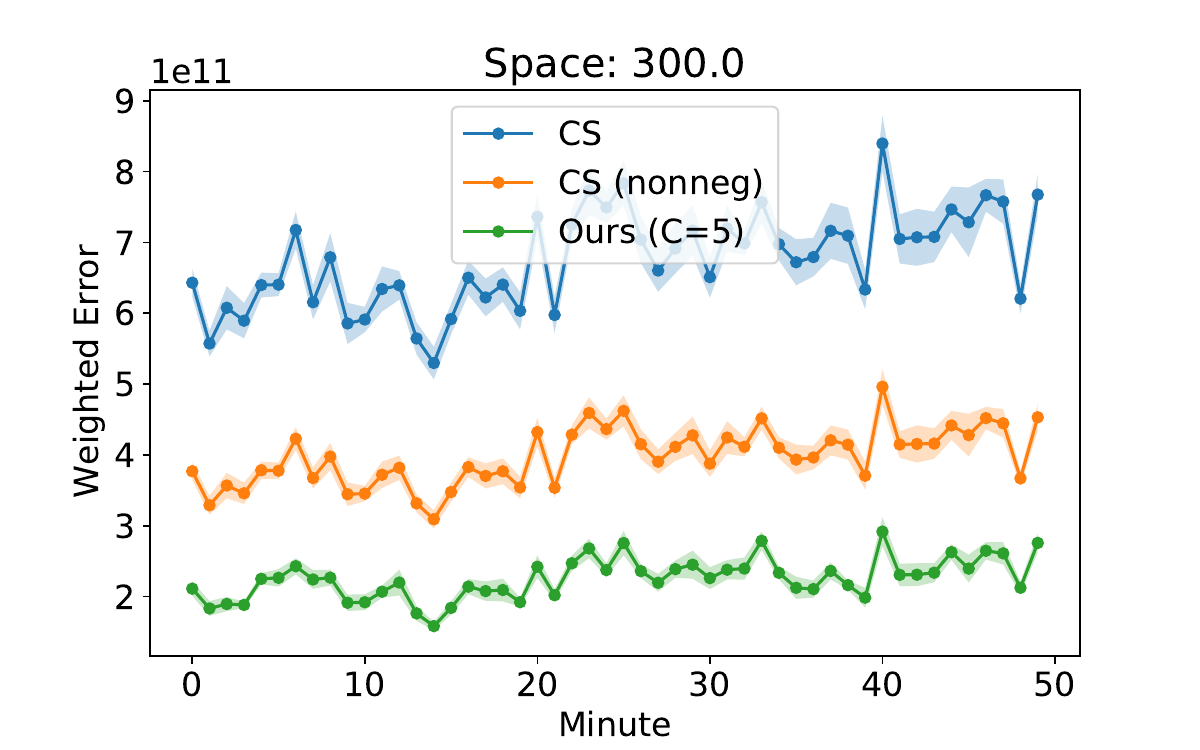}
    \includegraphics[width=0.49\textwidth]{plots/ip/werr-predFalse-space750.0.pdf}
    \includegraphics[width=0.49\textwidth]{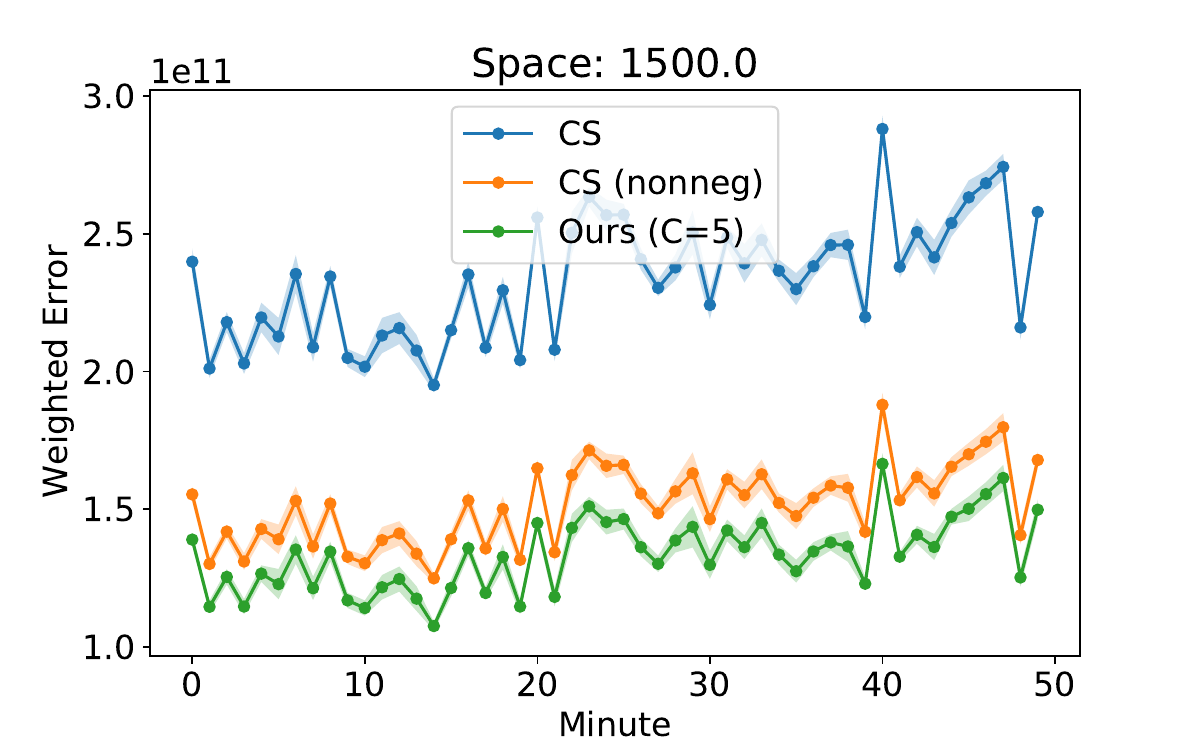}
    \includegraphics[width=0.49\textwidth]{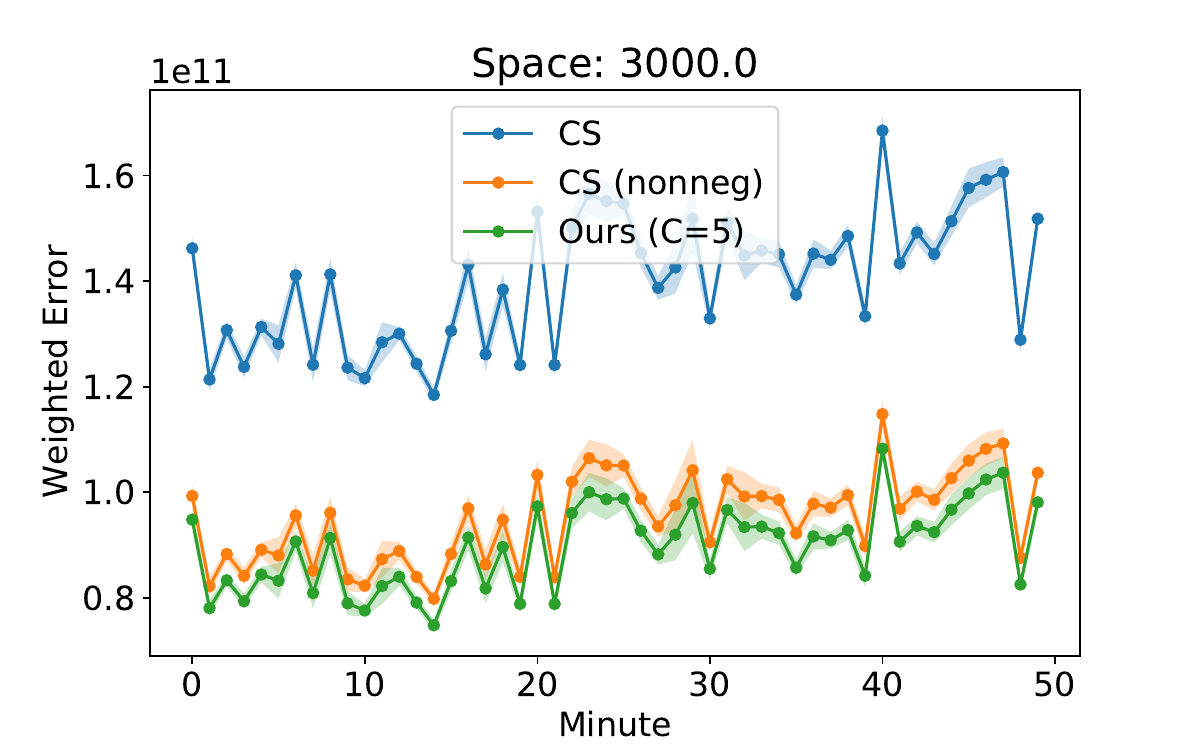}
    \caption{Comparison of weighted errors without predictions on the CAIDA dataset}
\end{figure}

\begin{figure}
    \centering
    \includegraphics[width=0.49\textwidth]{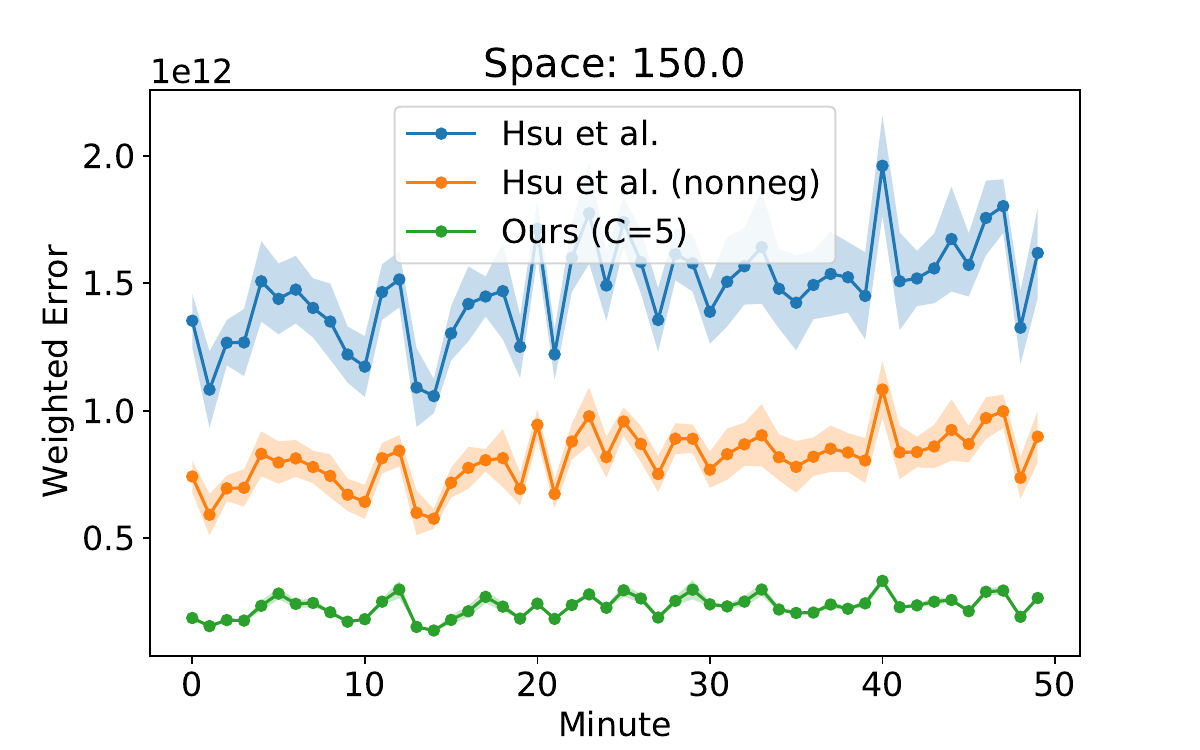}
    \includegraphics[width=0.49\textwidth]{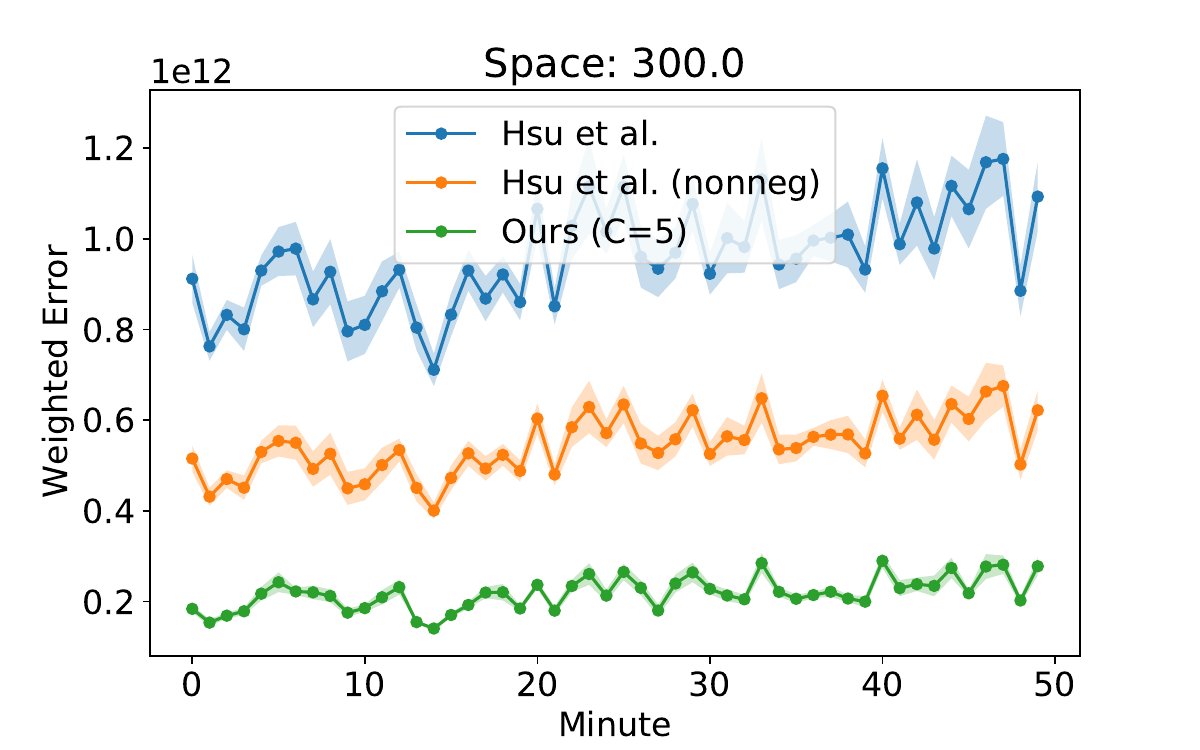}
    \includegraphics[width=0.49\textwidth]{plots/ip/werr-predTrue-space750.0.pdf}
    \includegraphics[width=0.49\textwidth]{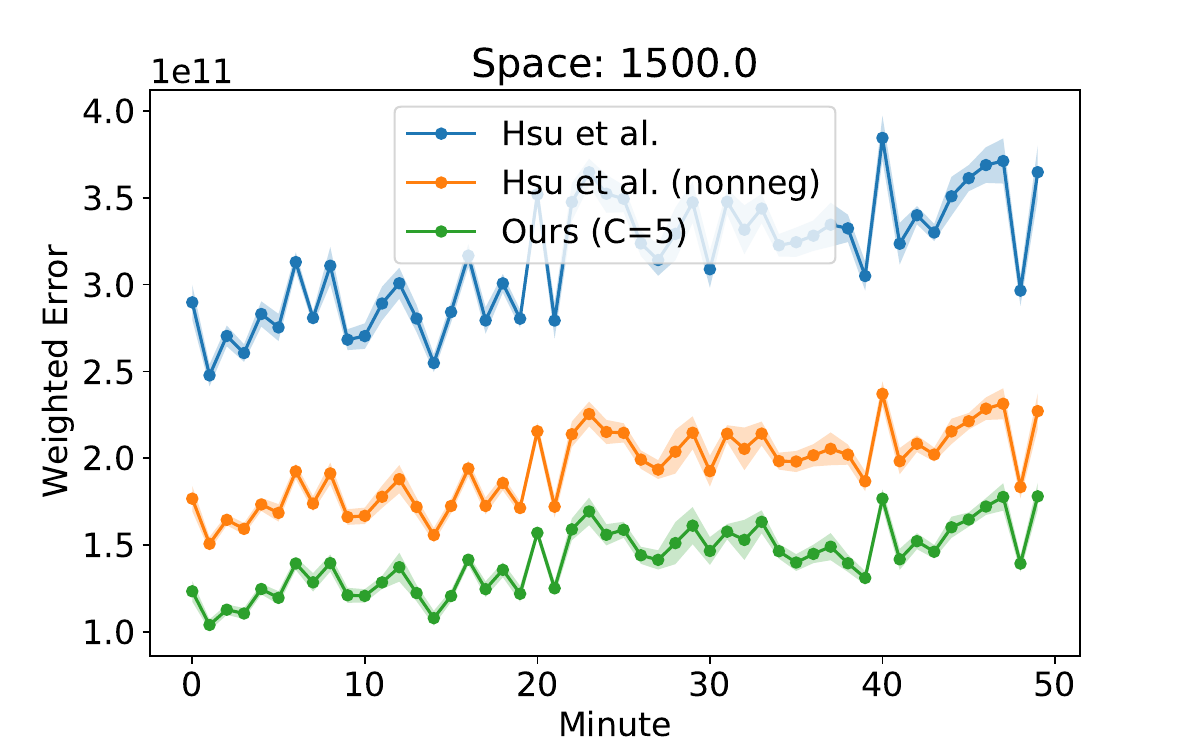}
    \includegraphics[width=0.49\textwidth]{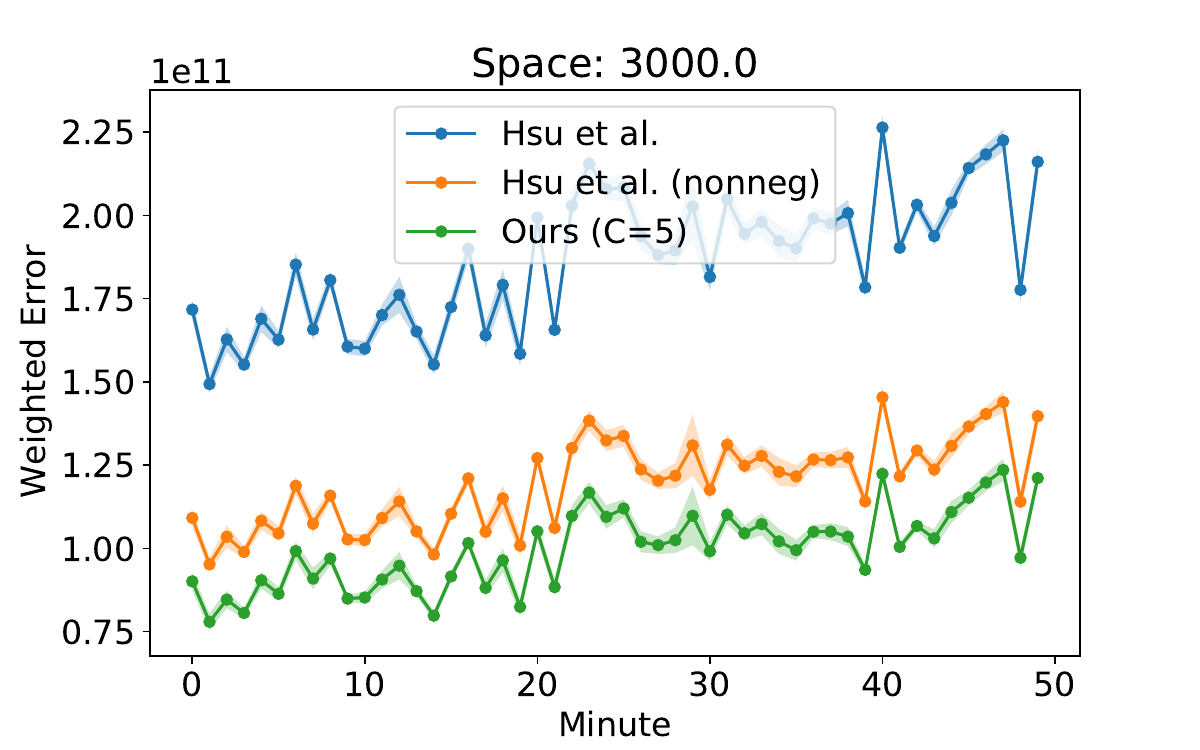}
    \caption{Comparison of weighted errors with predictions on the CAIDA dataset}
\end{figure}

\begin{figure}
    \centering
    \includegraphics[width=0.49\textwidth]{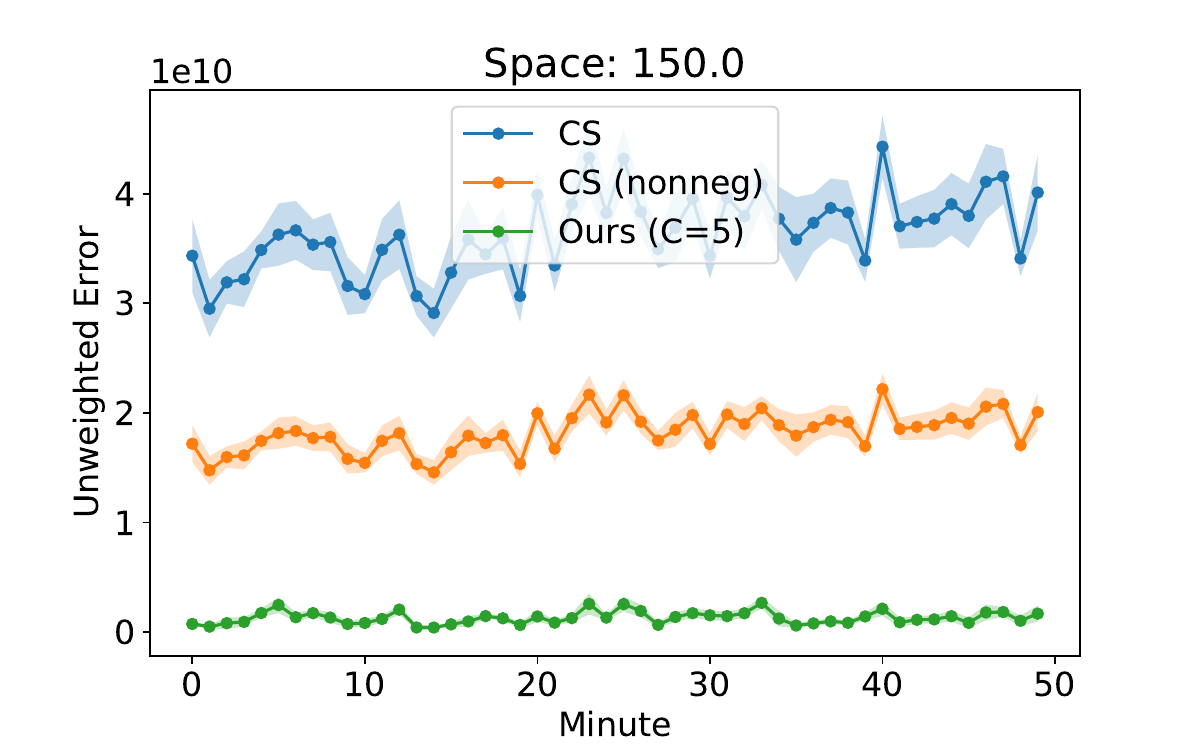}
    \includegraphics[width=0.49\textwidth]{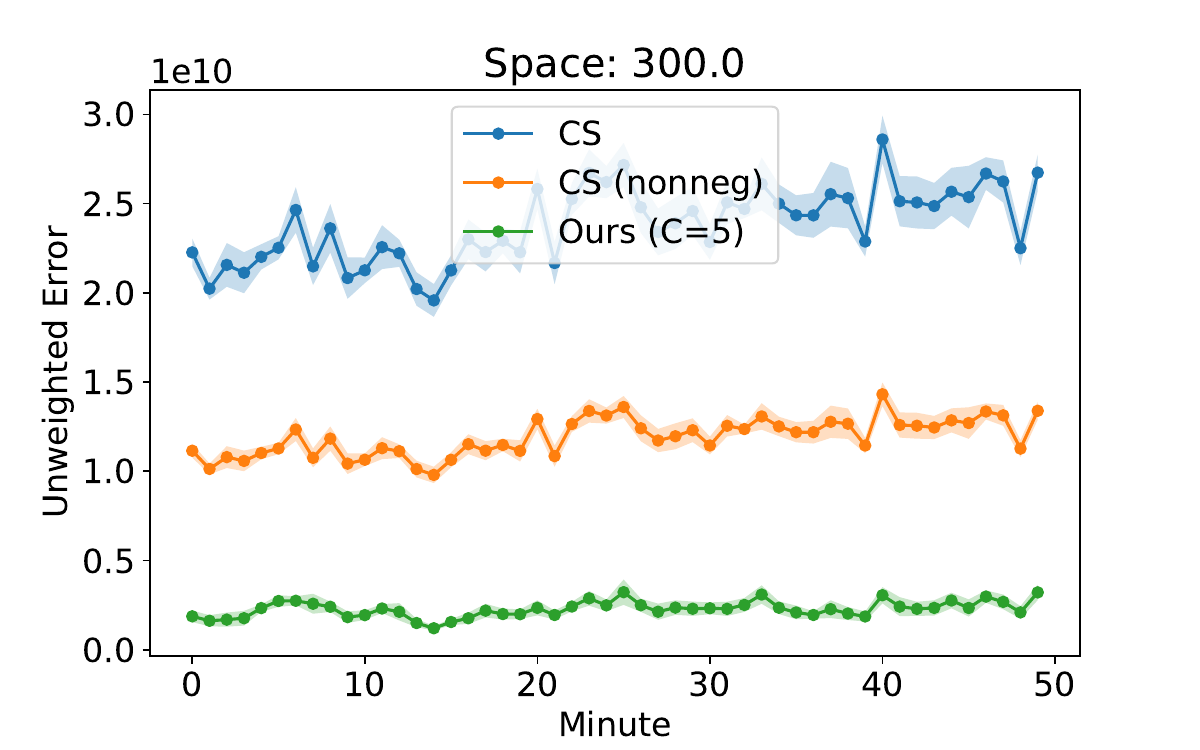}
    \includegraphics[width=0.49\textwidth]{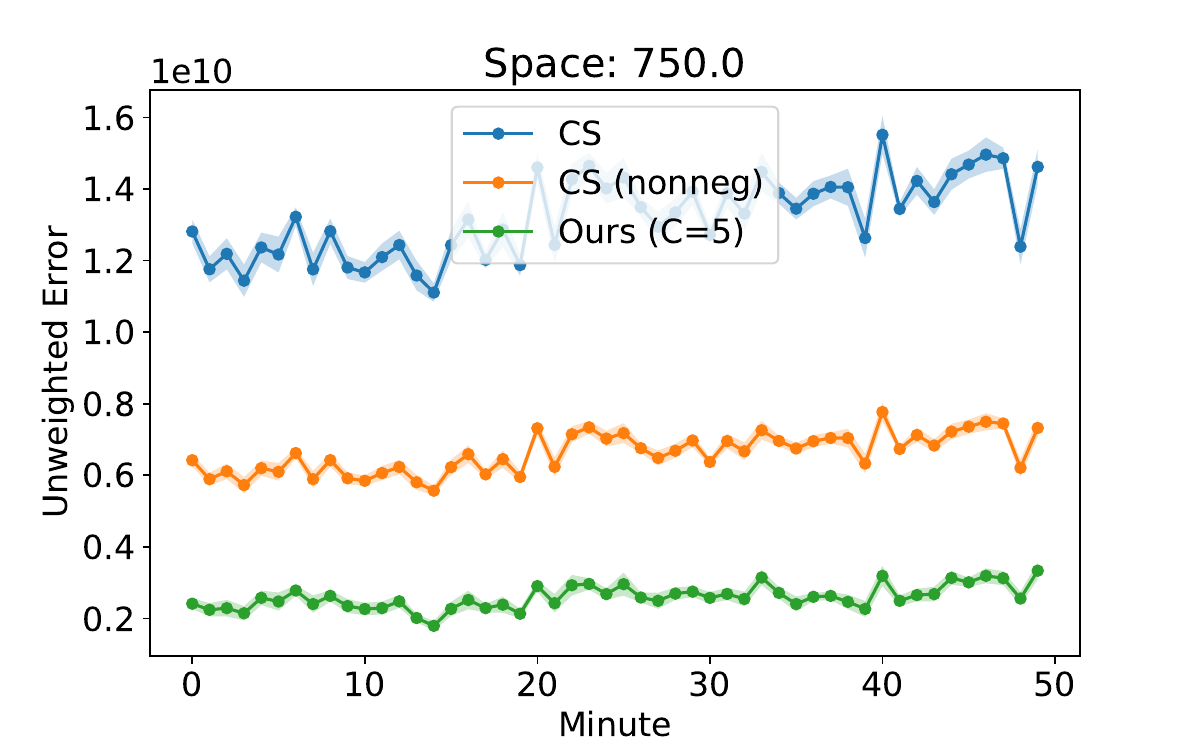}
    \includegraphics[width=0.49\textwidth]{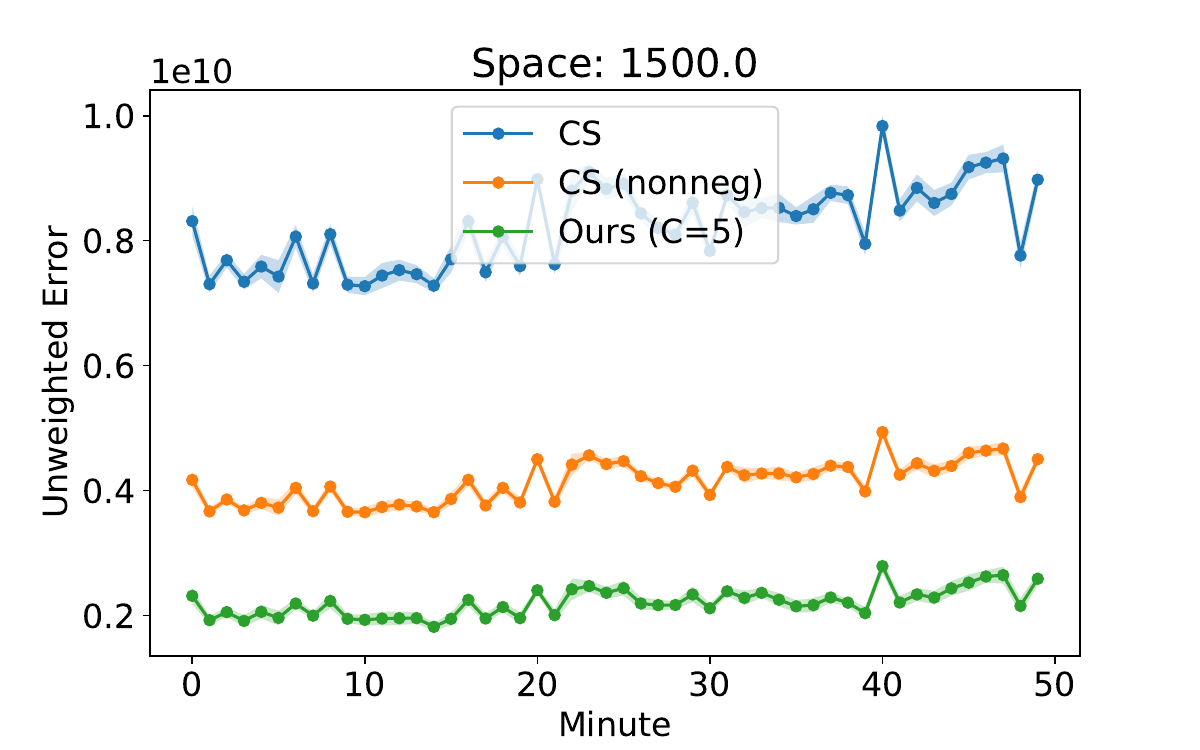}
    \includegraphics[width=0.49\textwidth]{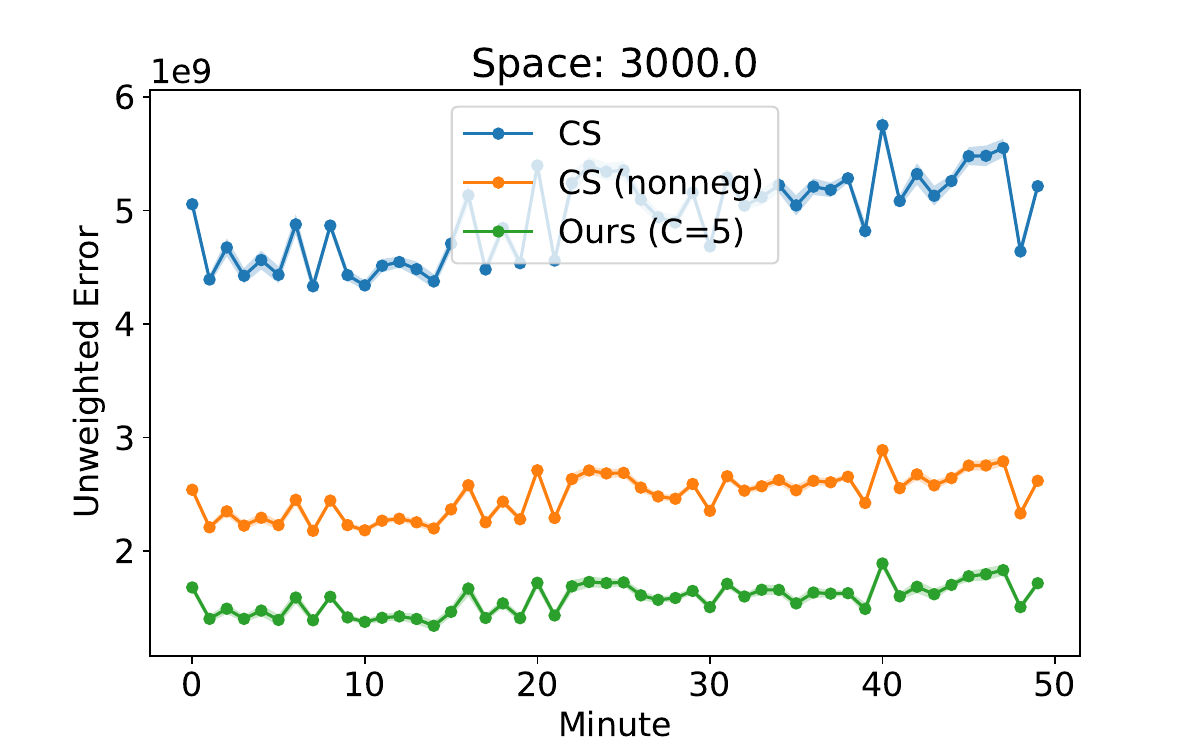}
    \caption{Comparison of unweighted errors without predictions on the CAIDA dataset}
\end{figure}

\begin{figure}
    \centering
    \includegraphics[width=0.49\textwidth]{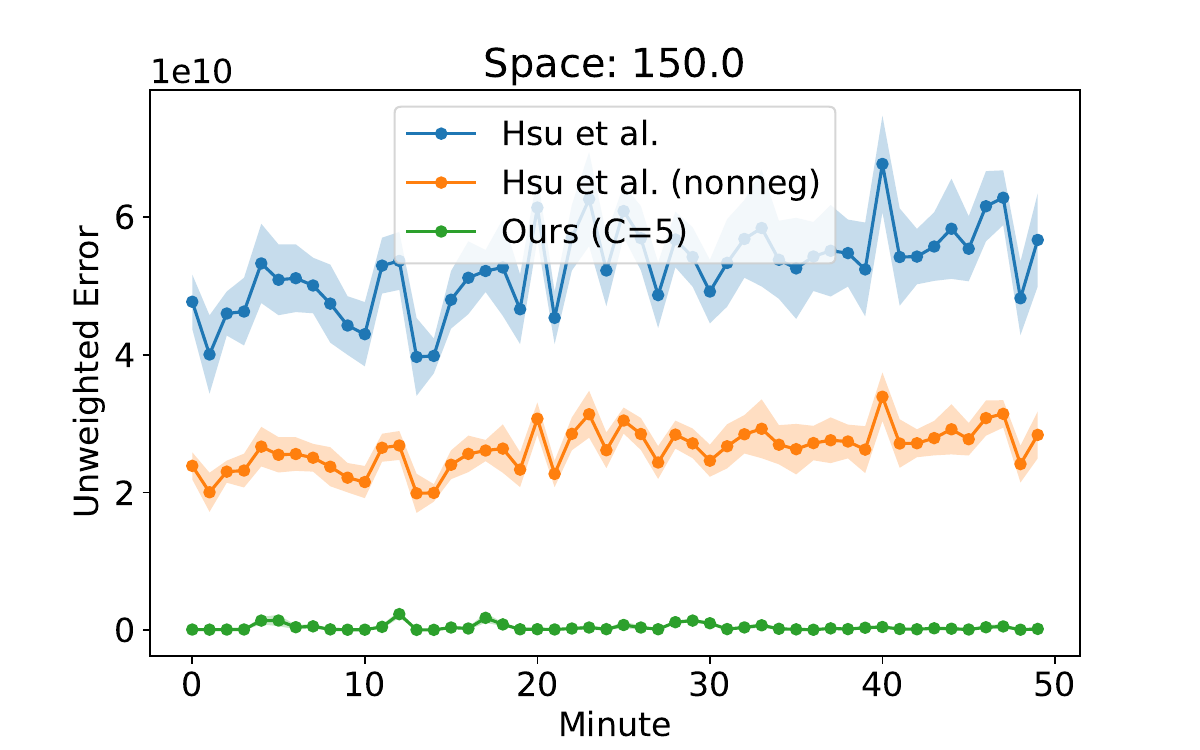}
    \includegraphics[width=0.49\textwidth]{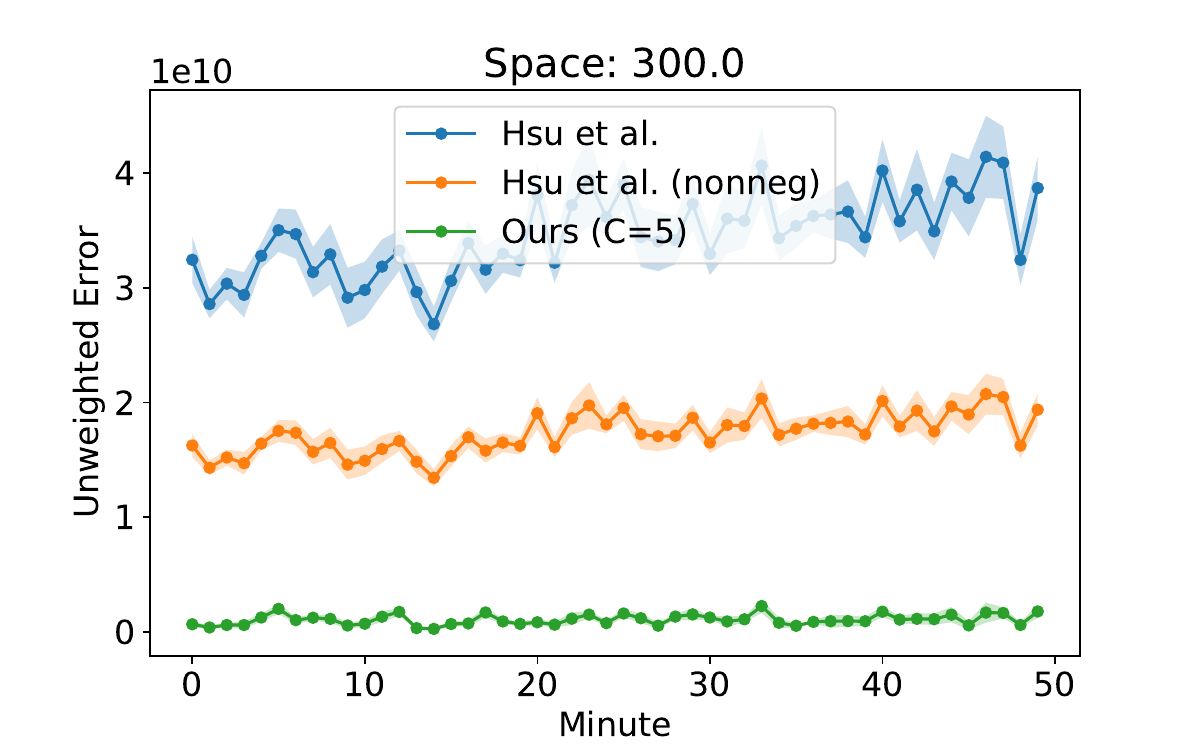}
    \includegraphics[width=0.49\textwidth]{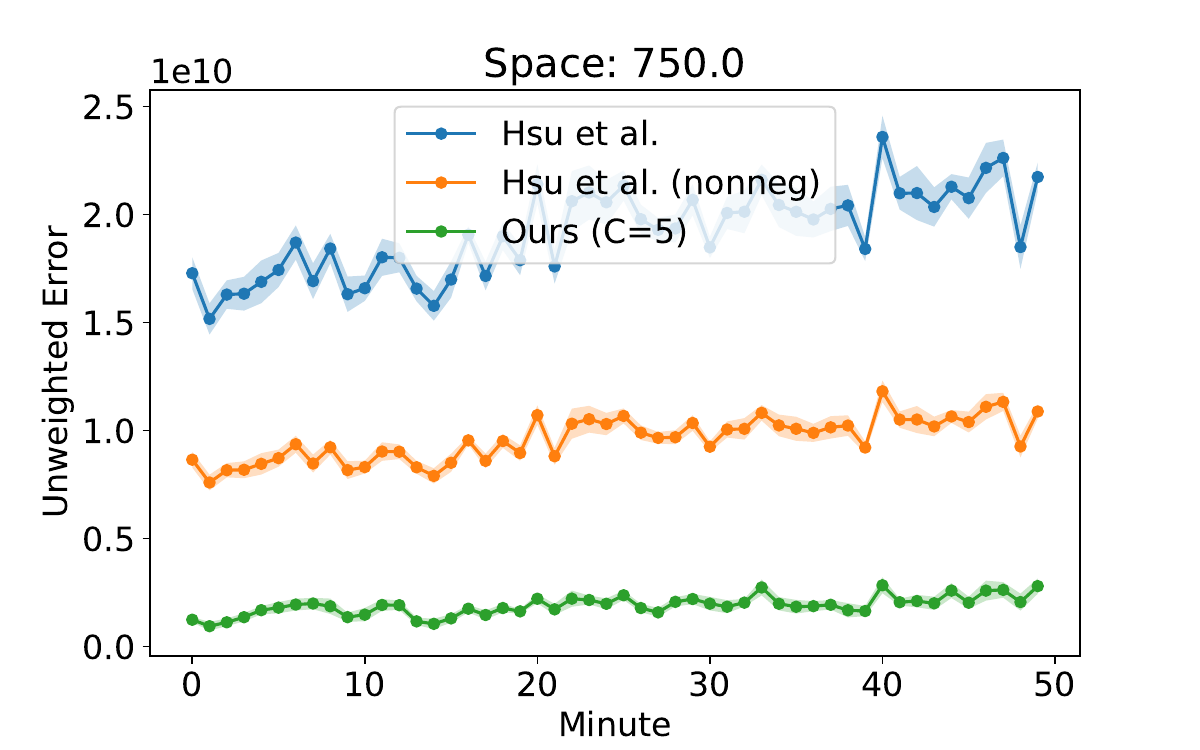}
    \includegraphics[width=0.49\textwidth]{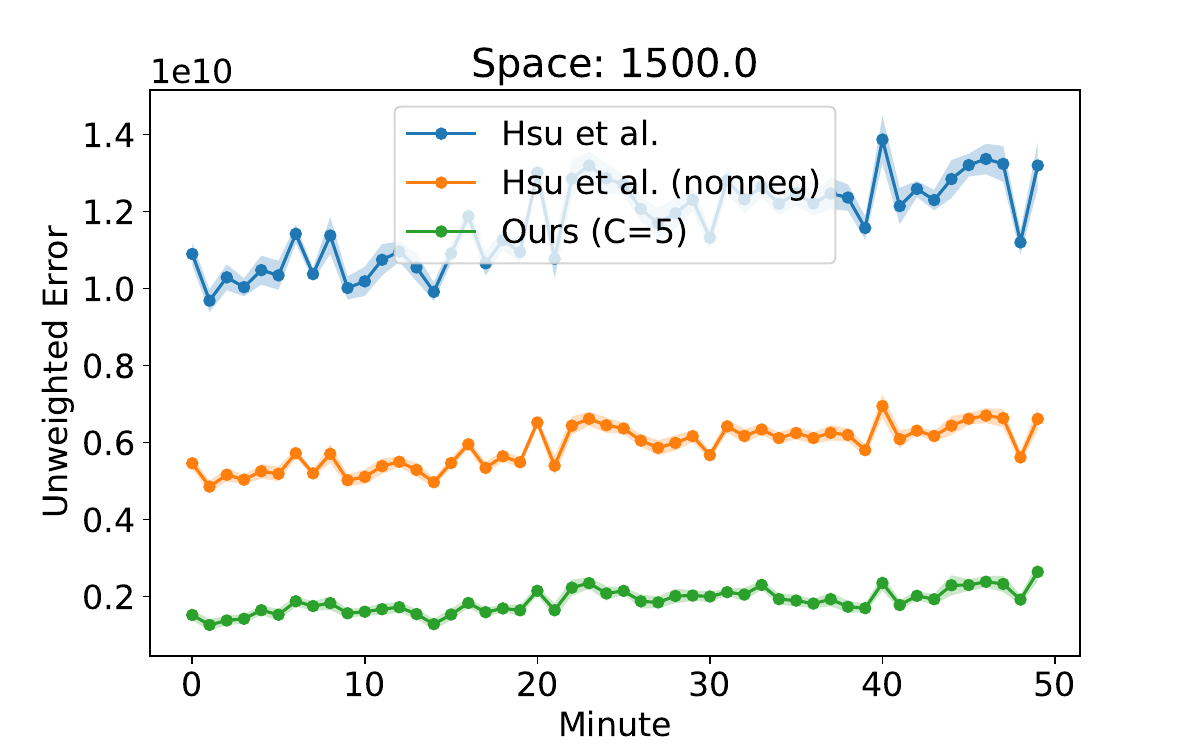}
    \includegraphics[width=0.49\textwidth]{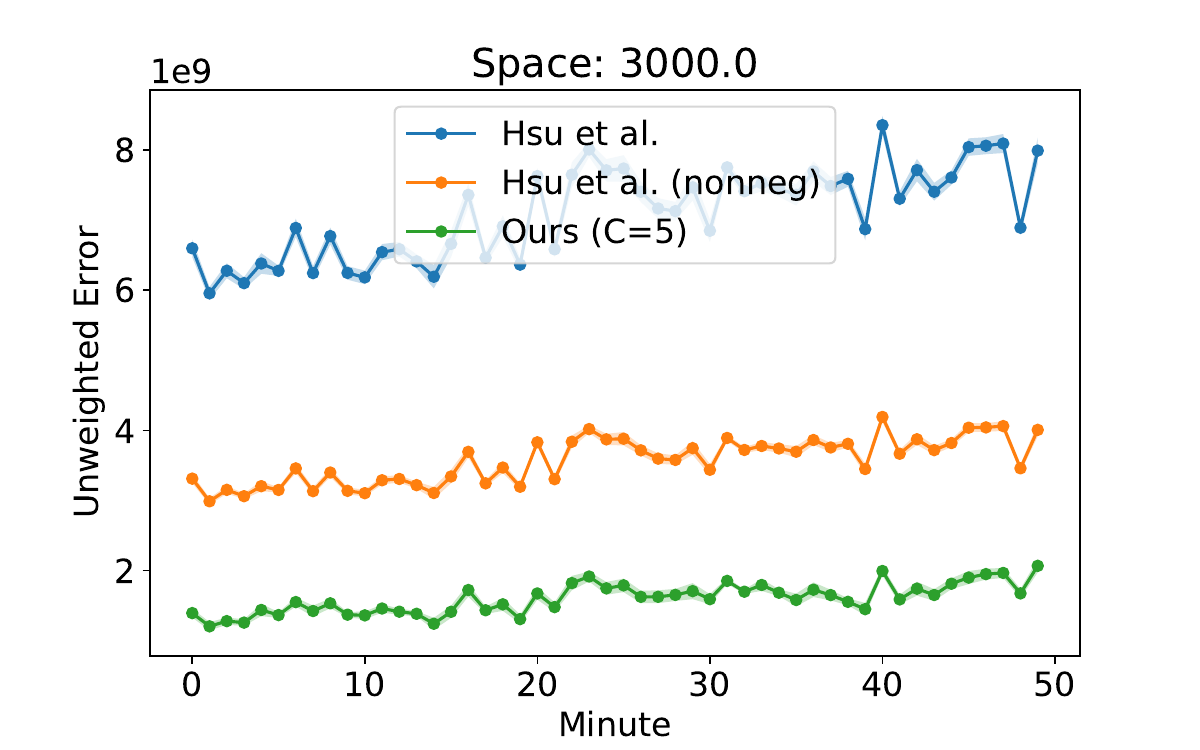}
    \caption{Comparison of unweighted errors with predictions on the CAIDA dataset}
\end{figure}

\begin{figure}
    \centering
    \includegraphics[width=0.49\textwidth]{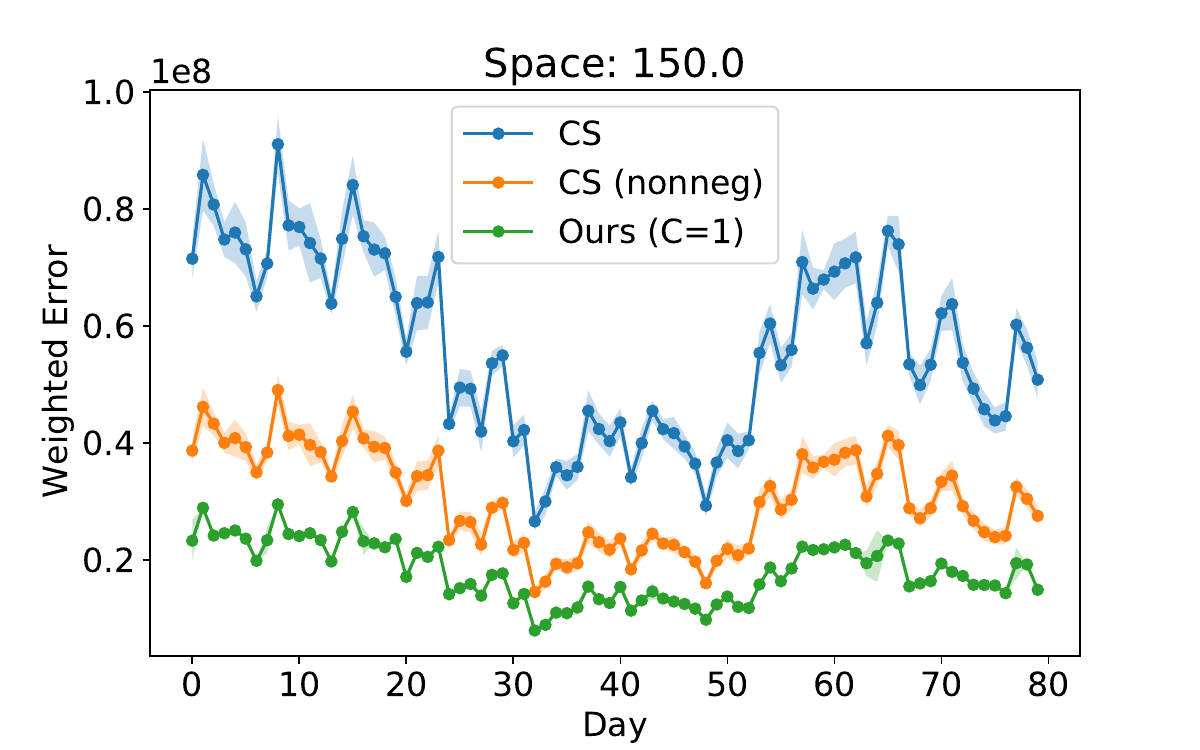}
    \includegraphics[width=0.49\textwidth]{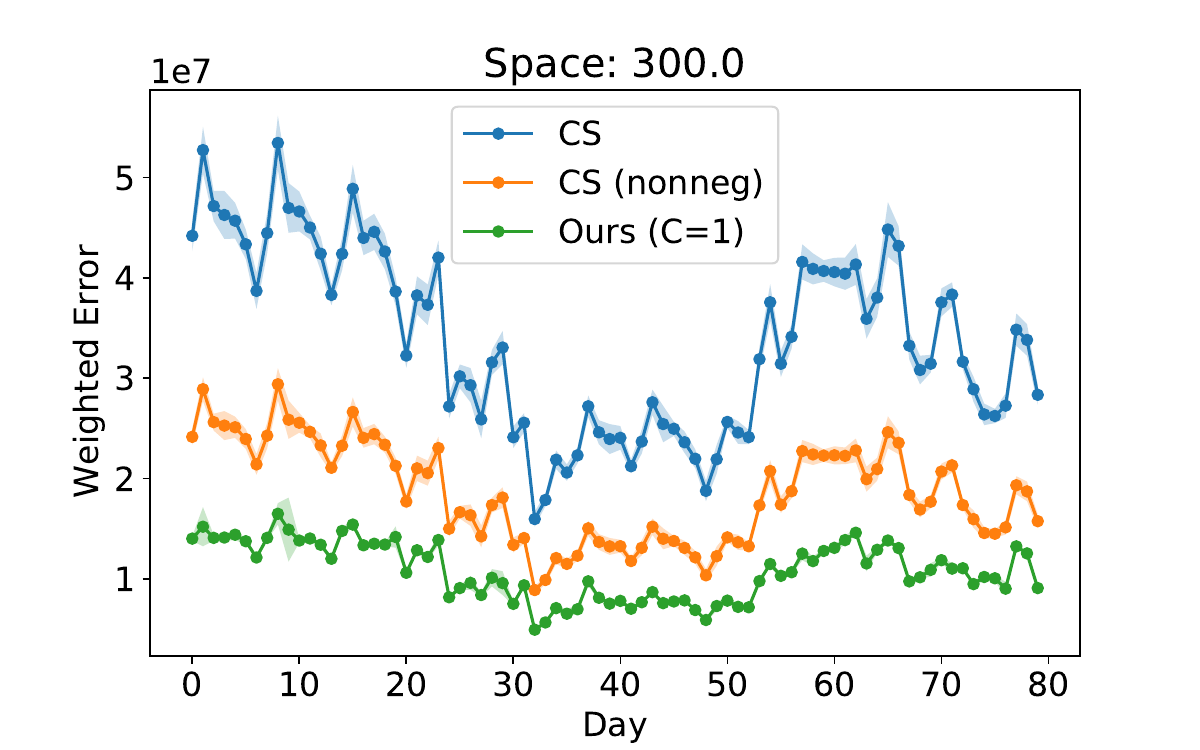}
    \includegraphics[width=0.49\textwidth]{plots/aol/werr-predFalse-space750.0.pdf}
    \includegraphics[width=0.49\textwidth]{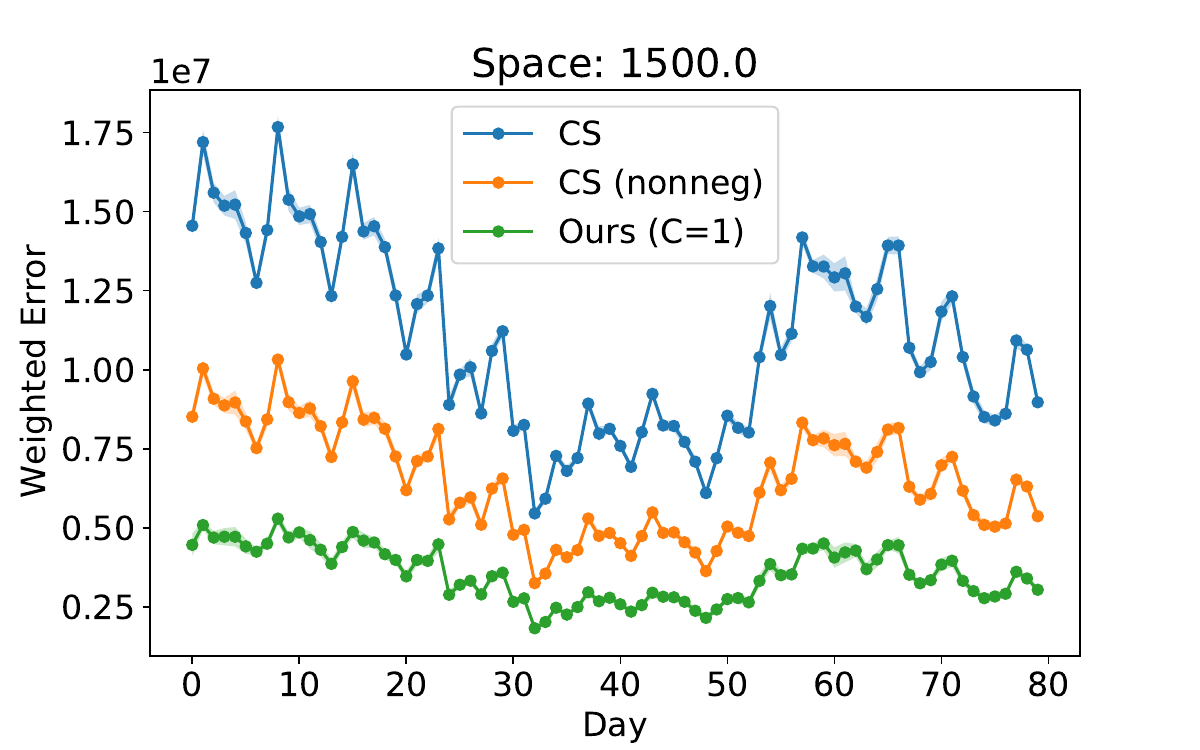}
    \includegraphics[width=0.49\textwidth]{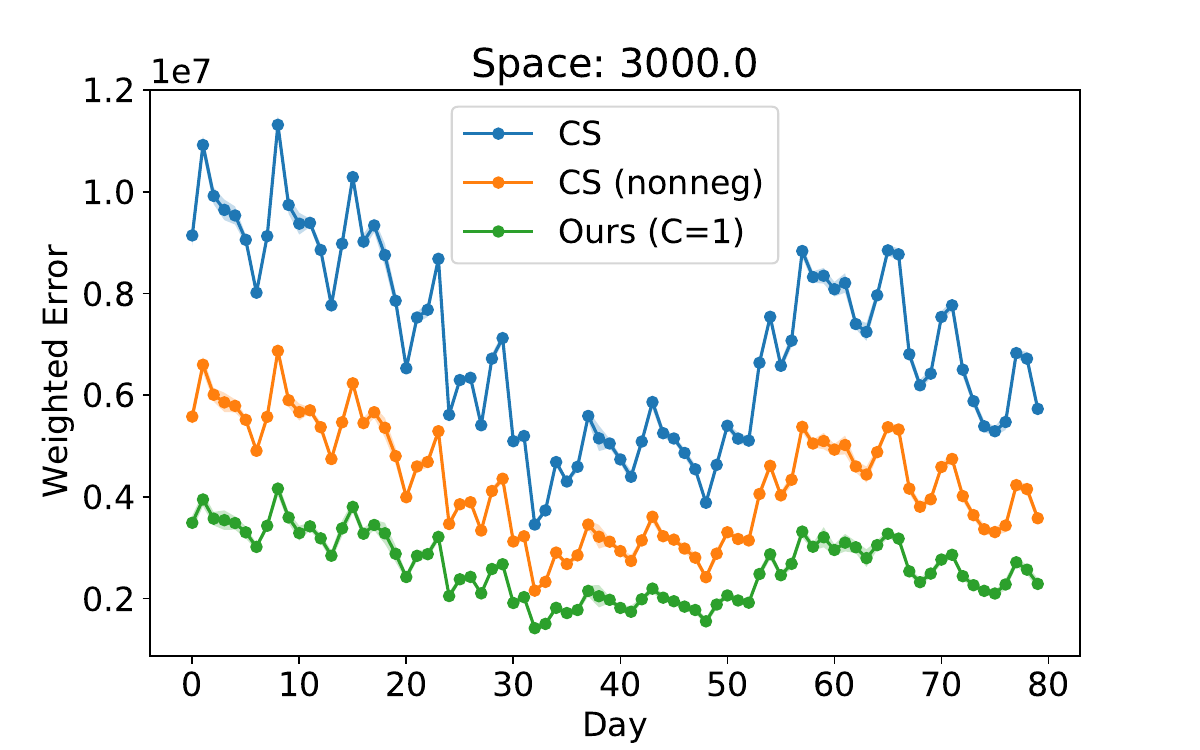}
    \caption{Comparison of weighted errors without predictions on the AOL dataset}
\end{figure}

\begin{figure}
    \centering
    \includegraphics[width=0.49\textwidth]{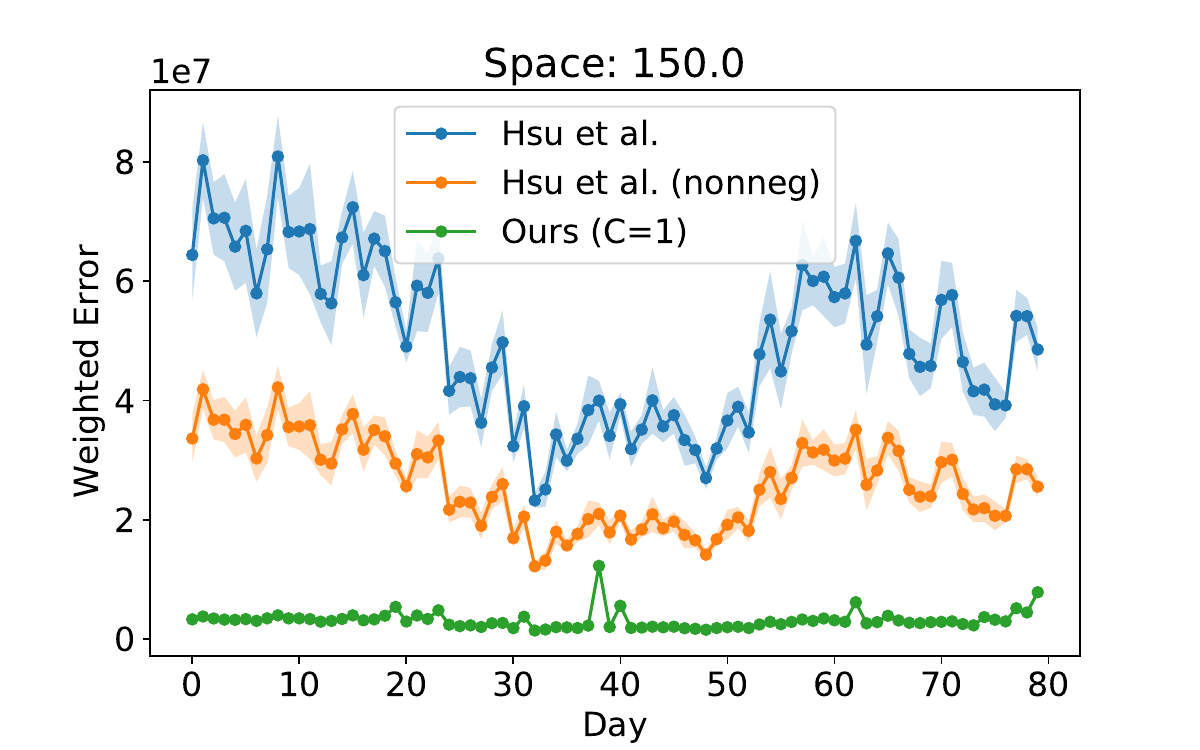}
    \includegraphics[width=0.49\textwidth]{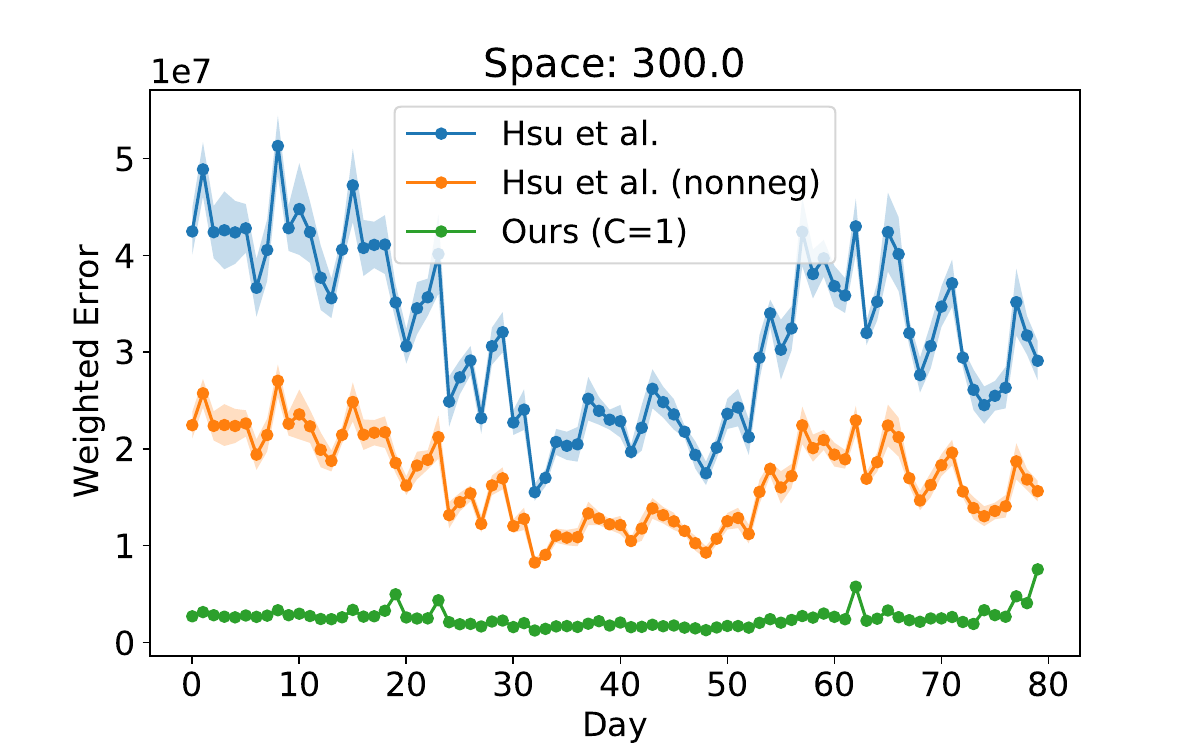}
    \includegraphics[width=0.49\textwidth]{plots/aol/werr-predTrue-space750.0.pdf}
    \includegraphics[width=0.49\textwidth]{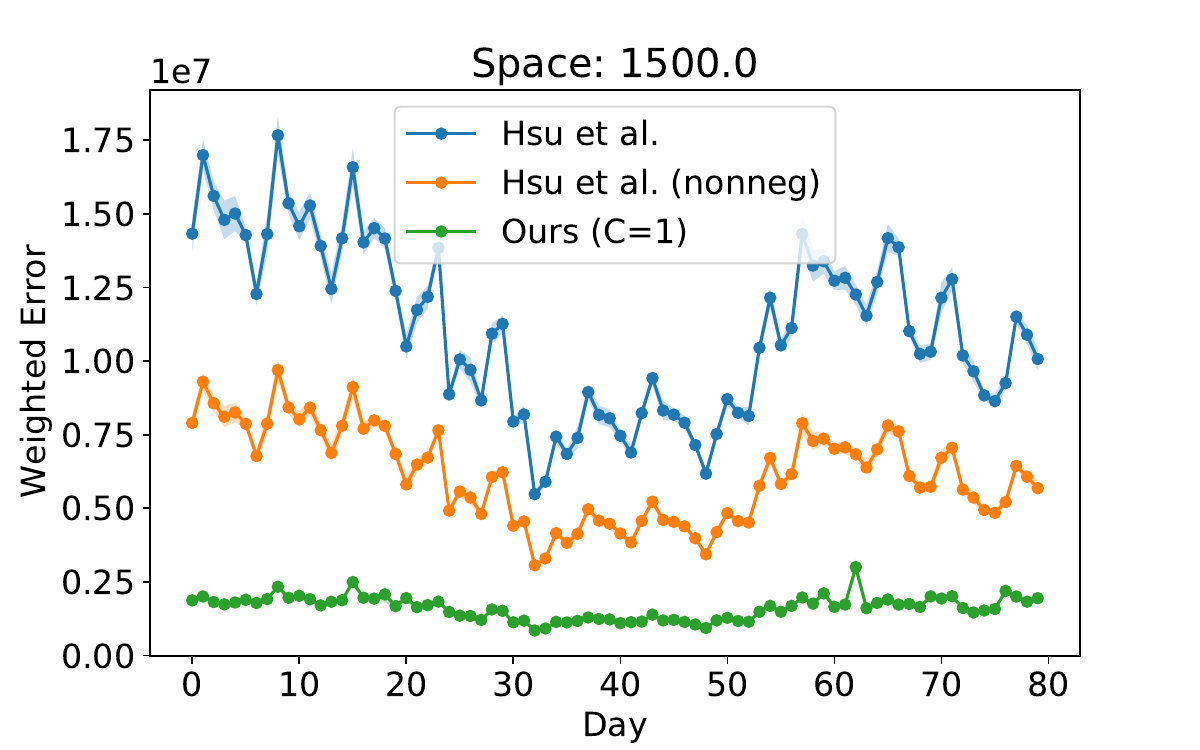}
    \includegraphics[width=0.49\textwidth]{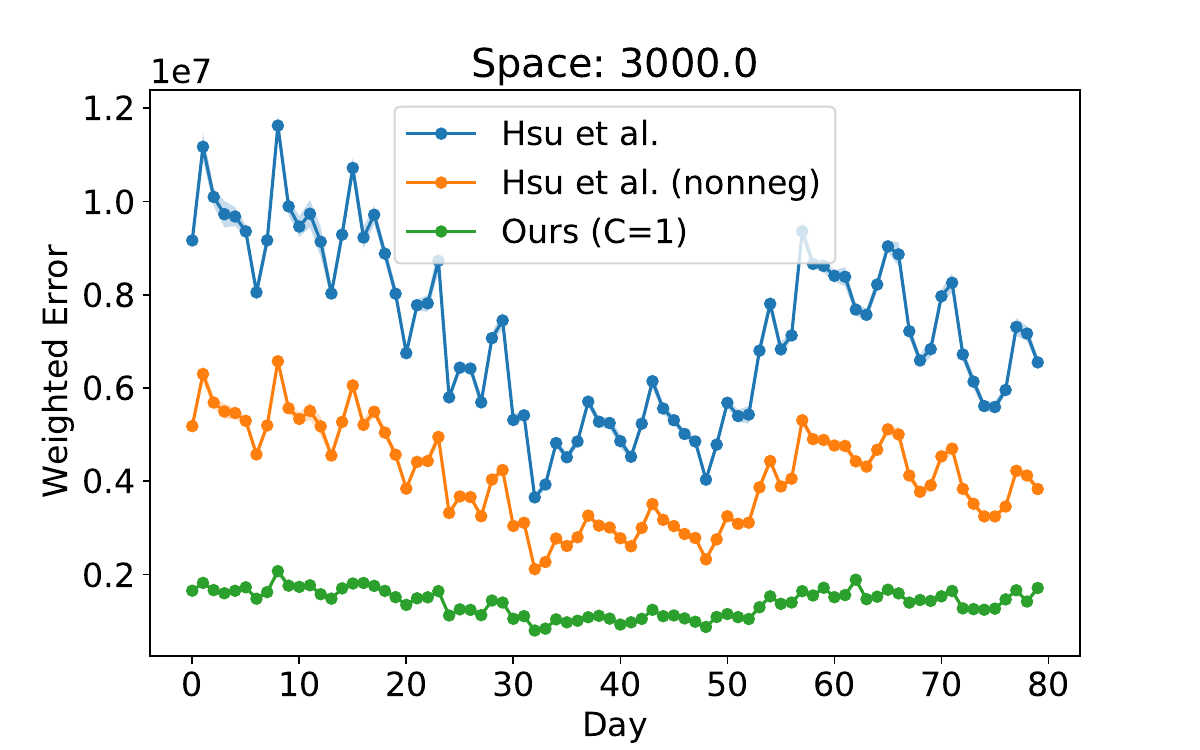}
    \caption{Comparison of weighted errors with predictions on the AOL dataset}
\end{figure}

\begin{figure}
    \centering
    \includegraphics[width=0.49\textwidth]{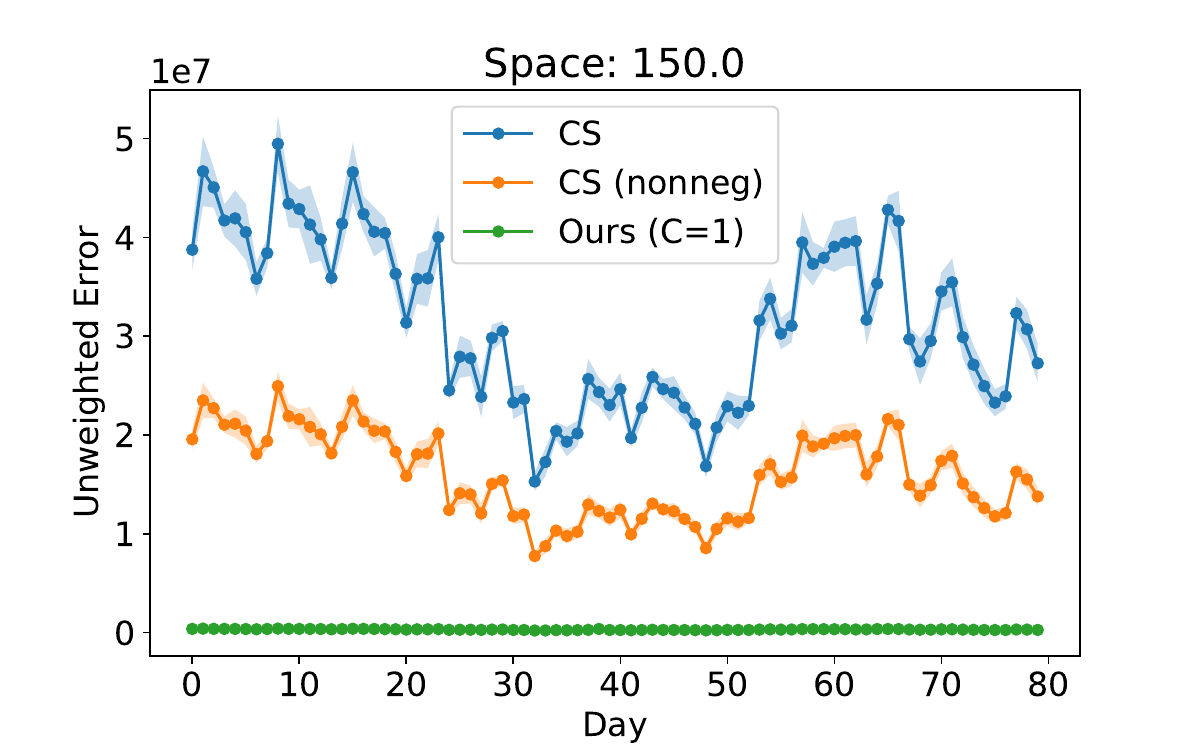}
    \includegraphics[width=0.49\textwidth]{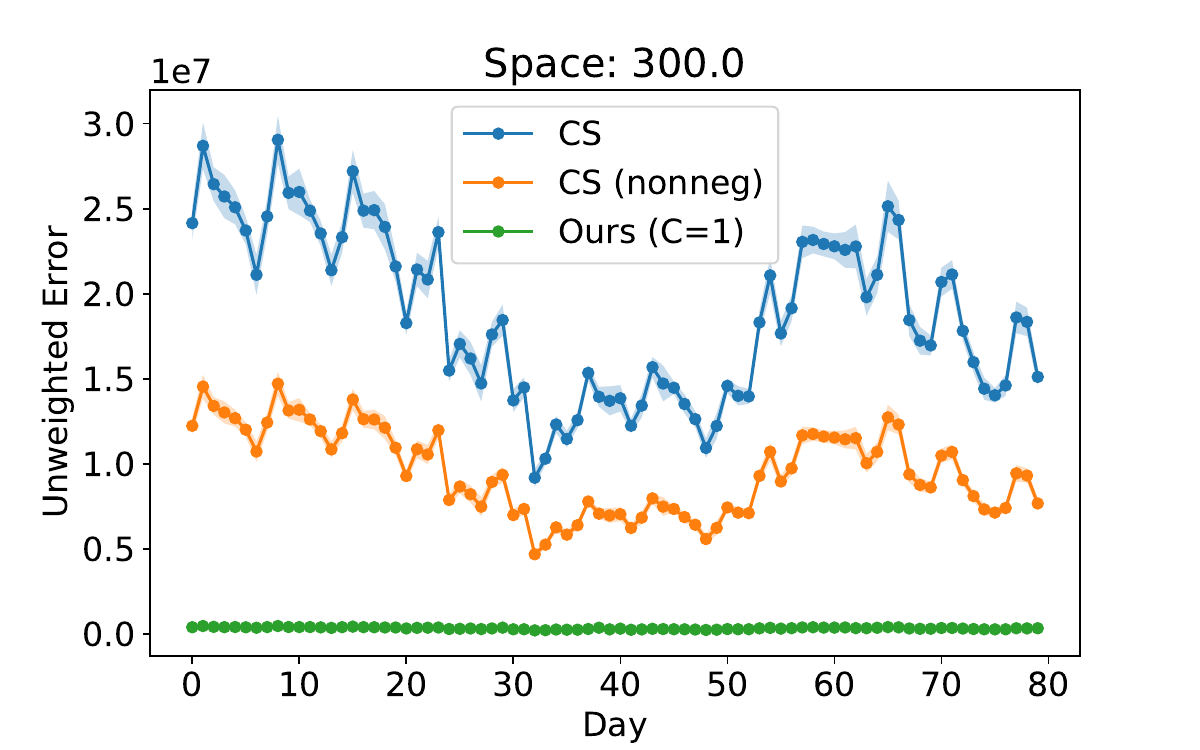}
    \includegraphics[width=0.49\textwidth]{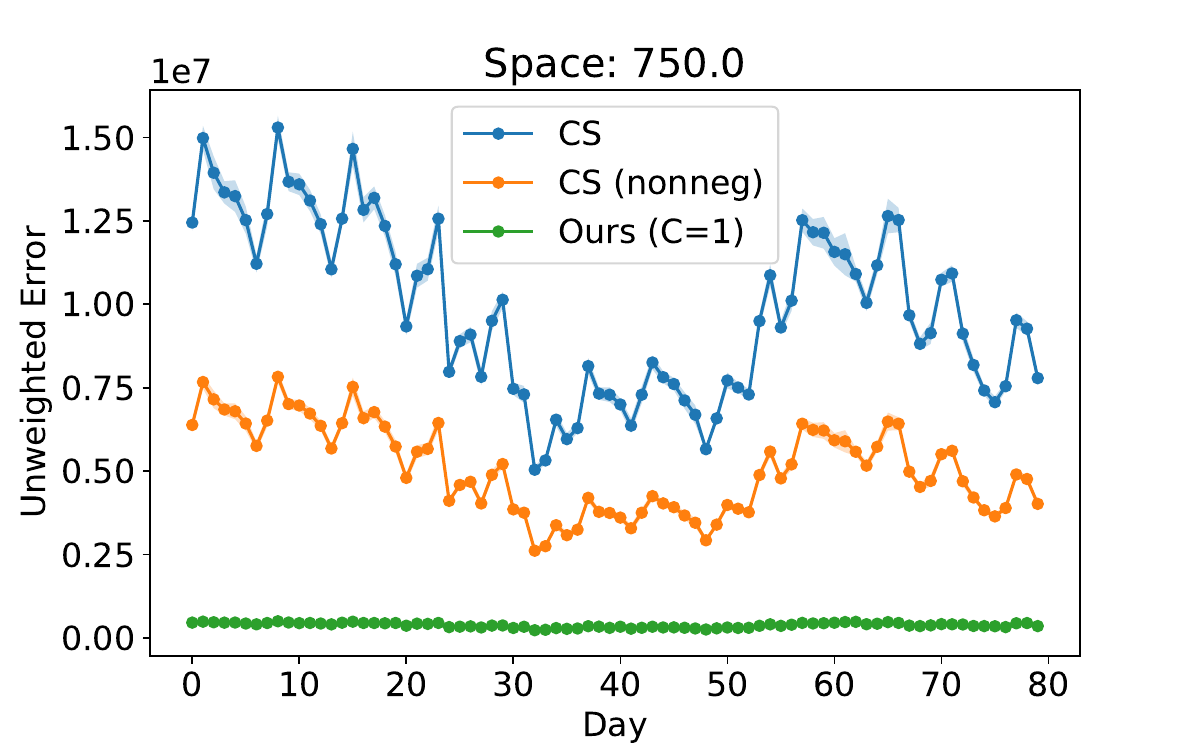}
    \includegraphics[width=0.49\textwidth]{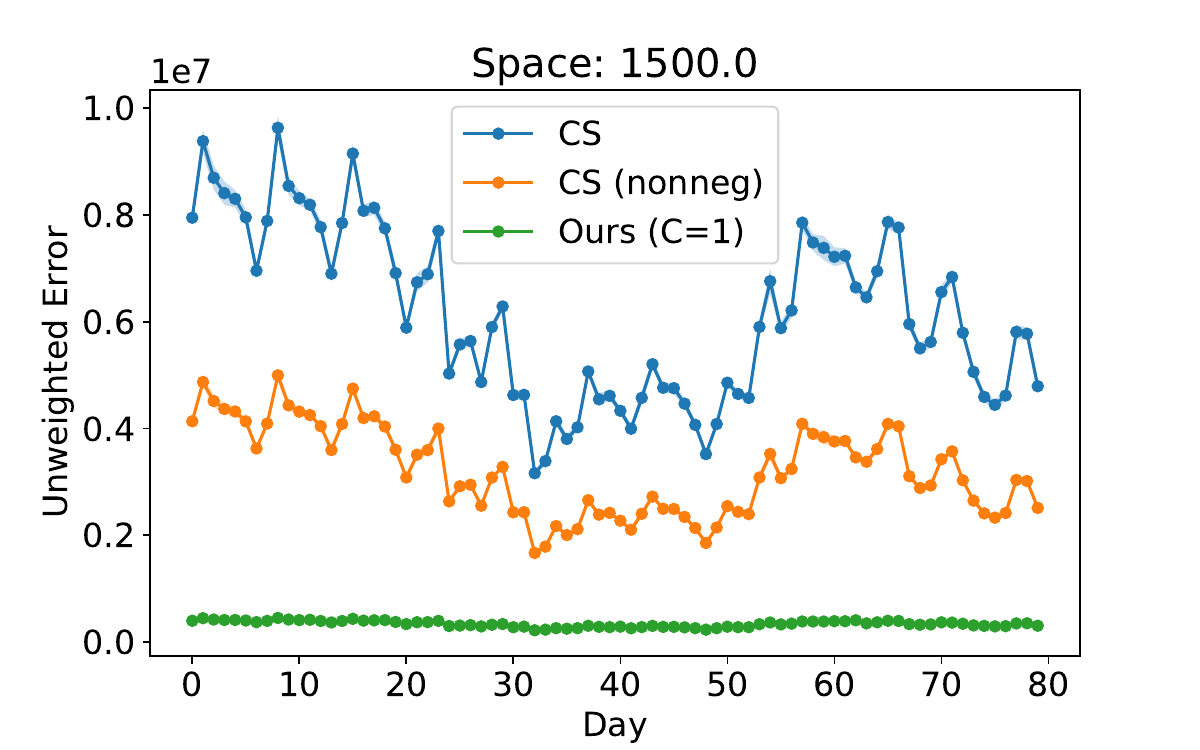}
    \includegraphics[width=0.49\textwidth]{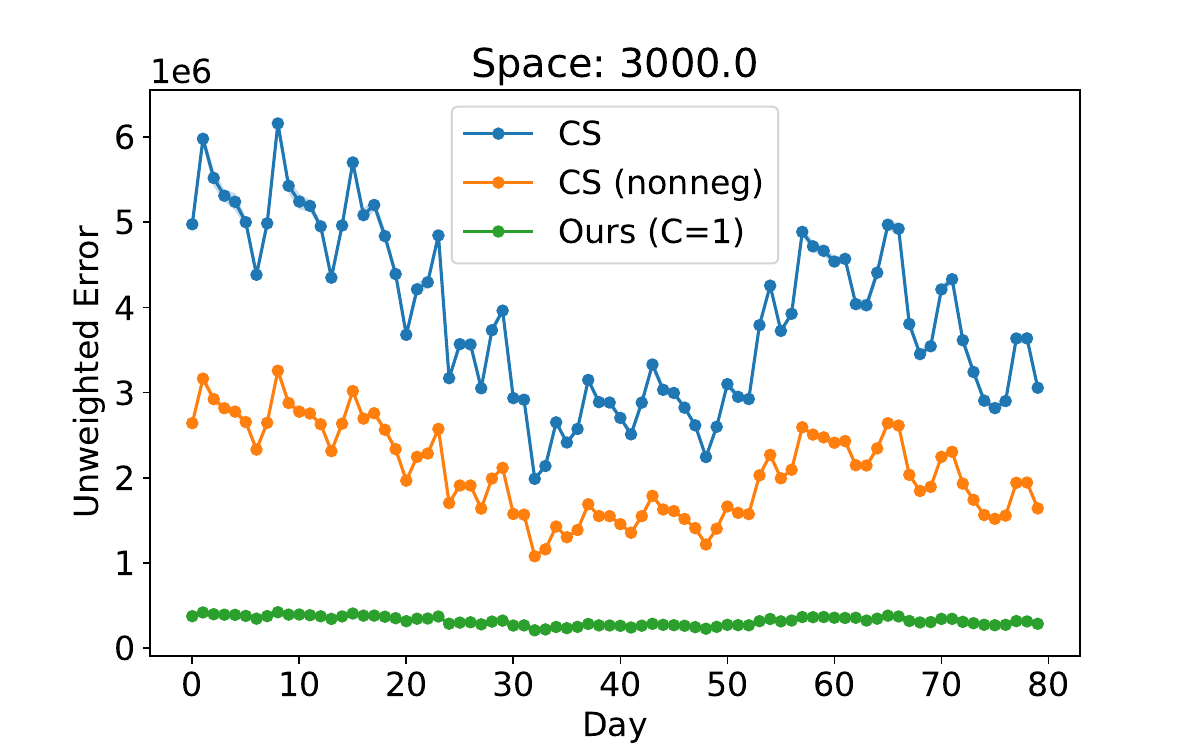}
    \caption{Comparison of unweighted errors without predictions on the AOL dataset}
\end{figure}

\begin{figure}
    \centering
    \includegraphics[width=0.49\textwidth]{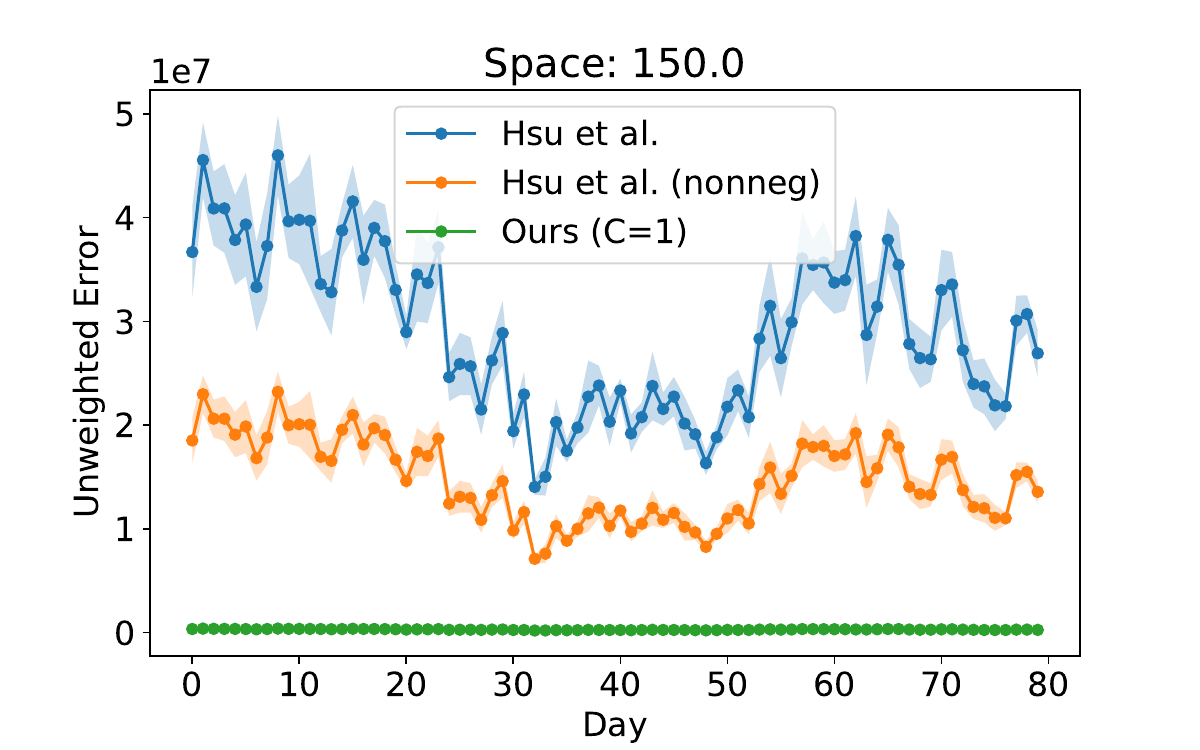}
    \includegraphics[width=0.49\textwidth]{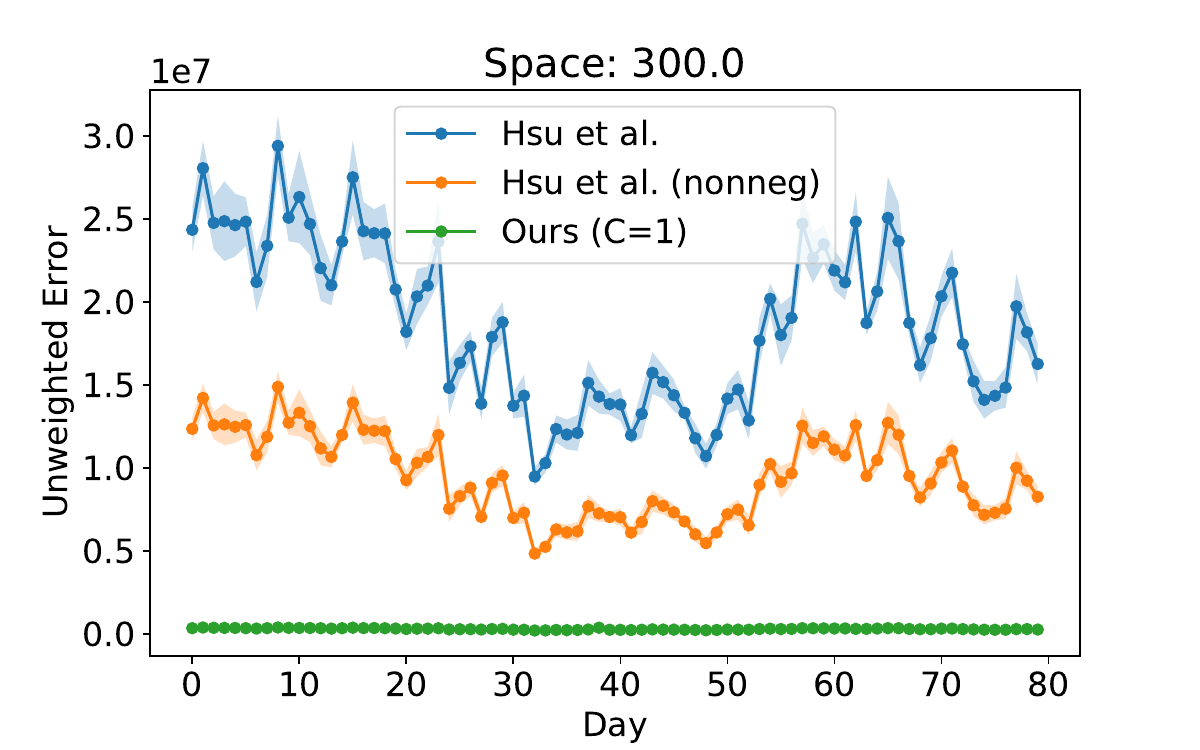}
    \includegraphics[width=0.49\textwidth]{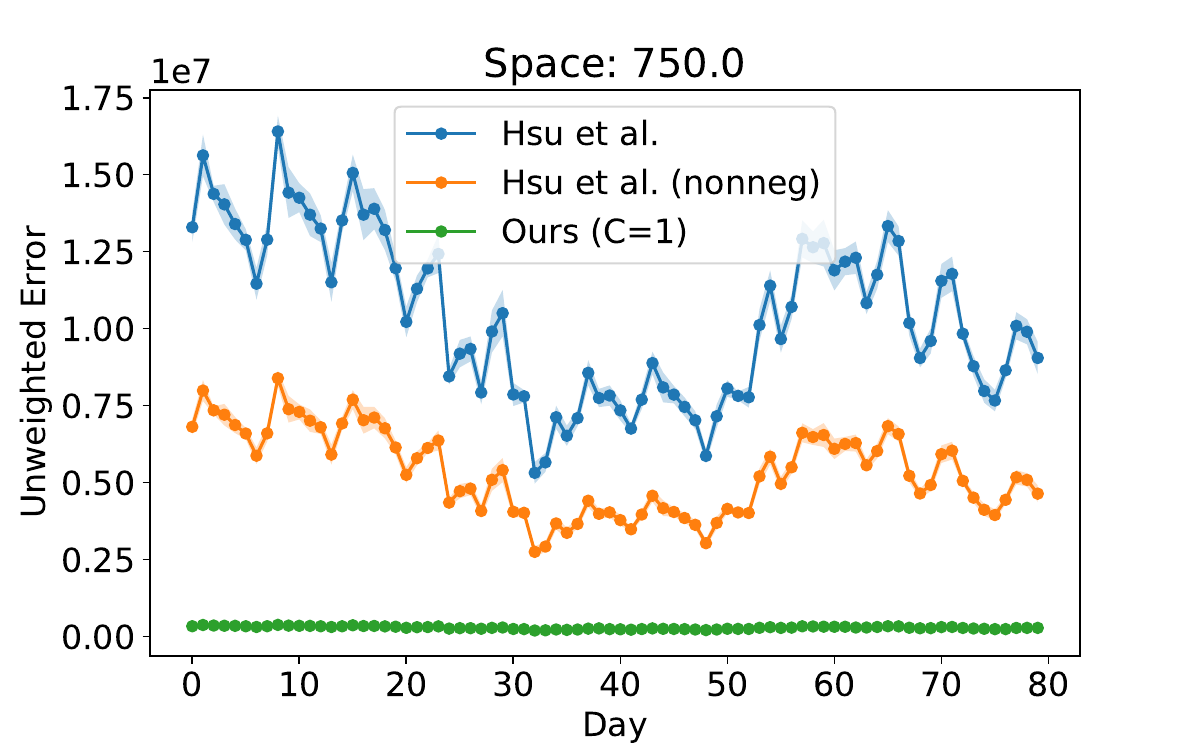}
    \includegraphics[width=0.49\textwidth]{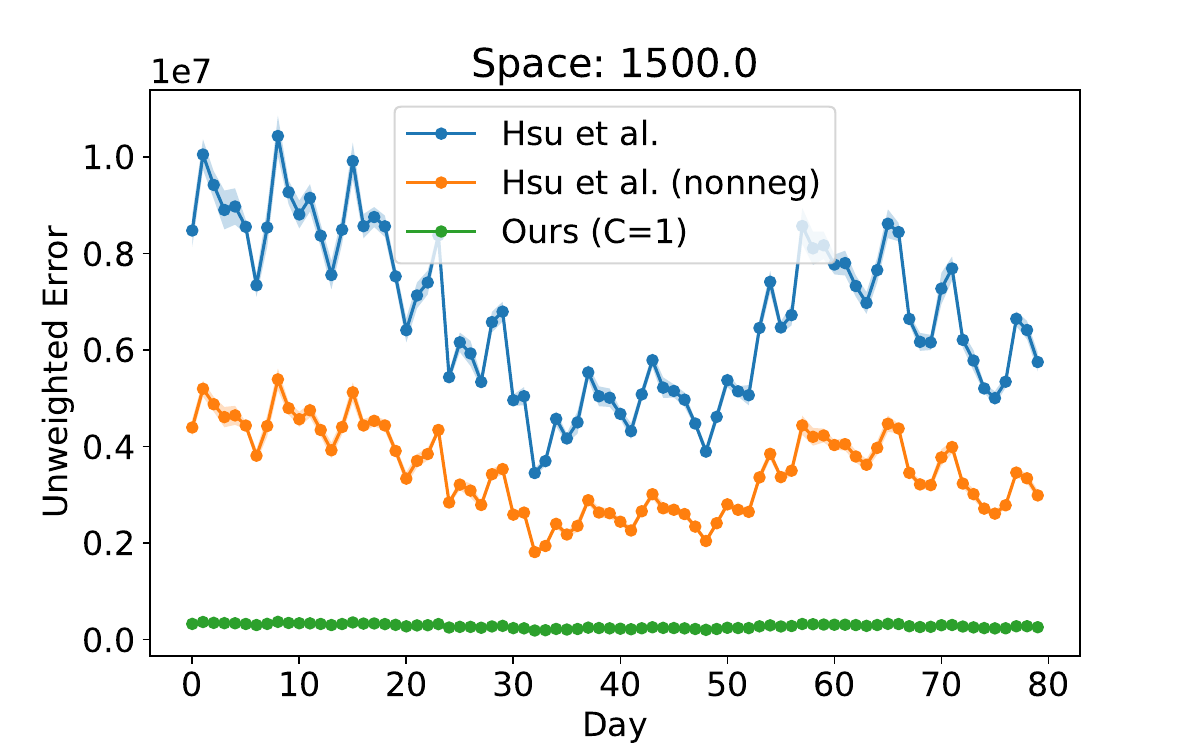}
    \includegraphics[width=0.49\textwidth]{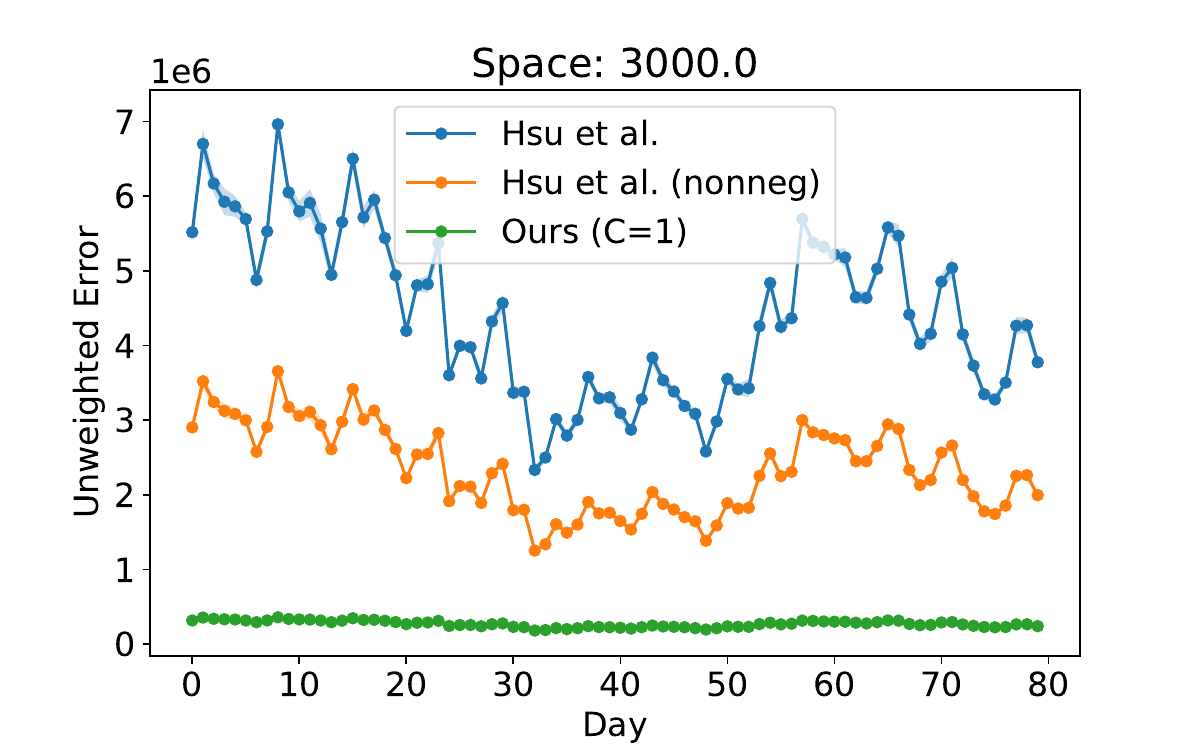}
    \caption{Comparison of unweighted errors with predictions on the AOL dataset}
\end{figure}

\end{document}